\documentclass[12pt]{article}%
\usepackage{enumitem}
\usepackage{adjustbox}
\usepackage{bm}
\usepackage{amsfonts}
\usepackage[colorlinks,citecolor=blue,urlcolor=blue,hyperfootnotes=false]{hyperref}
\usepackage{amssymb}
\usepackage[bottom]{footmisc}
\usepackage[hyperref]{ntheorem}
\usepackage{graphicx}
\usepackage{float}
\usepackage[figuresright]{rotating}
\usepackage{natbib}  
\usepackage[doublespacing,nodisplayskipstretch]{setspace}
\usepackage[margin=1in, a4paper]{geometry}
\usepackage{booktabs}
\usepackage{threeparttable}
\usepackage{amsmath}
\usepackage{lscape}
\usepackage{scrextend}
\usepackage{relsize}
\usepackage{multicol}
\usepackage{chngcntr}
\usepackage{etoolbox}
\usepackage{amssymb}
\usepackage{tabularx}
\usepackage[english]{babel}
\usepackage{multirow}
\usepackage{booktabs}
\usepackage[labelfont=bf]{caption}
\usepackage{float}
\usepackage{titlesec}
\usepackage{array}
\usepackage{color}
\usepackage{afterpage}
\usepackage{framed}
\usepackage{float}
\usepackage{afterpage}%
\usepackage{appendix}
\usepackage{subcaption}
\usepackage{mathrsfs}
\providecommand{\U}[1]{\protect\rule{.1in}{.1in}}
\allowdisplaybreaks
\renewcommand\arraystretch{1.15}
\theoremstyle{plain}

\newtheorem{theorem}{Theorem}[section]
\newtheorem{Assumption}{Assumption}[section]

\newtheorem{corollary}{Corollary}[section]

\newtheorem{lemma}{Lemma}[section]
\newtheorem{remark}{Remark}[section]

\newenvironment{proof}[1][Proof]{\noindent\textbf{#1.} }{\ \rule{0.5em}{0.5em}}

\numberwithin{equation}{section}
\numberwithin{table}{section}
\numberwithin{figure}{section}

\theorembodyfont{\normalfont}

\deffootnote{1em}{1.6em}{\thefootnotemark \enskip}

\numberwithin{equation}{section}
\numberwithin{Assumption}{section}

\definecolor{lightgray}{gray}{0.9}

\begin{document}

\title{A Projection Approach to Nonparametric Significance and Conditional Independence Testing\footnote{The authors contributed equally to this work and are listed in alphabetical order.}}
%\title{Tailor-made Nonparametric Significance and Conditional Independence Tests\footnote{The authors contributed equally to this work and are listed in alphabetical order.}}
\author{Xiaojun Song\thanks{Corresponding author: Department of Business Statistics and Econometrics, Guanghua School of Management, Peking University. Email: \texttt{sxj@gsm.pku.edu.cn}. This work was supported by the National Natural Science Foundation of China [Grant Numbers 72373007 and 72333001]. The author also gratefully acknowledges the research support from the Center for Statistical Science of Peking University, China, and the Key Laboratory of Mathematical Economics and Quantitative Finance (Peking University) of the Ministry of Education, China.}\\
	\and Jichao Yuan\thanks{Department of Business Statistics and Econometrics, Guanghua School of Management, Peking University. Email: \texttt{2201111030@stu.pku.edu.cn}.}
	}

\maketitle
\begin{abstract}
This paper develops a novel nonparametric significance test based on a tailored nonparametric-type projected weighting function that exhibits appealing theoretical and numerical properties. We derive the asymptotic properties of the proposed test and show that it can detect local alternatives at the parametric rate. Using the nonparametric orthogonal projection, we construct a computationally convenient multiplier bootstrap to obtain critical values from the case-dependent asymptotic null distribution. Compared with the existing literature, our approach overcomes the need for a stronger compact support assumption on the density of covariates arising from random denominators. We also extend the tailor-made projection procedure to test the conditional independence assumption. The simulation experiments further illustrate the advantages of our proposed method in testing significance and conditional independence in finite samples.
\end{abstract}

\newpage

\doublespacing

\section{Introduction}
Testing the significance of a subset of explanatory variables in a regression model is a fundamental problem in statistics and econometrics and is essential for variable selection and model specification. Given the risks of inconsistent estimation and misleading inference associated with misspecified parametric models, the development of tests within a nonparametric framework has attracted considerable attention. Theories have been established for cross-sectional studies based on independent and identically distributed (i.i.d.) data, with early contributions by \cite{lewbel1995consistent}, \cite{fan1996consistent}, \cite{lavergne2000nonparametric}, and \cite{delgado2001significance}, and more recent work such as \cite{zhu2018dimension} and \cite{lundborg2024projected}. Further developments address significance tests for more complex data structures, such as time series settings, categorical regressors, high-dimensional models, spatial point patterns, and network data, with related work including \cite{chen1999consistent}, \cite{li1999consistent}, \cite{racine2006testing}, \cite{zhu2018significance}, \cite{kojaku2018generalised}, and \cite{dvovrak2024nonparametric}. In addition, modern machine learning methods have motivated new significance testing problems; see \cite{horel2020significance}, \cite{wu2021generalized}, and \cite{giesecke2025aico}.

Generally, the literature on nonparametric significance testing can be categorized into two mainstreams: local smoothing-based and global nonsmoothing-based methods. The local smoothing-based approach, exemplified by \cite{fan1996consistent} and \cite{lavergne2000nonparametric}, typically constructs test statistics using kernel-based squared distance measures. While these tests are consistent against omnibus alternatives, they often suffer from the \textquotedblleft curse of dimensionality\textquotedblright{}, where the local power degrades rapidly as the dimension of the covariates increases. Although the test of \cite{zhu2018dimension} addresses the problem via a dimension-reduction method, nonparametric estimation remains necessary, yielding a nonparametric convergence rate and making it difficult to fully resolve the issue. %An alternative approach is the global nonsmoothing-based approach, which originates from the integrated conditional moment (ICM) principle of \cite{bierens1982consistent} and was further developed by \cite{bierens1990consistent}, \cite{bierens1997asymptotic}, \cite{stute1997nonparametric}, and \cite{delgado2001significance}. The key idea is to transform the nonparametric-type conditional moment restrictions under the null into an infinite number of unconditional moment restrictions. These tests typically achieve the fastest possible parametric rate and are less sensitive to bandwidth choices. However, they often face challenges in characterizing the asymptotic null behavior of the empirical process when nuisance functions are estimated, particularly when the nonparametric null model estimator used in significance tests takes a particular form.
An alternative approach is the global nonsmoothing-based approach adopted by \cite{delgado2001significance}, which is based on the integrated conditional moment (ICM) principle of \cite{bierens1982consistent} and \cite{bierens1990consistent}, originally developed for specification tests of parametric regressions. The ICM idea in the context of nonparametric significance testing is to transform the nonparametric-type conditional moment restrictions under the null into an infinite number of unconditional moment restrictions. The tests proposed by \cite{delgado2001significance} achieve the fastest possible parametric rate and are less sensitive to bandwidth choices. However, these tests still face challenges in effectively leveraging the asymptotic null behavior of the $U$-process when nuisance functions of the null model are estimated.  %especially in the particular form of the nonparametric estimator of the null model considered in significance tests. %characterizing

More specifically, the nonparametric estimation of the conditional mean function inevitably entails a random denominator problem, a difficulty that has been emphasized in \cite{lundborg2024projected} as an essential barrier to the broader use of more flexible testing procedures. A standard way to control this problem is to construct density–weighted residuals so that the estimated unconditional moment restrictions can be written as a nondegenerate $U$–process, at the cost of introducing an additional \textquotedblleft nonparametric estimation effect\textquotedblright{}. This additional effect creates a substantial technical hurdle when attempting to obtain critical values via computationally efficient bootstrap schemes. To address this difficulty, previous work, for example, \cite{delgado2001significance}, modifies the multiplier bootstrap version of the $U$-process used in their procedure and adopts more elaborate setups. These approaches typically rely on technically demanding trimming schemes or impose stronger assumptions on the density of the covariates under the null, %$X$, 
for instance, by requiring it to be bounded away from zero. This type of restriction excludes many commonly used distributions, such as the normal and the Student's $t$.

The primary contributions of this paper are to address the issues discussed above by introducing a nonparametric projection and to derive multiplier bootstrap critical values in a computationally efficient and interpretable manner. To be precise, we project the weighting function used in constructing the integrated conditional moments onto the orthocomplement of the space spanned by the density evaluated at all covariates under consideration, including both those already known to be significant and those whose significance is being tested. It is worth noting that, although, to our knowledge, the \textquotedblleft nonparametric estimation effect\textquotedblright{} has not been emphasized in the existing literature, an analogous phenomenon has been recognized earlier in testing parametric models; see \cite{durbin1973distribution}. Moreover, the novel projection approach that we develop to address this effect is motivated by the parametric projection approaches proposed in \cite{sant2019specification}, \cite{yang2024model}, and \cite{song2025unified}. Thanks to the projection, the critical values for our statistics can be simulated via the computationally fast multiplier bootstrap procedure rather than the computationally intensive wild bootstrap, as in much of the nonparametric significance testing literature, for example, \cite{gu2007bootstrap}. % Journal of Nonparametric Statistics
Finally, given that testing the conditional independence assumption is closely related to testing significance in the mean, we also extend the tailor-made projection procedure to test this important assumption. 

The rest of the paper is organized as follows. Section \ref{sec.Test} outlines the testing procedure, including the nonparametric projection-based methodology and the construction of our statistics. The asymptotic properties with some reasonable assumptions of the test under the null, local alternatives, and global alternatives are established in Section \ref{sec.Asy}. We further detail the implementation of the multiplier bootstrap procedure in Section \ref{sec.boot} and extend the above projection test and the multiplier bootstrap procedure to testing the conditional independence assumption in Section \ref{sec.CI}. The finite-sample performance of the proposed test is investigated via a set of Monte Carlo simulations in Section \ref{sec.Simulation}. Finally, Section \ref{sec.Conclusion} concludes the paper. Additional simulation results and detailed proofs of the theoretical results are provided in the online supplement.

\section{Testing procedure}\label{sec.Test}
Let $(S,\mathcal{F},\mathbb{P})$ be the probability space of the random vector $\chi = (Y,W^\top)^\top$, where $Y$ is a scalar response variable and $W=(X^\top,Z^\top)^\top$ is the vector of explanatory variables, with $X$ being $\mathbb{R}^q$-valued and $Z$ being $\mathbb{R}^p$-valued. Henceforth, $A^\top$ denotes the matrix transpose of $A$. We consider testing whether the vector of covariates $Z$ has no additional explanatory power for the conditional mean of the dependent variable $Y$, given that $X$ has a statistically significant effect. The null hypothesis of interest is 
\begin{align}\label{hyp.null}
    H_0:\,\mathbb{E}\left[Y\vert W\right] = \mathbb{E}\left[Y\vert X\right]\quad  a.s.,
\end{align}
and the alternative hypothesis, $H_1$, is the negation of $H_0$. Let $f_X(\cdot)$ denote the probability density function (PDF) of $X$. Using the fact that the PDF $f_X(X)>0$ $a.s.$, we can rewrite $H_0$ in \eqref{hyp.null} as 
\begin{align}\label{hyp.null f}
    H_0:\, f_X(X)\mathbb{E}\left[\epsilon\vert W\right]=0\quad a.s.,
\end{align}
where $\epsilon=Y-\mathbb{E}[Y\vert X]:=Y-m(X)$ is the nonparametric error, with $m(X)$ denoting the unknown conditional mean function under the null. 

It is well established in the testing literature that the conditional moment restriction in \eqref{hyp.null f} is equivalent to a continuum of unconditional moment restrictions. More specifically, following the ICM principle introduced by \cite{bierens1982consistent}, and further developed by \cite{stute1997nonparametric}, \cite{stute1998bootstrap}, \cite{delgado2001significance}, \cite{stute2002model}, and \cite{escanciano2006consistent}, we employ the indicator function $1_w(W):=1(W\leq w)=1(X\leq x)1(Z\leq z)$ with $w=(x^\top,z^\top)^\top\in\mathbb R^{q+p}$ as the weighting scheme to transform the conditional moment restriction to the unconditional version. Consequently, the null hypothesis $H_0$ in \eqref{hyp.null f} can be equivalently characterized by the following infinite number of unconditional moment conditions indexed by $w$:
\begin{align}\label{hyp.null eq}
    H_0:\,\mathbb{E}\left[\epsilon f_X(X)1_w(W)\right]=0\quad \text{for all }w\in\mathbb{R}^{q+p},
\end{align}
which converts the original nonparametric significance testing problem into a global testing problem. Therefore, testing \eqref{hyp.null eq} achieves the fastest possible parametric rate.

Suppose that we have a random sample $\{(Y_i,W_i^\top)^\top\}_{i=1}^n$ of size $n\geq 1$ consisting of independent and identically distributed (i.i.d.) random variables, the natural sample analog for \eqref{hyp.null eq} is given by
\begin{align*}
    \hat T_n(w) =& \frac{1}{n}\sum_{i=1}^n\hat{\epsilon}_i\hat{f}_X(X_i)1_w(W_i)\\
    =&\frac{1}{n(n-1)}\sum_{i=1}^n\sum_{j=1,j\neq i}^n\frac{1}{a^q}K\left(\frac{X_i-X_j}{a}\right)(Y_i-Y_j)1_w(W_i),
\end{align*}
where $\hat{\epsilon}_i=Y_i-\hat{m}(X_i)$ is the nonparametric residual, with the leave-one-out nonparametric estimators for $m(X_i)$ and $f_X(X_i)$ given by
\begin{align*}
    \hat{m}(X_i) = \frac{1}{\hat{f}_X(X_i)}\frac{1}{(n-1)a^q}\sum_{j=1,j\neq i}^nK\left(\frac{X_i-X_j}{a}\right)Y_j
\end{align*}
and
\begin{align*}
    \hat{f}_{X}(X_i) = \frac{1}{(n-1)a^q}\sum_{j=1,j\neq i}^nK\left(\frac{X_i-X_j}{a}\right),
\end{align*}
respectively. Here, $a=a(n)\in\mathbb{R}^{+}$ is a bandwidth parameter shrinking to zero at a suitable rate as $n\to\infty$ and $K(u)=\Pi_{j=1}^qk(u^{(j)})$ is a product kernel for a $q$-dimensional vector $u$, where $k(\cdot)$ is the univariate kernel and $u^{(j)}$ denotes the $j$-th component of $u$. Note that the random denominator problem is effectively eliminated by incorporating the estimated density $\hat{f}_X(X_i)$ into the $U$-process $\hat T_n(\cdot)$, significantly facilitating the theoretical analysis of $\hat T_n(\cdot)$. 

By imposing the standard regularity conditions considered in \cite{delgado2001significance} (see their Assumptions A1--A5), under the null $H_0$, we have
\begin{align}\label{stat.delgado}
    \sup_{w}\left\vert \hat T_n(w)-\frac{1}{n}\sum_{i=1}^n\epsilon_if_X(X_i)1_x(X_i)\left[1_z(Z_i)-F_{Z\vert X}(z\vert X_i)\right]\right\vert = o_p\left(n^{-1/2}\right),
\end{align}
where $F_{Z\vert X}(z\vert x)$ is the conditional distribution function of $Z$ given $X$. It is observed that in the asymptotically uniform expansion \eqref{stat.delgado}, the term 
\begin{align*}
    \frac{1}{n}\sum_{i=1}^n\epsilon_if_X(X_i)1_x(X_i)1_z(Z_i)
\end{align*}
is the infeasible process for testing the null hypothesis in \eqref{hyp.null eq} if $\epsilon_i$ is known (equivalently if $m(X_i)$ is known), and the term  
\begin{align*}
    \frac{1}{n}\sum_{i=1}^n\epsilon_if_X(X_i)1_x(X_i)F_{Z\vert X}(z\vert X_i)
\end{align*}
can be regarded as the ``nonparametric estimation effect'' due to using the nonparametric residual $\hat{\epsilon}_i$ to replace $\epsilon_i$ (via the estimation of $m(X_i)$). 

Much less attention has been paid to studying the above term, unlike the ``parametric estimation effect'' also known as the ``Durbin problem'' mentioned in \cite{durbin1973distribution}. Specifically, since the asymptotic null distributions of the test statistics (constructed as continuous functionals of $\hat T_n(\cdot)$) are typically case-dependent, the use of bootstrap methods to obtain critical values is unavoidable. However, the ``nonparametric estimation effect'' in the asymptotically uniform expansion of $\hat T_n(\cdot)$ poses a substantial challenge to the direct application of the computationally efficient multiplier bootstrap. While the wild bootstrap is a potential alternative, it incurs significantly higher computational complexity. To address this, with the help of \eqref{stat.delgado}, \cite{delgado2001significance} propose a specifically constructed multiplier bootstrap-based $U$-process as follows:
\begin{align*}
    \frac{1}{n}\sum_{i=1}^nV_i\hat{\epsilon}_i\hat{f}_X(X_i)1_x(X_i)\left[1_z(Z_i)-\hat{F}_{Z\vert X}(z\vert X_i)\right],
\end{align*}
with $\{V_i\}_{i=1}^n$ being i.i.d. random variables (i.e., multipliers) that have mean zero, unit variance, and are independent of the original sample $\{(Y_i,W_i^\top)^\top\}_{i=1}^n$. However, this approach necessitates estimating the conditional distribution function $F_{Z\vert X}(z\vert X_i)$ nonparametrically within the multiplier bootstrap procedure, which is
\begin{align*}
    \frac{1}{\hat{f}_X(X_i)}\frac{1}{(n-1)a^q}\sum_{j=1,j\neq i}^nK\left(\frac{X_i-X_j}{a}\right)1_z(Z_j).
\end{align*}
See Section $3$ of \citet{delgado2001significance} for a detailed description of this implementation. Unfortunately, this inevitably reintroduces the random denominator issue in the multiplier bootstrap, thereby undermining the purpose of using $\hat T_n(w)$. %, which was intended to be avoided in the initial construction of $\hat T_n(w)$. %by using $\hat f_X(X_i)$. 
Indeed, the solution adopted by \citet{delgado2001significance} to ensure the validity of this bootstrap is to impose stronger assumptions, such as requiring the density $f_X(X)$ to be bounded away from zero (see their Assumptions A9), which was exactly intended to be avoided in the initial construction of $\hat T_n(w)$ through the density weighting $\hat f_X(X_i)$. Such a restrictive bounded support assumption constitutes a significant barrier to applying the test to data with unbounded support, including the commonly used normal distribution. Another potential solution is to use trimming; however, the prohibitive complexity of the required theoretical proofs and the practical choice of the trimming level have prevented its successful implementation.

In this paper, to eliminate the nonparametric estimation effect, we propose a tailor-made nonparametric significance test that depends on a novel nonparametric-type projection of the weighting function $1_z(Z_i)$ onto the orthocomplement of the space spanned by $f_W(W_i)$ (hereafter, the PDF of $W$), where orthogonality is understood conditionally on $X_i$. To this end, we first consider the following infeasible projected weighting function:
\begin{align*}
    \mathcal{P}1_z(Z_i) :=& 1_z(Z_i)-f_{Z\vert X}(Z_i\vert X_i)\left[\int_{-\infty}^{\infty}f_{Z\vert X}^2(\bar z\vert X_i)\,d\bar z\right]^{-1}\int_{-\infty}^zf_{Z\vert X}(\bar z\vert X_i)\,d\bar z\notag\\
    =&1_z(Z_i)-f_W(W_i)\Delta^{-1}(X_i)G(z;X_i),
\end{align*}
where $f_{Z\vert X}(z\vert x)$ is the conditional density function of $Z$ given $X$, 
\begin{align*}
    \Delta(X_i) = \int_{-\infty}^{\infty}f_{W}^2(X_i,\bar z)\,d\bar z, \quad \text{and}\quad G(z;X_i) = \int_{-\infty}^zf_{W}(X_i,\bar z)\,d\bar z.
\end{align*}
In fact, the projection constructed above admits an intuitive interpretation as the linear regression of $1_z(Z_i)$ on $f_W(W_i)$ conditional on $X_i$, where $\Delta^{-1}(X_i)G(z;X_i)$ represents the coefficients of the linear projection. In other words, the term $f_W(W_i)\Delta^{-1}(X_i)G(z;X_i)$ constitutes the best linear predictor of $1_z(Z_i)$ given $f_W(W_i)$. Consequently, our constructed weighting function $\mathcal{P}1_z(Z_i)$ corresponds precisely to the associated prediction error, which implies that $\mathcal{P}1_z(Z_i)$ is orthogonal to $f_W(W_i)$ conditional on $X_i$, i.e.,
\begin{align*}
    \mathbb{E}\left[f_W(W_i)\mathcal{P}1_z(Z_i)\vert X_i\right] = f_X^{-1}(X_i)\left[G(z;X_i) - \Delta(X_i)\Delta^{-1}(X_i)G(z;X_i)\right]\equiv 0 \,\,\,a.s.
\end{align*}

\begin{remark}
While this paper introduces a nonparametric orthogonal projection approach, analogous parametric orthogonal projection-based techniques are well established in the testing literature, primarily for assessing the correct specification of parametric models. Examples include \cite{escanciano2014specification}, \cite{sant2019specification}, \cite{yang2024model}, and \cite{song2025unified}, which utilize orthogonal projected weighting functions onto the score functions to handle the ``parametric estimation effect'' in testing the (partially) linear quantile regression processes, propensity score models, and partially linear spatial autoregressive models. 
\end{remark}

Building on the sample version of the nonparametric projection $\mathcal{P}1_z(Z_i)$, our proposed projected $U$-process %residual–marked empirical process 
is constructed as follows:
\begin{align}\label{stat.stat}
    \hat R_n(w) = \frac{1}{n}\sum_{i=1}^n\hat{\epsilon}_i\hat{f}_X(X_i)1_x(X_i)\hat{\mathcal{P}}_n1_z(Z_i).
\end{align}
where
\begin{align*}
    \hat{\mathcal{P}}_n1_z(Z_i) = 1_z(Z_i)-\hat{f}_{W}(W_i)\hat{\Delta}_n^{-1}(X_i)\hat{G}_n(z;X_i)
\end{align*}
is the natural estimator for $\mathcal{P}1_z(Z_i)$. Here,
\begin{align*}
    \hat{\Delta}_n(X_i) = \int_{-\infty}^{\infty}\hat{f}^2_W(X_i,\bar z)\,d\bar z \quad \text{and}\quad \hat{G}_n(z;X_i) = \int_{-\infty}^z\hat{f}_W(X_i,\bar z)\,d\bar z
\end{align*}
are the sample versions of $\Delta(X_i)$ and $G(z;X_i)$, respectively, with 
\begin{align*}
    \hat{f}_W(X_i,z) = \frac{1}{(n-1)a^qb^{p}}\sum_{j=1,j\neq i}^nK\left(\frac{X_i-X_j}{a}\right)L\left(\frac{z-Z_j}{b}\right)
\end{align*}
being the leave-one-out kernel density estimator for the PDF $f_W(X_i,z)$, where $L(\cdot)$ and $b=b(n)\in\mathbb R^+$ are the kernel and bandwidth, respectively. 

\begin{remark}
    %It is worth noting that we construct the kernel $L(\cdot)$ for the $p$-dimensional vector $Z_i$ as a product of univariate kernels and introduce a second bandwidth sequence, say $b=b(n)$, that shrinks to zero at a suitable rate as $n\to\infty$. 
    We construct the kernel $L(\cdot)$ for the $p$-dimensional vector $Z$ as a product of the univariate kernel $l(\cdot)$ and introduce a second bandwidth $b$ that shrinks to zero at a suitable rate as $n\to\infty$. Such a setting is standard in the literature on smoothing-based methods, where the additional bandwidth $b$ introduces flexibility in characterizing bias and is typically used to control the relative rate with respect to $a$; see, for example, \cite{fan1996consistent} and \cite{lavergne2000nonparametric}. In our setting, however, the restrictions on the bandwidths $a$ and $b$ are less stringent, requiring only that they satisfy Assumption \ref{ass.bandwidth} in Section \ref{sec.Asy}. One may even set $a=b$ and take $K(\cdot)$ and $L(\cdot)$ to be the same kernel without affecting the theoretical results and the implementation of the proposed tests, as illustrated in the simulations reported in Section \ref{sec.Simulation}.
\end{remark}

Note that the proposed process $\hat R_n(w)$ based on the nonparametric-type orthogonal projection $\hat{\mathcal{P}}_n1_z(Z_i)$ can account for effectively the ``nonparametric estimation effect'' discussed before, although at the cost of nonparametrically estimating the density $f_W$. As such, it can be implemented directly using a convenient multiplier bootstrap without reintroducing the random denominator, in sharp contrast to the multiplier bootstrap used by \citet{delgado2001significance}. Further details are provided in Section \ref{sec.boot}. The test statistics are then constructed as continuous functionals of $\hat R_n(w)$. In this paper, we focus on the two most commonly employed choices, namely, the Cram\'{e}r--von Mises (CvM) and the Kolmogorov--Smirnov (KS) test statistics, which are explicitly given as
\begin{align*}
CvM_n = \int\left\vert \hat R_n(w)\right\vert^2F_{W_n}(dw)    \quad \text{and}\quad KS_n = \sup_{w}\left\vert \hat R_n(w)\right\vert.
\end{align*}
From a theoretical perspective, the integrating measure $F_{W_n}(\cdot)$ appearing in the $CvM_n$ is assumed to be a random measure that converges in probability to $F_W(\cdot)$, which is absolutely continuous with respect to the Lebesgue measure on $\mathbb{R}^{q+p}$. In practice, we adopt the approach suggested in \cite{escanciano2006consistent}: for the computation of $CvM_n$ we take $F_{W_n}(\cdot)$ to be the empirical distribution function of $\{W_i\}_{i=1}^n$, and for the computation of $KS_n$ we replace the theoretical supremum by its maximum taken over the sample points. Under the null hypothesis $H_0$, the test statistics $CvM_n$ and $KS_n$ are expected to take sufficiently small values, whereas relatively large realizations provide evidence against $H_0$ in favor of the alternative hypothesis $H_1$. The critical values used to assess the magnitude of the test statistics are obtained via the multiplier bootstrap procedure detailed in Section \ref{sec.boot}.

\section{Theoretical results}\label{sec.Asy}
In this section, the large-sample properties of the proposed $CvM_n$ and $KS_n$ test statistics will be considered, with particular attention given to their asymptotic null distributions, local power, and consistency. Before presenting the theoretical results, it is necessary to adopt several definitions as introduced in \cite{delgado2001significance} and to state a set of regularity conditions that hold uniformly across all cases. Let $\mathscr{K}_l$ with $l\geq1$ be the class of even functions of uniformly bounded variation $k:\mathbb{R}\to\mathbb{R}$, which satisfy
\begin{align*}
    k(u) = O\left((1+\left\vert u\right\vert^{l+1+\eta})^{-1}\right) \text{ for some }\eta>0 \text{ and }\int_\mathbb{R}u^ik(u)\,du=\delta_{i0}\text{ for }i=0,\cdots,l-1,
\end{align*}
where $\delta_{ij}$ is the Kroneker's delta. Let $\mathscr{L}_\beta^\alpha$ with $\alpha>0$ and $\beta>0$ be the class of functions $g:\mathbb{R}^q\to\mathbb{R}$, which satisfy uniformly $(b-1)$-times continuously differentiability for $b-1\leq \beta\leq b$. In addition, we impose the assumption that there exist a positive constant $\rho$ and a function $d$ with finite $\alpha$-th moments such that 
\begin{align*}
    \sup_{\left\vert v-u\right\vert<\rho}\left\vert g(v)-g(u)-Q(v,u)\right\vert/\left\vert v-u\right\vert^\beta\leq d(u)\text{ for all } u,
\end{align*}
where $Q$ is a $(b-1)$-th degree homogeneous polynomial in $v-u$ with coefficients given by the partial derivatives of $g$ at $u$ of orders up to $b-1$, all of which are assumed to have finite $\alpha$-th moments; in particular, $Q=0$ when $b=1$.
\begin{Assumption}\label{ass.sample}
    $\{(Y_i,W_i^\top)^\top\}_{i=1}^n$ are i.i.d. observations drawn from $(Y,W^\top)^\top$, which satisfies $\mathbb{E}\vert Y-m(X)\vert^{2+\delta_1+\delta_2}<\infty$ and $\mathbb{E}\vert \Delta^{-1}(X)f_X^2(X)\vert^{2+(4+2\delta_2)/\delta_1}<\infty$ for some positive constants $\delta_1$ and $\delta_2$.
\end{Assumption}
\begin{Assumption}\label{ass.structual}
    $f_X(\cdot)\in\mathscr{L}_{\lambda_1}^{\infty}$, $f_{Z\vert X}(z|\cdot)\in\mathscr{L}_{\lambda_2}^{\infty}$, and $m(\cdot)\in\mathscr{L}_\tau^2$ for some positive constants $\lambda_1$, $\lambda_2$, and $\tau$. $k(\cdot)\in\mathscr{K}_{\max\{l_1,l_2\}+t-1}$, where $l_1-1<\lambda_1\leq l_1$, $l_2-1<\lambda_2\leq l_2$, and $t-1<\tau\leq t$. In addition, we assume $\sup_{(x,z)}\vert f_{Z\vert X}(z|x)/f_X(x)\vert<\infty$.
\end{Assumption}
\begin{Assumption}\label{ass.bandwidth}
    $(na^{2q}b^{2p})^{-1}(\log n)^{2}+na^{2\min(\tau,\lambda_1)}+nb^{2\min(\tau,\lambda_2)}\to 0$ as $n\to \infty$.
\end{Assumption}
Assumption \ref{ass.sample} concerns the moments of $\epsilon=Y-m(X)$ and $\Delta^{-1}(X)f_X^2(X)$. %the conditional density function $f_{Z|X}$. 
While \cite{fan1996consistent} required the existence of the fourth moment, this condition was relaxed in \cite{delgado2001significance}. We adopt the relaxed version to guarantee the finiteness of the variances of the statistics, which is crucial for establishing the theoretical results. Moreover, the assumption concerning $\Delta^{-1}(X)f_X^2(X)$ is essential to guarantee the convergence of the proposed projection structure. In fact, it can be reformulated as the moment condition $\mathbb{E}\vert\int f_{Z\vert X}^2(\bar z|X)\,d\bar z\vert^{-2-(4+2\delta_2)/\delta_1}<\infty$, which restricts the tail properties of the conditional density function $f_{Z|X}$ and the density function $f_X$ and is satisfied by most commonly used distributions. This assumption permits an unbounded covariate $X$ and is weaker than the stringent bounded–support assumption on $X$ as imposed in \cite{delgado2001significance} (see their Assumptions A9). %and is satisfied by most commonly used distributions. %\footnote{It is interesting to note that, when $X$ and $Z$ are independent, the required moment condition $\mathbb{E}\vert\int f_{Z\vert X}^2(\bar z|X)\,d\bar z\vert^{-2-(4+2\delta_2)/\delta_1}<\infty$ reduces to $\int f^2_Z(\bar z)d\bar z>0$, %, which holds naturally. 
%which is a very mild condition.} 
%even when $Z$ follows an unbounded distribution; for example, $\int f^2_Z(\bar z)d\bar z$ is equal to ? for $Z\sim N(0,1)$, ? for $Z\sim t_\nu$, and ? for $Z\sim Cauchy $.} %, such as the uniform. 
%It is worth emphasizing that 
This theoretical refinement extends the applicability of the ICM–type tests for nonparametric moment restrictions to data generated from normal and other commonly used unbounded distributions that could not be handled under the conditions in \cite{delgado2001significance}, as confirmed by the simulation results for the unbounded covariate case in Section \ref{sec.Simulation}.  %permit an unbounded covariate $X$ and are easier to satisfy when the dependence between $X$ and $Z$ is stronger. Further numerical evidence for the unbounded case is reported in Section \ref{sec.Simulation}.

The function classes specified in Assumption \ref{ass.structual} impose the Lipschitz continuity of the %structural 
functions $f_X(\cdot)$, $f_{Z\vert X}(z|\cdot)$, and $m(\cdot)$, which in turn ensures that the bias of the constructed empirical process is asymptotically negligible and the variance is finite.\footnote{We note that the condition $\sup_{(x,z)}\vert f_{Z\vert X}(z\vert x)/f_X(x)\vert<\infty$ is sufficient but not necessary and permits an unbounded covariate $X$. This condition is assumed to control uniformly the bias term of $\hat{\Delta}_n(x)$, as shown in Lemma \ref{lemma.delta}. In fact, alternative conditions may also deliver the desired result. Let $d_{f_X}(\cdot)$ denote the Lipschitz coefficient of $f_X(\cdot)$ and $d_{f_{Z\vert X}}(z\vert \cdot)$ denote the Lipschitz coefficient of $f_{Z\vert X}(z\vert \cdot)$, in the sense of the definitions of $\mathscr{L}_{\lambda_1}^\infty$ and $\mathscr{L}_{\lambda_2}^\infty$, respectively. For example, if the conditions $\sup_{(x,z)}\vert d_{f_X}(x)f_{Z\vert X}(z\vert x)/f_X(x)\vert<\infty$, $\sup_{(x,z)}\vert d_{f_{Z\vert X}}(z\vert x)\vert<\infty$, and $\sup_{(x,z)}\vert d_{f_X}(x)d_{f_{Z\vert X}}(z\vert x)/f_X(x)\vert<\infty$ are imposed in place of $\sup_{(x,z)}\vert f_{Z\vert X}(z\vert x)/f_X(x)\vert<\infty$, the same conclusion continues to hold. %In addition, Assumptions \ref{ass.sample}--\ref{ass.structual} permit an unbounded covariate $X$ and are easier to satisfy when the dependence between $X$ and $Z$ is stronger. Further numerical evidence for the unbounded case is reported in Section \ref{sec.Simulation}.
} Note that standard densities, such as the normal distribution, and widely used linear models satisfy this assumption. In particular, the function class is restricted to be bounded when $\alpha=\infty$, for example, in the case of density functions. Moreover, the requirements on the kernel function $k(\cdot)$ strengthen the conventional high-order kernel assumptions by imposing a decay rate condition, as mentioned in \cite{robinson1988root}. In addition, the order of the kernel function depends on the moments of the %structural 
functions $f_X(\cdot)$, $f_{Z\vert X}(z|\cdot)$, and $m(\cdot)$. 

Assumption \ref{ass.bandwidth} determines the lower and upper bounds for the bandwidth sequences $a$ and $b$ and thus guarantees the convergence of the proposed $U$-process. On the one hand, the lower bounds on the bandwidths, which are consistent with those in \cite{robinson1988root} and \cite{delgado2001significance}, serve to ensure the asymptotic negligibility of the bias of the $U$-process. On the other hand, the upper bounds on the bandwidths in Assumption \ref{ass.bandwidth} provide a sufficient, though not necessarily weakest possible, condition to establish the uniform convergence of the nonparametric estimator $\hat{\Delta}_n(\cdot)$ for $\Delta(\cdot)$ as shown in Lemma \ref{lemma.delta}. The slightly strengthened assumptions on the upper bounds are used to simplify the theoretical analysis. Specifically, we expand $\hat{\Delta}_n(\cdot)$ only up to the second order, although the upper bounds could be relaxed with higher order expansions of $\hat{\Delta}_n(\cdot)$. The upper bounds cannot, however, be relaxed beyond the rate $(na^{2q}b^{p})^{-1}\to 0$, which guarantees that the variance of the test statistics remains bounded.

In the following, let ``$\Longrightarrow$'' represent weak convergence on $(l^{\infty}(\Pi),\mathcal{B}_\infty)$ in the sense of Hoffmann–-J\textup{\o{}}rgensen, where $\mathcal{B}_\infty$ denotes the corresponding Borel $\sigma$-algebra, see, e.g., Definition $1.3.3$ in \cite{van1996weak}. We establish formally the uniform decomposition of $\hat R_n(\cdot)$, which indicates that the nonparametric estimation effects arising from $\hat m$ and $\hat f_W$ are asymptotically negligible. Denote 
\begin{align*}
    \zeta(Y,X,Z;x,z)=\left[Y-m(X)\right]f_X(X)1_x(X)\mathcal{P}1_z(Z).
\end{align*}

\begin{theorem}\label{thm.null}
Suppose that Assumptions \ref{ass.sample}--\ref{ass.bandwidth} hold. Under the null hypothesis $H_0$, %we can show that  
%\begin{align}\label{thm.eq null}
%    \sup_{w}\left\vert \hat R_n(w)-\frac{1}{n}\sum_{i=1}^n\epsilon_if_X(X_i)1_x(X_i)\mathcal{P}1_z(Z_i)\right\vert=o_p\left(n^{-1/2}\right).
%\end{align}
\begin{align}\label{thm.eq null}
    \sup_{w}\left\vert \hat R_n(w)-\frac{1}{n}\sum_{i=1}^n\zeta(Y_i,X_i,Z_i;x,z)\right\vert=o_p\left(n^{-1/2}\right).
\end{align}
Furthermore,
\begin{equation*}
\sqrt{n}\hat R_{n}(\cdot) \Longrightarrow R_{\infty}(\cdot),
\end{equation*}
where $R_{\infty}(\cdot)$ is a centered Gaussian process with the covariance structure given by
\begin{align*}
    &Cov\left[R_{\infty}(w), R_{\infty}(w^\prime)\right]= \mathbb{E}\left[\zeta(Y,X,Z;x,z)\zeta(Y,X,Z;x^\prime,z^\prime)\right].
\end{align*}
%where
%\begin{align*}
%    \zeta_\infty(Y,X,Z;x,z)=\left[Y-m(X)\right]f_X(X)1_x(X)\mathcal{P}1_z(Z).
%\end{align*}
\end{theorem}

To complement the theoretical results under the null hypothesis, and building upon Theorem \ref{thm.null}, the asymptotic distributions of the $CvM_n$ and $KS_n$ statistics introduced in Section \ref{sec.Test} are established by applying the continuous mapping theorem, as discussed in \cite{van1996weak}.
\begin{corollary}\label{Cor.null}
Suppose that Assumptions \ref{ass.sample}--\ref{ass.bandwidth} hold. Under the null hypothesis $H_0$,
\begin{align*}
    &nCvM_{n}\stackrel{d}\longrightarrow \int\left\vert R_{\infty}(w)\right\vert^2F_W(dw) \text{ and } \sqrt{n}KS_{n}\stackrel{d}\longrightarrow \sup\limits_{w}\left\vert R_{\infty}(w)\right\vert.
\end{align*}
\end{corollary}
Theorem \ref{thm.null} establishes the theoretical advantage of the proposed test, namely, that despite the applications of kernel smoothing methods, the resulting empirical process still achieves the parametric convergence rate, with a centered Gaussian process as its limiting process. Subsequently, the Corollary \ref{Cor.null} ensures that the statistics $nCvM_{n}$ and $\sqrt{n}KS_{n}$ converge to the squared norm and the sup norm of the aforementioned Gaussian process, respectively. The uniform decomposition given in \eqref{thm.eq null} is both theoretically and empirically attractive for implementing the multiplier bootstrap detailed in Section \ref{sec.boot}. To investigate the power of the proposed test, we introduce a sequence of local alternatives converging to the null hypothesis at the parametric rate, 
\begin{align}\label{hyp.local}
    H_{1n}:\mathbb{E}(Y \vert W)=\mathbb{E}(Y \vert X)+\frac{\Psi(W)}{\sqrt{n}} \quad a.s.,
\end{align}
under which the asymptotic theory of the constructed empirical process is subsequently established. We note that $\Psi:\mathbb{R}^{p+q}\to\mathbb{R}$ is a bounded nonzero function.
\begin{theorem}\label{thm.local}
Suppose that Assumptions \ref{ass.sample}--\ref{ass.bandwidth} hold. Under the sequence of local alternatives $H_{1n}$,
\begin{align*}
\sqrt{n}\hat R_{n}(\cdot) \Longrightarrow R_{\infty}(\cdot)+\mu(\cdot),
\end{align*}
where $R_{\infty}(\cdot)$ is the centered Gaussian process as described in Theorem \ref{thm.null}, and the deterministic shift function $\mu(w)$ is given by
\begin{align*}
    \mu(w) = \mathbb{E}\left\{\Psi(W)f_X(X)1_x(X)\left[1_z(Z)-f_W(W)\Delta^{-1}(X)G(z;X)\right]\right\}.
\end{align*}
\end{theorem}
Owing to the similarity between the theoretical derivations under the sequence of local alternatives and those established under the null hypothesis as in Corollary \ref{Cor.null}, the results $nCvM_{n}\stackrel{d}\longrightarrow \int\vert R_{\infty}(w)+\mu(w)\vert^2F_W(dw)$ and $\sqrt{n}KS_{n}\stackrel{d}\longrightarrow \sup\limits_{w}\vert R_{\infty}(w)+\mu(w)\vert$ are presented directly, where we omit the verification details of the convergence of the statistics. Consequently, the local powers of the proposed statistics are characterized by $\mu(\cdot)$. Specifically, the deterministic shift term determines the non-centrality of the limiting process in comparison with the centered Gaussian limiting process under the null hypothesis, and it leads to the fact that whenever the deterministic shift function is nonzero for at least some $w\in\mathbb R^{p+q}$ with a positive Lebesgue measure, the proposed $CvM_n$ and $KS_n$ statistics possess nontrivial power against the local alternatives. Indeed, since $\mu(\cdot)$ depends on the conditional $L^2$ inner product of $\Psi(X,Z)$ and $\mathcal{P}1_z(Z)$ given $X$, in the sense that $\mu(w) = \mathbb{E}\{f_X(X)1_x(X)\mathbb{E}[\Psi(W)\mathcal{P}1_z(Z)\vert X]\}$, the nontrivial local power is well-established as long as for any constant $\gamma$,
\begin{align*}
    \mathbb{P}\left[\Psi(W)=\gamma f_W(W)\right]<1.
\end{align*}

Finally, to derive the asymptotic global power properties of the proposed test statistics, we investigate the asymptotic behavior of the constructed empirical process under the alternative hypothesis.
\begin{theorem}\label{thm.alt}
%Under $H_{1}$, if Assumptions \ref{ass.sample}--\ref{ass.bandwidth} hold,
Suppose that Assumptions \ref{ass.sample}--\ref{ass.bandwidth} hold. Under the alternative hypothesis $H_1$, 
\begin{equation*}
    \sup_{w\in\Pi}\left\vert \hat R_{n}(w) - C(w) \right\vert = o_p(1),
\end{equation*}
where 
\begin{align*}
    &C(w)  = \mathbb{E}\left\{\left[\mathbb{E}\left(Y\vert W\right)-\mathbb{E}\left(Y\vert X\right)\right]f_X(X)1_x(X)\left[1_z(Z)-f_W(W)\Delta^{-1}(X)G(z;X)\right]\right\}.
\end{align*}
\end{theorem}
The correspondence between local and global power is straightforward. The local power is driven by $\mu(\cdot)$, and the global power, in turn, depends on the conditional $L^2$ inner product of $\mathbb{E}(Y\vert W)-\mathbb{E}(Y\vert X)$ and $\mathcal{P}1_z(Z)$ given $X$. More specifically, once
\begin{align}\label{thm.alt eq}
    \mathbb{P}\left\{\left[\mathbb{E}\left(Y\vert W\right)-\mathbb{E}\left(Y\vert X\right)\right]=\gamma f_W(W)\right\}<1
\end{align}
is satisfied for any constant $\gamma$, the unconditional expectation $C(w)\neq 0$. This in turn implies that, as $n$ goes to infinity, the constructed empirical process converges uniformly in probability to a non-vanishing function $C(\cdot)$, which further leads the test statistics to diverge to positive infinity in probability. However, we do not regard the potential loss of power along certain directions of the alternative space as a primary empirical concern. As shown in \cite{bierens1997asymptotic}, ICM–type tests are admissible under suitable regularity conditions, so that no uniformly more powerful procedure exists within a broad class of alternatives. These trade–offs are detailed in the simulation studies in Section \ref{sec.Simulation}.

\section{Multiplier bootstrap}\label{sec.boot}
The case-dependent asymptotic limiting process established in Section \ref{sec.Asy} necessitates the consideration of bootstrap procedures. In the traditional literature, two distinct bootstrap schemes are frequently considered. The residual-based bootstrap employed in \cite{koul1994bootstrapping} generates bootstrap samples by resampling from the empirical distribution of the estimated residuals (possibly after a smoothing step), and has been widely employed in classical nonparametric testing procedures. In contrast, the wild bootstrap, as mentioned in \cite{mammen1993bootstrap}, constructs bootstrap errors by multiplying the estimated residuals with independent random multipliers of zero mean and unit variance, thereby creating bootstrap samples without explicitly resampling from the residual distribution. 

Although traditional bootstrap methods may be justified after careful theoretical verification, we consider more computationally efficient alternatives. Motivated by \cite{van1996weak}, the uniform decomposition in \eqref{thm.eq null} inspires an intuitive, easy-to-implement multiplier bootstrap procedure. Towards this end, define the multiplier bootstrapped projected process as
\begin{align}\label{stat.boot}
    \hat R_n^\ast(w)=\frac{1}{n}\sum_{i=1}^n V_i\hat{\epsilon}_i\hat{f}_X(X_i)1_x(X_i)\hat{\mathcal{P}}_n1(Z_i\leq z),
\end{align}
where $\{V_i\}_{i=1}^n$ is a sequence of random variables (i.e., multipliers) that are mean zero, unit variance, and independent of the original sample $\{(Y_i,W_i^\top)^\top\}_{i=1}^n$. It is worth noting that in contrast to the multiplier bootstrap proposed in \cite{delgado2001significance}, which is motivated by the uniform decomposition in \eqref{stat.delgado}, our multiplier bootstrap version $\hat R_n^\ast(w)$ does not involve a random denominator thanks to the novel projection, and thus it is not necessary to impose the stringent compact support assumption $\mathbb{P}(f_X(X)>\nu)=1$ for some constant $\nu>0$. This fact allows us to
accommodate important distributions for X with unbounded support, like the t and normal.

\begin{theorem}\label{thm.boot}
Suppose that Assumptions \ref{ass.sample}--\ref{ass.bandwidth} hold. Under the null hypothesis $H_0$ or the sequence of local alternatives $H_{1n}$,
\begin{align*}
\sqrt{n}\hat R^{\ast}_{n}(\cdot) \underset{\ast}{\Longrightarrow } R_{\infty}(\cdot)\quad \text{in probability},
\end{align*}
where \textquotedblleft$\underset{\ast}{\Longrightarrow }$\textquotedblright denotes weak convergence under the bootstrap law, i.e., conditional on the original sample $\{(Y_i,W_i^\top)^\top\}_{i=1}^n$, and $R_{\infty}(\cdot)$ is the centered Gaussian process as described in Theorem \ref{thm.null}. Under the alternative hypothesis $H_1$,
\begin{align*}
\sqrt{n}\hat R^{\ast}_{n}(\cdot) \underset{\ast}{\Longrightarrow } R^1_{\infty}(\cdot)\quad \text{in probability},
\end{align*}
where $R^1_{\infty}(\cdot)$ is a centered Gaussian process that is different from $R_{\infty}(\cdot)$. 
\end{theorem}
It is worth noting that, under the null hypothesis, the bootstrap version coincides with the original sample-based empirical process in both convergence rate and asymptotic distribution. By constructing the multiplier bootstrap versions of the $CvM_n$ and $KS_n$ statistics in a manner analogous to Section \ref{sec.Test},
\begin{align*}
    CvM_n^\ast = \int\left\vert \hat R^\ast_n(w)\right\vert^2F_{W_n}(dw) \quad \text{and}\quad KS_n^\ast = \sup_{w}\left\vert \hat R^\ast_n(w)\right\vert,
\end{align*}
and applying the continuous mapping theorem in a way analogous to that employed in Corollary \ref{Cor.null}, we establish that under the null hypothesis the bootstrap $CvM_n^\ast$ and $KS_n^\ast$ statistics consistently approximate the asymptotic null distributions of $CvM_n$ and $KS_n$, respectively, thereby ensuring that the empirical sizes of the two tests are theoretically close to the nominal level. Under local alternatives $H_{1n}$, however, the local power depends on $\mu(\cdot)$ introduced in Section \ref{sec.Asy}. This is because the bootstrap versions converge to the same limiting processes as under the null, whereas the original statistics converge to limiting processes with additional drift components. Therefore, taking the $CvM_n$ statistic as an example, the different limiting distributions of $CvM_n$ and its bootstrap counterpart $CvM_n^\ast$ imply that the proportion of rejecting $H_0$ will exceed the nominal level, thereby confirming the nontrivial local power of the proposed test under $H_{1n}$. As a consequence of the above analysis, the asymptotic critical value (still taking the $CvM_n$ statistic as an example) at the significance level $\alpha$ is $c^\ast_{\alpha}=\inf\{c_\alpha\in[0,\infty):\lim_{n\to\infty}\mathbb{P}_n^\ast(nCvM_n^\ast>c_\alpha)=\alpha\}$, where $\mathbb{P}_n^\ast$ is the bootstrap probability under the bootstrap law. In practice, $c^\ast_{\alpha}$ can be approximated as $c^\ast_{n,\alpha} = \{nCvM_n^\ast\}_{B(1-\alpha)}$, the $B(1-\alpha)$-th order statistic for $B$ replicates $\{nCvM_{n,b}^\ast\}_{b=1}^B$ %of $CvM_n^\ast$ 
and we reject $H_0$ if $nCvM_n>c^\ast_{n,\alpha}$. Finally, under the alternative hypothesis $H_1$, when the condition \eqref{thm.alt eq} is satisfied, the consistency of our bootstrap statistics against $H_1$ follows from the fact that the original statistics diverge to infinity in probability, whereas the bootstrap versions remain stochastically bounded.

\section{Testing conditional independence}\label{sec.CI}
The projection techniques developed in this article to handle the nonparametric estimation effect can readily be extended to testing other restrictions on regression curves, thereby allowing ICM-type tests, which are particularly advantageous in certain settings, to be applied to these important problems. A particularly important example is the problem of testing conditional independence, which has attracted considerable attention in recent years and has been emphasized as a key issue in, among others, \cite{su2014testing}, \cite{huang2016flexible}, \cite{wang2018characteristic}, \cite{li2020nonparametric}, \cite{shah2020hardness}, \cite{neykov2021minimax} and \cite{cai2022distribution}. Formally, we are interested in testing whether $Y$ is independent of $Z$ given $X$, denoted by $H_0^{CI}:Y\perp Z\vert X$. Using the previous notations, we consider testing
\begin{align*}
    H_0^{CI}:\mathbb{E}\left[1_y(Y)\vert W\right] = \mathbb{E}\left[1_y(Y)\vert X\right]\quad a.s. \text{ for all } y\in \mathbb{R},
\end{align*}
which is further equivalent to
\begin{align*}
    H_0^{CI}:\mathbb{E}\left[\epsilon(y)f_X(X)1_w(W)\right] = 0\text{ for all } (y,w)\in \mathbb{R}^{1+q+p},
\end{align*}
where $\epsilon(y) = 1_y(Y)-\mathbb{E}[1_y(Y)\vert X] = 1_y(Y)-F_{Y\vert X}(y\vert X)$, with $F_{Y\vert X}(y\vert x)$ the conditional distribution function of $Y$ given $X$. Similar to $T_n(w)$ for the case of significance testing in the conditional mean, \cite{delgado2001significance} have proposed to use the following unprojected $U$-process to test $H_0^{CI}$:
\begin{align*}
    \hat L_n(y,w) = &\frac{1}{n}\sum_{i=1}^n\hat{\epsilon_i}(y)\hat{f}_X(X_i)1_w(W_i)\\
    =&\frac{1}{n(n-1)}\sum_{i=1}^n\sum_{j=1,j\neq i}^n\frac{1}{a^q}K\left(\frac{X_i-X_j}{a}\right)\left[1_y(Y_i)-1_y(Y_j)\right]1_w(W_i),
\end{align*}
which is shown to satisfy
\begin{align}\label{stat.delgado ci}
    \sup_w\left\vert \hat L_n(y,w)-\frac{1}{n}\sum_{i=1}^n\epsilon_i(y)f_X(X_i)1_x(X_i)\left[1_z(Z_i)-F_{Z\vert X}(z\vert X_i)\right]\right\vert = o_p\left(n^{-1/2}\right).
\end{align}
Motivated by the uniform decomposition in \eqref{stat.delgado ci} and the expected case-dependent asymptotic limiting process, we would need to obtain critical values for the tests based on $\hat L_n(\cdot)$ via a multiplier bootstrap scheme similar to that described in Section \ref{sec.boot}. This, however, requires the nonparametric estimation of $F_{Z\vert X}(z|X_i)$. The random denominator problem leads either to overly stringent assumptions on the distribution of $X$ or to substantial additional technical difficulties. Similar to the discussion in Section \ref{sec.Test}, when ICM-type tests are used to test conditional independence, \textquotedblleft nonparametric estimation effect\textquotedblright{} arises. In particular, these difficulties stem from the presence of the term
\begin{align*}
    \frac{1}{n}\sum_{i=1}^n\epsilon_i(y)f_X(X_i)1_x(X_i)F_{Z\vert X}(z\vert X_i).
\end{align*}

Thanks to the projection-based principle proposed in previous sections, the following projected empirical process can be defined to solve the problems discussed above:
\begin{align*}
    \hat I_n(y,w) = \frac{1}{n}\sum_{i=1}^n\hat{\epsilon}_i(y)\hat{f}_X(X_i)1_x(X_i)\hat{\mathcal{P}}_n1_z(Z_i),
\end{align*}
where $\hat{\epsilon}_i(y)=1_y(Y_i)-\hat{F}_{Y\vert X}(y\vert X_i)$ is the estimator for $\epsilon_i(y)$, with the leave-one-out nonparametric estimator
\begin{align*}
    \hat{F}_{Y\vert X}(y\vert X_i) = \frac{1}{\hat{f}_X(X_i)}\frac{1}{(n-1)a^q}\sum_{j=1,j\neq i}^n K\left(\frac{X_i-X_j}{a}\right)1_y(Y_j).
\end{align*}
All other quantities are defined as before. The analysis of $\hat I_n(y,w)$ is identical to $\hat R_n(w)$ but with $\hat{\epsilon}_i$ substituted by $\hat{\epsilon}_i(y)$. Thus, reasoning as in the Section \ref{sec.Asy}, we have 
\begin{align}\label{stat.CI}
    \sup_{w}\left\vert \hat I_n(y,w)-\frac{1}{n}\sum_{i=1}^n\epsilon_i(y)f_X(X_i)1_x(X_i)\mathcal{P}1_z(Z_i)\right\vert = o_p\left(n^{-1/2}\right).
\end{align}
This leads to the fact that our constructed projected $U$-process $\sqrt n\hat I_n(\cdot)$ weakly converges to a centered Gaussian process $R_\infty^{CI}(\cdot)$ with covariance structure 
\begin{align*}
    \mathbb{E}\left[\epsilon(y)\epsilon(y^\prime)f^2_X(X)1_x(X)1_{x^\prime}(X)\mathcal{P}1_z(Z)\mathcal{P}1_{z^\prime}(Z)\right]
\end{align*}
at the parametric rate under the null. Correspondingly, under the alternative hypothesis, the constructed empirical process converges uniformly in probability to
\begin{align*}
    C^{CI}(y,w) = \mathbb{E}\left\{\left[\mathbb{E}\left(1_y(Y)\vert W\right)-\mathbb{E}\left(1_y(Y)\vert X\right)\right]f_X(X)1_x(X)\mathcal{P}1_z(Z)\right\}.
\end{align*}
As a result, the test is consistent as long as for any constant $\gamma$, the following requirement
\begin{align}\label{stat.alt CI}
    \mathbb{P}\left\{\left[\mathbb{E}\left(1_y(Y)\vert W\right)-\mathbb{E}\left(1_y(Y)\vert X\right)\right]=\gamma f_W(W)\right\}<1
\end{align}
holds. In addition, this condition also characterizes the nontrivial asymptotic local power of the test for conditional independence. The widely used CvM and KS test statistics can be constructed as 
\begin{align*}
   CvM^{CI}_n = \int\left\vert \hat I_n(y,w)\right\vert^2F_{Y_n,W_n}(dy,dw)  \quad \text{and}\quad KS^{CI}_n = \sup_{(y,w)}\left\vert \hat I_n(y,w)\right\vert,
\end{align*}
so that large values provide evidence against the null in favor of the alternative hypothesis, where $F_{Y_n,W_n}(\cdot)$ is a random measure that converges in probability to $F_{Y,W}(\cdot)$ that is absolutely continuous with respect to the Lebesgue measure on $\mathbb{R}^{1+q+p}$. To implement the proposed test for conditional independence, still motivated by \eqref{stat.CI}, we construct the multiplier bootstrap version of $\hat I_n(y,w) $ as follows:
\begin{align*}
    \hat I_n^\ast(y,w) = \frac{1}{n}\sum_{i=1}^nV_i\hat{\epsilon}_i(y)\hat{f}_X(X_i)1_x(X_i)\hat{\mathcal{P}}_n1_z(Z_i)
\end{align*}
and use it to obtain the critical values. The desired asymptotic properties follow from the following result:
\begin{align*}
    \sup_{w}\left\vert \hat I^\ast_n(y,w)-\frac{1}{n}\sum_{i=1}^nV_i\epsilon_i(y)f_X(X_i)1_x(X_i)\mathcal{P}1_z(Z_i)\right\vert=o_p\left(n^{-1/2}\right).
\end{align*}
The corresponding multiplier bootstrap $CvM^{CI,\ast}_n$ and $KS^{CI,\ast}_n$ test statistics can be defined in a way analogous to that described in Section \ref{sec.boot}, and can be shown, under the null and under local alternatives, to provide good approximations to the sup norm and the squared norm of $R_\infty^{CI}(\cdot)$, respectively. Under fixed alternatives, the bootstrap versions of these statistics remain stochastically bounded, whereas the original statistics diverge to infinity whenever condition \eqref{stat.alt CI} is satisfied, thereby establishing the validity of the bootstrap procedure. The implementation of the resulting test can therefore follow the same steps as those outlined in Section \ref{sec.boot}. More Simulation results are presented in Section \ref{sec.AppendixA}.

\section{Simulation study}\label{sec.Simulation}
In this section, a Monte Carlo study is conducted to evaluate the finite-sample performance of the proposed nonparametric projection method. The main objectives are as follows. First, the proposed CvM test based on projection (denoted by PJ) is compared with the non-projected CvM test of \cite{delgado2001significance} (denoted by DM) in terms of empirical size and power.\footnote{For brevity, we have focused on comparing the CvM-type statistics. The comparison using the KS-type statistics is similar and available upon request.} Second, the robustness of the test performance with respect to bandwidth choice is examined under various bandwidth settings. Third, the behavior of the projection-based test is investigated across different covariate distributions and alternative frequencies. 

In the simulation design presented in the main text, the covariates $X$ and $Z$ are both one-dimensional and follow a specified dependence structure $X=Z+U$, where $Z$ and $U$ are independent $N(0,1)$ variables. It is worth emphasizing that our procedures are not restricted to this case. The tests are applicable when $X$ and $Z$ are multivariate, and to illustrate this property, we also consider designs with $p=2$ and $q=2$. The proposed tests also apply when $X$ and $Z$ follow the uniform distribution $U(0,1)$. Moreover, when some parts of Assumptions \ref{ass.sample}--\ref{ass.bandwidth} are violated, for instance when both covariates are $N(0,1)$, the size accuracy of the test is affected but remains within an acceptable range and compares favorably with the DM test, which suggests that our procedure has a broader practical scope. For the sake of brevity, and because the main conclusions are similar, these additional results are reported in Appendix \ref{sec.AppendixA}, whereas the main text focuses on the case $X=Z+U$ with $Z$ and $U$ being independent $N(0,1)$. %$Z\sim N(0,1)$, $U\sim N(0,1)$, and independent. 
To compare the performance of the tests under alternatives with different frequencies, the response variable $Y$ is generated under the following data-generating process (DGP),
\begin{align*}
    Y = 1+X+\sin\left(\gamma Z\right)+\epsilon,
\end{align*}
where $\epsilon\sim N(0,1)$ and is independent of $(X,Z)$. Under the null hypothesis, the parameter $\gamma$ is fixed at zero, whereas under the alternative hypothesis, the frequency is determined by the value of $\gamma$. Moreover, to assess, in the spirit of \cite{bierens1997asymptotic}, the admissibility properties of ICM-type tests, we also report in Appendix \ref{sec.AppendixA} simulation results for designs with interaction effects between covariates. In particular, we consider alternatives of the form $X(Z^2-1)$ and examine the power of the proposed procedures in this case. In addition, the adoption of the kernel function and method of bandwidth selection follows those in \cite{delgado2001significance}, where an $l$-th order Epanechnikov kernel is employed, and the bandwidth is determined by the rule-of-thumb with coefficient $c\in\{0.5,1,2\}$ to assess the sensitivity to bandwidth. The order of the kernel and the explicit bandwidth formula are consistent with the requirements of Assumption \ref{ass.structual}. Specifically, $b=cn^{-1/3}$ is used for $l=2$, and $b=cn^{-1/6}$ is adopted for $l=4$. The results were examined under sample sizes $n\in\{200,400\}$ with the nominal significance level set at $\alpha\in\{0.01,0.05,0.1\}$, and all outcomes were obtained from $1000$ Monte Carlo replications and $199$ multiplier bootstrap samples.

\begin{table}[htbp]
  \centering
  \setlength{\tabcolsep}{5pt}
  \renewcommand{\arraystretch}{1.1}
  \caption{Significance test when $X=Z+U$, $Z\sim N(0,1)$ and $U\sim N(0,1)$}
  \begin{tabular}{llcccccccc}
    \toprule
    $c=0.5$ & & \multicolumn{3}{c}{\textbf{$n=200$}} & \multicolumn{3}{c}{\textbf{$n=400$}} \\
    \cmidrule(lr){3-5}\cmidrule(lr){6-8}
    & \textbf{$\alpha/\gamma$} & 0 & 5 & 10 & 0 & 5 & 10\\
    \midrule
    \multirow{3}{*}{DM}
      & 0.10  & 0.144 & 0.413 & 0.200 & 0.137 & 0.836 & 0.225\\
      & 0.05  & 0.072 & 0.212 & 0.091 & 0.070 & 0.583 & 0.131\\
      & 0.01  & 0.017 & 0.051 & 0.026 & 0.022 & 0.175 & 0.038\\
    \midrule
    \multirow{3}{*}{PJ}
      & 0.10  & 0.113 & 0.456 & 0.211 & 0.110 & 0.758 & 0.342\\
      & 0.05  & 0.055 & 0.283 & 0.126 & 0.051 & 0.588 & 0.221\\
      & 0.01  & 0.011 & 0.106 & 0.042 & 0.018 & 0.312 & 0.081\\
    \midrule
    $c=1$ & & \multicolumn{3}{c}{\textbf{$n=200$}} & \multicolumn{3}{c}{\textbf{$n=400$}} \\
    \cmidrule(lr){3-5}\cmidrule(lr){6-8}
    & \textbf{$\alpha/\gamma$} & 0 & 5 & 10 & 0 & 5 & 10\\
    \midrule
    \multirow{3}{*}{DM}
      & 0.10  & 0.133 & 0.419 & 0.172 & 0.131 & 0.856 & 0.215\\
      & 0.05  & 0.065 & 0.210 & 0.095 & 0.068 & 0.590 & 0.126\\
      & 0.01  & 0.018 & 0.043 & 0.018 & 0.018 & 0.171 & 0.030\\
    \midrule
    \multirow{3}{*}{PJ}
      & 0.10  & 0.096 & 0.453 & 0.219 & 0.095 & 0.775 & 0.355\\
      & 0.05  & 0.043 & 0.307 & 0.119 & 0.045 & 0.614 & 0.227\\
      & 0.01  & 0.007 & 0.129 & 0.040 & 0.009 & 0.342 & 0.096\\
    \midrule
    $c=2$ & & \multicolumn{3}{c}{\textbf{$n=200$}} & \multicolumn{3}{c}{\textbf{$n=400$}} \\
    \cmidrule(lr){3-5}\cmidrule(lr){6-8}
    & \textbf{$\alpha/\gamma$} & 0 & 5 & 10 & 0 & 5 & 10\\
    \midrule
    \multirow{3}{*}{DM}
      & 0.10  & 0.133 & 0.442 & 0.186 & 0.144 & 0.861 & 0.234\\
      & 0.05  & 0.062 & 0.213 & 0.089 & 0.069 & 0.599 & 0.127\\
      & 0.01  & 0.019 & 0.040 & 0.021 & 0.022 & 0.186 & 0.035\\
    \midrule
    \multirow{3}{*}{PJ}
      & 0.10  & 0.086 & 0.484 & 0.201 & 0.096 & 0.839 & 0.329\\
      & 0.05  & 0.040 & 0.340 & 0.112 & 0.048 & 0.674 & 0.189\\
      & 0.01  & 0.012 & 0.126 & 0.032 & 0.011 & 0.379 & 0.081\\
    \bottomrule
  \end{tabular}
  \label{tab:SNsind1p1q1}
\end{table}

\begin{table}[htbp]
  \centering
  \setlength{\tabcolsep}{5pt}
  \renewcommand{\arraystretch}{1.1}
  \caption{Significance test when $X\sim N(0,1)$ and $Z\sim N(0,1)$}
  \begin{tabular}{llcccccccc}
    \toprule
    $c=0.5$ & & \multicolumn{3}{c}{\textbf{$n=200$}} & \multicolumn{3}{c}{\textbf{$n=400$}} \\
    \cmidrule(lr){3-5}\cmidrule(lr){6-8}
    & \textbf{$\alpha/\gamma$} & 0 & 5 & 10 & 0 & 5 & 10\\
    \midrule
    \multirow{3}{*}{DM}
      & 0.10  & 0.164 & 0.428 & 0.196 & 0.147 & 0.729 & 0.257\\
      & 0.05  & 0.084 & 0.249 & 0.098 & 0.079 & 0.542 & 0.132\\
      & 0.01  & 0.020 & 0.070 & 0.030 & 0.024 & 0.230 & 0.037\\
    \midrule
    \multirow{3}{*}{PJ}
      & 0.10  & 0.109 & 0.346 & 0.256 & 0.115 & 0.648 & 0.411\\
      & 0.05  & 0.052 & 0.167 & 0.157 & 0.061 & 0.394 & 0.264\\
      & 0.01  & 0.012 & 0.056 & 0.040 & 0.014 & 0.169 & 0.114\\
    \midrule
    $c=1$ & & \multicolumn{3}{c}{\textbf{$n=200$}} & \multicolumn{3}{c}{\textbf{$n=400$}} \\
    \cmidrule(lr){3-5}\cmidrule(lr){6-8}
    & \textbf{$\alpha/\gamma$} & 0 & 5 & 10 & 0 & 5 & 10\\
    \midrule
    \multirow{3}{*}{DM}
      & 0.10  & 0.143 & 0.414 & 0.180 & 0.135 & 0.743 & 0.249\\
      & 0.05  & 0.078 & 0.241 & 0.092 & 0.066 & 0.546 & 0.128\\
      & 0.01  & 0.019 & 0.068 & 0.022 & 0.020 & 0.227 & 0.041\\
    \midrule
    \multirow{3}{*}{PJ}
      & 0.10  & 0.101 & 0.299 & 0.245 & 0.107 & 0.642 & 0.406\\
      & 0.05  & 0.051 & 0.175 & 0.151 & 0.058 & 0.460 & 0.280\\
      & 0.01  & 0.013 & 0.046 & 0.059 & 0.013 & 0.207 & 0.126\\
    \midrule
    $c=2$ & & \multicolumn{3}{c}{\textbf{$n=200$}} & \multicolumn{3}{c}{\textbf{$n=400$}} \\
    \cmidrule(lr){3-5}\cmidrule(lr){6-8}
    & \textbf{$\alpha/\gamma$} & 0 & 5 & 10 & 0 & 5 & 10\\
    \midrule
    \multirow{3}{*}{DM}
      & 0.10  & 0.146 & 0.401 & 0.165 & 0.131 & 0.744 & 0.254\\
      & 0.05  & 0.066 & 0.226 & 0.088 & 0.065 & 0.539 & 0.117\\
      & 0.01  & 0.018 & 0.062 & 0.023 & 0.019 & 0.237 & 0.041\\
    \midrule
    \multirow{3}{*}{PJ}
      & 0.10  & 0.100 & 0.388 & 0.234 & 0.105 & 0.786 & 0.390\\
      & 0.05  & 0.052 & 0.238 & 0.132 & 0.051 & 0.593 & 0.273\\
      & 0.01  & 0.012 & 0.067 & 0.046 & 0.017 & 0.282 & 0.110\\
    \bottomrule
  \end{tabular}
  \label{tab:SNsind2p1q1}
\end{table}

\begin{table}[htbp]
  \centering
  \setlength{\tabcolsep}{5pt}
  \renewcommand{\arraystretch}{1.1}
  \caption{Significance test when $X\sim U(0,1)$ and $Z\sim U(0,1)$}
  \begin{tabular}{llcccccccc}
    \toprule
    $c=0.5$ & & \multicolumn{3}{c}{\textbf{$n=200$}} & \multicolumn{3}{c}{\textbf{$n=400$}} \\
    \cmidrule(lr){3-5}\cmidrule(lr){6-8}
    & \textbf{$\alpha/\gamma$} & 0 & 5 & 10 & 0 & 5 & 10\\
    \midrule
    \multirow{3}{*}{DM}
      & 0.10  & 0.126 & 0.999 & 0.437 & 0.125 & 1.000 & 0.678\\
      & 0.05  & 0.058 & 0.983 & 0.303 & 0.054 & 1.000 & 0.569\\
      & 0.01  & 0.015 & 0.772 & 0.147 & 0.015 & 0.999 & 0.366\\
    \midrule
    \multirow{3}{*}{PJ}
      & 0.10  & 0.105 & 0.996 & 0.437 & 0.105 & 1.000 & 0.739\\
      & 0.05  & 0.054 & 0.980 & 0.335 & 0.055 & 1.000 & 0.644\\
      & 0.01  & 0.010 & 0.840 & 0.159 & 0.019 & 0.999 & 0.433\\
    \midrule
    $c=1$ & & \multicolumn{3}{c}{\textbf{$n=200$}} & \multicolumn{3}{c}{\textbf{$n=400$}} \\
    \cmidrule(lr){3-5}\cmidrule(lr){6-8}
    & \textbf{$\alpha/\gamma$} & 0 & 5 & 10 & 0 & 5 & 10\\
    \midrule
    \multirow{3}{*}{DM}
      & 0.10  & 0.120 & 0.999 & 0.461 & 0.117 & 1.000 & 0.700\\
      & 0.05  & 0.058 & 0.988 & 0.342 & 0.055 & 1.000 & 0.592\\
      & 0.01  & 0.011 & 0.796 & 0.158 & 0.011 & 1.000 & 0.376\\
    \midrule
    \multirow{3}{*}{PJ}
      & 0.10  & 0.095 & 0.997 & 0.463 & 0.100 & 1.000 & 0.737\\
      & 0.05  & 0.052 & 0.988 & 0.332 & 0.061 & 1.000 & 0.645\\
      & 0.01  & 0.015 & 0.872 & 0.157 & 0.018 & 1.000 & 0.449\\
    \midrule
    $c=2$ & & \multicolumn{3}{c}{\textbf{$n=200$}} & \multicolumn{3}{c}{\textbf{$n=400$}} \\
    \cmidrule(lr){3-5}\cmidrule(lr){6-8}
    & \textbf{$\alpha/\gamma$} & 0 & 5 & 10 & 0 & 5 & 10\\
    \midrule
    \multirow{3}{*}{DM}
      & 0.10  & 0.125 & 0.998 & 0.509 & 0.123 & 1.000 & 0.742\\
      & 0.05  & 0.059 & 0.992 & 0.383 & 0.059 & 1.000 & 0.636\\
      & 0.01  & 0.018 & 0.830 & 0.184 & 0.016 & 1.000 & 0.430\\
    \midrule
    \multirow{3}{*}{PJ}
      & 0.10  & 0.107 & 0.998 & 0.424 & 0.109 & 1.000 & 0.700\\
      & 0.05  & 0.063 & 0.976 & 0.286 & 0.058 & 1.000 & 0.602\\
      & 0.01  & 0.016 & 0.861 & 0.129 & 0.017 & 0.999 & 0.397\\
    \bottomrule
  \end{tabular}
  \label{tab:SNsind3p1q1}
\end{table}

Tables \ref{tab:SNsind1p1q1}--\ref{tab:SNsind3p1q1} report results for the case $p=q=1$. Table \ref{tab:SNsind1p1q1} corresponds to the design $X=Z+U$, where $Z$ and $U$ are independent $N(0,1)$ variables, which satisfies Assumptions \ref{ass.sample}--\ref{ass.bandwidth} but not the bounded-support assumption on the distribution of $X$ imposed in \cite{delgado2001significance}. Table \ref{tab:SNsind2p1q1} considers the case where $X$ and $Z$ are independent $N(0,1)$ variables, and Table \ref{tab:SNsind3p1q1} the case where $X$ and $Z$ are independent $U(0,1)$ variables. As is typically observed, both size and power improve substantially as sample size increases, demonstrating that the proposed procedure possesses desirable large-sample properties. 

For comparison with the method introduced in \cite{delgado2001significance} (named DM), the performance of each method in each case is reported in the tables. It is observed that DM tests exhibit lower accuracy, whereas the projection-based method achieves significant improvements. Regarding the power properties, Table \ref{tab:SNsind1p1q1} shows that the proposed projection-based tests (PJ) perform markedly better than the DM tests, which illustrates the suitability of our procedure for a wider class of unbounded distributions that satisfy Assumptions \ref{ass.sample}--\ref{ass.bandwidth}. In Table \ref{tab:SNsind2p1q1}, we investigate how the two tests behave when commonly used distributions that violate these assumptions are employed. We observe that the DM tests suffer from substantial size distortion, with empirical rejection frequencies well above the nominal level (oversized), whereas the PJ tests exhibit much better size accuracy. In Table \ref{tab:SNsind3p1q1}, both tests perform well, which is mainly due to the fact that the data are generated from the standard uniform distribution, a commonly used design that is fully compatible with the theoretical assumptions.

Further notable advantages of the proposed tests include reasonable robustness to bandwidth selection, as shown by results across different values of $c$, and relatively high power, as evidenced by comparisons across different values of $\gamma$. These features are characteristic of ICM-type tests in general and are naturally also reflected in the performance of the DM tests.

Finally, as further support for the theoretical results established in Section \ref{sec.CI}, we use the above DGPs to test the conditional independence assumption and the results are reported in Tables \ref{tab:CIsind1p1q1}--\ref{tab:CIsind3p1q1}. The only difference in the design is that
\begin{align*}
    Y = 1+\left[X+\sin\left(\gamma Z\right)\right]\epsilon,
\end{align*}
where the conditional independence under test pertains only to the variance structure. Note that the significance tests discussed in Sections \ref{sec.Test}--\ref{sec.boot} can also be interpreted as tests of conditional independence for the conditional mean function, and the conclusions drawn from Tables \ref{tab:CIsind1p1q1}--\ref{tab:CIsind3p1q1} are similar to those obtained above.
\begin{table}[htbp]
  \centering
  \setlength{\tabcolsep}{5pt}
  \renewcommand{\arraystretch}{1.1}
  \caption{Testing conditional independence when $\Psi(X,Z)=\sin(\gamma Z)$, $X=Z+U$, $Z\sim N(0,1)$ and $U\sim N(0,1)$}
  \begin{tabular}{llcccccccc}
    \toprule
    $c=0.5$ & & \multicolumn{3}{c}{\textbf{$n=200$}} & \multicolumn{3}{c}{\textbf{$n=400$}} \\
    \cmidrule(lr){3-5}\cmidrule(lr){6-8}
    & \textbf{$\alpha/\gamma$} & 0 & 5 & 10 & 0 & 5 & 10\\
    \midrule
    \multirow{3}{*}{DM}
      & 0.10  & 0.200 & 0.137 & 0.152 & 0.169 & 0.162 & 0.133\\
      & 0.05  & 0.105 & 0.078 & 0.096 & 0.079 & 0.088 & 0.083\\
      & 0.01  & 0.032 & 0.023 & 0.033 & 0.022 & 0.025 & 0.025\\
    \midrule
    \multirow{3}{*}{PJ}
      & 0.10  & 0.122 & 0.119 & 0.129 & 0.100 & 0.140 & 0.107\\
      & 0.05  & 0.071 & 0.066 & 0.076 & 0.048 & 0.069 & 0.061\\
      & 0.01  & 0.023 & 0.018 & 0.022 & 0.012 & 0.021 & 0.019\\
    \midrule
    $c=1$ & & \multicolumn{3}{c}{\textbf{$n=200$}} & \multicolumn{3}{c}{\textbf{$n=400$}} \\
    \cmidrule(lr){3-5}\cmidrule(lr){6-8}
    & \textbf{$\alpha/\gamma$} & 0 & 5 & 10 & 0 & 5 & 10\\
    \midrule
    \multirow{3}{*}{DM}
      & 0.10  & 0.161 & 0.128 & 0.126 & 0.170 & 0.158 & 0.119\\
      & 0.05  & 0.084 & 0.062 & 0.073 & 0.076 & 0.087 & 0.070\\
      & 0.01  & 0.024 & 0.016 & 0.024 & 0.018 & 0.025 & 0.020\\
    \midrule
    \multirow{3}{*}{PJ}
      & 0.10  & 0.102 & 0.098 & 0.103 & 0.086 & 0.117 & 0.093\\
      & 0.05  & 0.049 & 0.045 & 0.060 & 0.048 & 0.059 & 0.047\\
      & 0.01  & 0.013 & 0.012 & 0.020 & 0.010 & 0.025 & 0.015\\
    \midrule
    $c=2$ & & \multicolumn{3}{c}{\textbf{$n=200$}} & \multicolumn{3}{c}{\textbf{$n=400$}} \\
    \cmidrule(lr){3-5}\cmidrule(lr){6-8}
    & \textbf{$\alpha/\gamma$} & 0 & 5 & 10 & 0 & 5 & 10\\
    \midrule
    \multirow{3}{*}{DM}
      & 0.10  & 0.197 & 0.121 & 0.122 & 0.197 & 0.138 & 0.123\\
      & 0.05  & 0.100 & 0.055 & 0.067 & 0.096 & 0.082 & 0.062\\
      & 0.01  & 0.026 & 0.013 & 0.020 & 0.020 & 0.025 & 0.019\\
    \midrule
    \multirow{3}{*}{PJ}
      & 0.10  & 0.092 & 0.096 & 0.099 & 0.084 & 0.102 & 0.093\\
      & 0.05  & 0.044 & 0.045 & 0.050 & 0.052 & 0.057 & 0.051\\
      & 0.01  & 0.008 & 0.006 & 0.010 & 0.012 & 0.018 & 0.018\\
    \bottomrule
  \end{tabular}
  \label{tab:CIsind1p1q1}
\end{table}

\begin{table}[htbp]
  \centering
  \setlength{\tabcolsep}{5pt}
  \renewcommand{\arraystretch}{1.1}
  \caption{Testing conditional independence when $\Psi(X,Z)=\sin(\gamma Z)$, $X\sim N(0,1)$ and $Z\sim N(0,1)$}
  \begin{tabular}{llcccccccc}
    \toprule
    $c=0.5$ & & \multicolumn{3}{c}{\textbf{$n=200$}} & \multicolumn{3}{c}{\textbf{$n=400$}} \\
    \cmidrule(lr){3-5}\cmidrule(lr){6-8}
    & \textbf{$\alpha/\gamma$} & 0 & 5 & 10 & 0 & 5 & 10\\
    \midrule
    \multirow{3}{*}{DM}
      & 0.10  & 0.188 & 0.150 & 0.149 & 0.159 & 0.152 & 0.139\\
      & 0.05  & 0.096 & 0.091 & 0.090 & 0.084 & 0.080 & 0.076\\
      & 0.01  & 0.030 & 0.023 & 0.027 & 0.022 & 0.023 & 0.026\\
    \midrule
    \multirow{3}{*}{PJ}
      & 0.10  & 0.137 & 0.143 & 0.123 & 0.110 & 0.144 & 0.109\\
      & 0.05  & 0.064 & 0.077 & 0.072 & 0.062 & 0.070 & 0.059\\
      & 0.01  & 0.014 & 0.025 & 0.023 & 0.017 & 0.018 & 0.020\\
    \midrule
    $c=1$ & & \multicolumn{3}{c}{\textbf{$n=200$}} & \multicolumn{3}{c}{\textbf{$n=400$}} \\
    \cmidrule(lr){3-5}\cmidrule(lr){6-8}
    & \textbf{$\alpha/\gamma$} & 0 & 5 & 10 & 0 & 5 & 10\\
    \midrule
    \multirow{3}{*}{DM}
      & 0.10  & 0.157 & 0.134 & 0.135 & 0.147 & 0.143 & 0.115\\
      & 0.05  & 0.088 & 0.073 & 0.069 & 0.081 & 0.066 & 0.062\\
      & 0.01  & 0.023 & 0.016 & 0.022 & 0.021 & 0.020 & 0.023\\
    \midrule
    \multirow{3}{*}{PJ}
      & 0.10  & 0.101 & 0.111 & 0.101 & 0.101 & 0.118 & 0.106\\
      & 0.05  & 0.049 & 0.057 & 0.054 & 0.051 & 0.060 & 0.055\\
      & 0.01  & 0.012 & 0.018 & 0.012 & 0.012 & 0.020 & 0.013\\
    \midrule
    $c=2$ & & \multicolumn{3}{c}{\textbf{$n=200$}} & \multicolumn{3}{c}{\textbf{$n=400$}} \\
    \cmidrule(lr){3-5}\cmidrule(lr){6-8}
    & \textbf{$\alpha/\gamma$} & 0 & 5 & 10 & 0 & 5 & 10\\
    \midrule
    \multirow{3}{*}{DM}
      & 0.10  & 0.178 & 0.134 & 0.120 & 0.158 & 0.133 & 0.108\\
      & 0.05  & 0.090 & 0.072 & 0.058 & 0.088 & 0.065 & 0.058\\
      & 0.01  & 0.018 & 0.014 & 0.016 & 0.019 & 0.019 & 0.021\\
    \midrule
    \multirow{3}{*}{PJ}
      & 0.10  & 0.099 & 0.103 & 0.080 & 0.107 & 0.115 & 0.109\\
      & 0.05  & 0.047 & 0.053 & 0.046 & 0.049 & 0.066 & 0.053\\
      & 0.01  & 0.011 & 0.018 & 0.013 & 0.011 & 0.015 & 0.016\\
    \bottomrule
  \end{tabular}
  \label{tab:CIsind2p1q1}
\end{table}

\begin{table}[htbp]
  \centering
  \setlength{\tabcolsep}{5pt}
  \renewcommand{\arraystretch}{1.1}
  \caption{Testing conditional independence when $\Psi(X,Z)=\sin(\gamma Z)$, $X\sim U(0,1)$ and $Z\sim U(0,1)$}
  \begin{tabular}{llcccccccc}
    \toprule
    $c=0.5$ & & \multicolumn{3}{c}{\textbf{$n=200$}} & \multicolumn{3}{c}{\textbf{$n=400$}} \\
    \cmidrule(lr){3-5}\cmidrule(lr){6-8}
    & \textbf{$\alpha/\gamma$} & 0 & 5 & 10 & 0 & 5 & 10\\
    \midrule
    \multirow{3}{*}{DM}
      & 0.10  & 0.184 & 0.954 & 0.275 & 0.180 & 1.000 & 0.506\\
      & 0.05  & 0.095 & 0.882 & 0.147 & 0.095 & 0.999 & 0.269\\
      & 0.01  & 0.034 & 0.563 & 0.042 & 0.028 & 0.957 & 0.073\\
    \midrule
    \multirow{3}{*}{PJ}
      & 0.10  & 0.102 & 0.863 & 0.281 & 0.104 & 0.997 & 0.592\\
      & 0.05  & 0.050 & 0.764 & 0.156 & 0.056 & 0.992 & 0.400\\
      & 0.01  & 0.010 & 0.521 & 0.057 & 0.012 & 0.955 & 0.173\\
    \midrule
    $c=1$ & & \multicolumn{3}{c}{\textbf{$n=200$}} & \multicolumn{3}{c}{\textbf{$n=400$}} \\
    \cmidrule(lr){3-5}\cmidrule(lr){6-8}
    & \textbf{$\alpha/\gamma$} & 0 & 5 & 10 & 0 & 5 & 10\\
    \midrule
    \multirow{3}{*}{DM}
      & 0.10  & 0.189 & 0.959 & 0.278 & 0.199 & 1.000 & 0.547\\
      & 0.05  & 0.107 & 0.902 & 0.148 & 0.110 & 0.999 & 0.282\\
      & 0.01  & 0.030 & 0.623 & 0.044 & 0.030 & 0.969 & 0.078\\
    \midrule
    \multirow{3}{*}{PJ}
      & 0.10  & 0.086 & 0.910 & 0.285 & 0.107 & 0.999 & 0.637\\
      & 0.05  & 0.052 & 0.828 & 0.173 & 0.058 & 0.996 & 0.440\\
      & 0.01  & 0.013 & 0.575 & 0.063 & 0.015 & 0.968 & 0.200\\
    \midrule
    $c=2$ & & \multicolumn{3}{c}{\textbf{$n=200$}} & \multicolumn{3}{c}{\textbf{$n=400$}} \\
    \cmidrule(lr){3-5}\cmidrule(lr){6-8}
    & \textbf{$\alpha/\gamma$} & 0 & 5 & 10 & 0 & 5 & 10\\
    \midrule
    \multirow{3}{*}{DM}
      & 0.10  & 0.363 & 0.964 & 0.305 & 0.374 & 1.000 & 0.594\\
      & 0.05  & 0.196 & 0.916 & 0.167 & 0.211 & 1.000 & 0.305\\
      & 0.01  & 0.059 & 0.697 & 0.045 & 0.066 & 0.983 & 0.081\\
    \midrule
    \multirow{3}{*}{PJ}
      & 0.10  & 0.097 & 0.933 & 0.344 & 0.101 & 1.000 & 0.692\\
      & 0.05  & 0.050 & 0.826 & 0.203 & 0.051 & 1.000 & 0.539\\
      & 0.01  & 0.013 & 0.508 & 0.082 & 0.012 & 0.961 & 0.244\\
    \bottomrule
  \end{tabular}
  \label{tab:CIsind3p1q1}
\end{table}

\section{Conclusion}\label{sec.Conclusion}

This paper introduces a novel nonparametric significance test that uses a tailored projected weighting function to address the nonparametric estimation effect.
%offering superior theoretical properties. 
The approach relaxes the restrictive compact support assumption, provides a computationally efficient multiplier bootstrap, and extends straightforwardly to testing conditional independence. Numerical evidence confirms the strong performance of our proposal in finite samples.

\bibliographystyle{apalike_revised}
\bibliography{Ref}

\newpage
\appendix
\section{Supplementary simulation results}\label{sec.AppendixA}
In this subsection, we replicate all tests presented in the main text to assess whether the conclusions still hold across different choices of the frequency of the alternatives, dimensionality, and covariate distributions. For notational convenience, we consider the following generic DGP:
\begin{align*}
    Y=1+X+\Psi(X,Z)+\epsilon.
\end{align*}

As a complement to the results reported in the main text, tables \ref{tab:SNsind1p2q1}--\ref{tab:SNsind3p2q1} and \ref{tab:SNsind1p1q2}--\ref{tab:SNsind3p1q2} report the results for the above designs in the cases $p=2$ and $q=2$, respectively, where the regression function is specified as $\Psi(x,z)=\sin(\gamma z^{(1)})$ and $z^{(1)}$ denotes the first component of $z$. These results show that, although the size accuracy and power of the tests deteriorate to some extent in the multivariate setting, they remain within an acceptable range, and the empirical size under the null is markedly improved relative to the DM test. For clarity, we list below all multivariate distributional designs considered in these cases.
\begin{align*}
    &X=Z^{(1)}+U,\quad  Z^{(1)}\sim N(0,1),Z^{(2)}\sim N(0,1), U\sim N(0,1);\\
    &X\sim N(0,1),Z^{(1)}\sim N(0,1),Z^{(2)}\sim N(0,1);\\
    &X\sim U(0,1),Z^{(1)}\sim U(0,1),Z^{(2)}\sim U(0,1);\\
    &X^{(1)}=Z+U,\quad  X^{(2)}\sim N(0,1),Z\sim N(0,1), U\sim N(0,1);\\
    &X^{(1)}\sim N(0,1),X^{(2)}\sim N(0,1),Z\sim N(0,1);\\
    &X^{(1)}\sim U(0,1),X^{(2)}\sim U(0,1),Z\sim U(0,1).
\end{align*}

\begin{table}[htbp]
  \centering
  \setlength{\tabcolsep}{5pt}
  \renewcommand{\arraystretch}{1.1}
  \caption{Significance test when $\Psi(X,Z)=\sin(\gamma Z^{(1)})$, $X=Z^{(1)}+U$, $Z^{(1)}\sim N(0,1)$, $Z^{(2)}\sim N(0,1)$ and $U\sim N(0,1)$}
  \begin{tabular}{llcccccccc}
    \toprule
    $c=0.5$ & & \multicolumn{3}{c}{\textbf{$n=200$}} & \multicolumn{3}{c}{\textbf{$n=400$}} \\
    \cmidrule(lr){3-5}\cmidrule(lr){6-8}
    & \textbf{$\alpha/\gamma$} & 0 & 5 & 10 & 0 & 5 & 10\\
    \midrule
    \multirow{3}{*}{DM}
      & 0.10  & 0.135 & 0.381 & 0.173 & 0.147 & 0.792 & 0.209\\
      & 0.05  & 0.067 & 0.209 & 0.091 & 0.086 & 0.502 & 0.122\\
      & 0.01  & 0.017 & 0.059 & 0.023 & 0.027 & 0.151 & 0.035\\
    \midrule
    \multirow{3}{*}{PJ}
      & 0.10  & 0.096 & 0.465 & 0.216 & 0.100 & 0.817 & 0.343\\
      & 0.05  & 0.044 & 0.310 & 0.118 & 0.053 & 0.668 & 0.198\\
      & 0.01  & 0.009 & 0.133 & 0.038 & 0.012 & 0.345 & 0.061\\
    \midrule
    $c=1$ & & \multicolumn{3}{c}{\textbf{$n=200$}} & \multicolumn{3}{c}{\textbf{$n=400$}} \\
    \cmidrule(lr){3-5}\cmidrule(lr){6-8}
    & \textbf{$\alpha/\gamma$} & 0 & 5 & 10 & 0 & 5 & 10\\
    \midrule
    \multirow{3}{*}{DM}
      & 0.10  & 0.120 & 0.355 & 0.146 & 0.127 & 0.806 & 0.189\\
      & 0.05  & 0.045 & 0.183 & 0.086 & 0.073 & 0.504 & 0.103\\
      & 0.01  & 0.011 & 0.054 & 0.022 & 0.027 & 0.143 & 0.032\\
    \midrule
    \multirow{3}{*}{PJ}
      & 0.10  & 0.088 & 0.480 & 0.191 & 0.107 & 0.803 & 0.334\\
      & 0.05  & 0.046 & 0.308 & 0.117 & 0.047 & 0.638 & 0.205\\
      & 0.01  & 0.010 & 0.126 & 0.038 & 0.013 & 0.343 & 0.076\\
    \midrule
    $c=2$ & & \multicolumn{3}{c}{\textbf{$n=200$}} & \multicolumn{3}{c}{\textbf{$n=400$}} \\
    \cmidrule(lr){3-5}\cmidrule(lr){6-8}
    & \textbf{$\alpha/\gamma$} & 0 & 5 & 10 & 0 & 5 & 10\\
    \midrule
    \multirow{3}{*}{DM}
      & 0.10  & 0.111 & 0.366 & 0.139 & 0.120 & 0.798 & 0.187\\
      & 0.05  & 0.051 & 0.180 & 0.078 & 0.066 & 0.512 & 0.099\\
      & 0.01  & 0.008 & 0.047 & 0.020 & 0.022 & 0.131 & 0.031\\
    \midrule
    \multirow{3}{*}{PJ}
      & 0.10  & 0.084 & 0.442 & 0.172 & 0.096 & 0.791 & 0.288\\
      & 0.05  & 0.044 & 0.283 & 0.100 & 0.045 & 0.617 & 0.172\\
      & 0.01  & 0.009 & 0.112 & 0.035 & 0.010 & 0.322 & 0.055\\
    \bottomrule
  \end{tabular}
  \label{tab:SNsind1p2q1}
\end{table}

\begin{table}[htbp]
  \centering
  \setlength{\tabcolsep}{5pt}
  \renewcommand{\arraystretch}{1.1}
  \caption{Significance test when $\Psi(X,Z)=\sin(\gamma Z^{(1)})$, $X\sim N(0,1)$, $Z^{(1)}\sim N(0,1)$ and $Z^{(2)}\sim N(0,1)$}
  \begin{tabular}{llcccccccc}
    \toprule
    $c=0.5$ & & \multicolumn{3}{c}{\textbf{$n=200$}} & \multicolumn{3}{c}{\textbf{$n=400$}} \\
    \cmidrule(lr){3-5}\cmidrule(lr){6-8}
    & \textbf{$\alpha/\gamma$} & 0 & 5 & 10 & 0 & 5 & 10\\
    \midrule
    \multirow{3}{*}{DM}
      & 0.10  & 0.144 & 0.389 & 0.193 & 0.137 & 0.792 & 0.206\\
      & 0.05  & 0.070 & 0.218 & 0.092 & 0.084 & 0.486 & 0.101\\
      & 0.01  & 0.021 & 0.057 & 0.026 & 0.028 & 0.135 & 0.026\\
    \midrule
    \multirow{3}{*}{PJ}
      & 0.10  & 0.102 & 0.415 & 0.216 & 0.102 & 0.814 & 0.318\\
      & 0.05  & 0.048 & 0.254 & 0.132 & 0.050 & 0.570 & 0.196\\
      & 0.01  & 0.008 & 0.101 & 0.040 & 0.015 & 0.309 & 0.078\\
    \midrule
    $c=1$ & & \multicolumn{3}{c}{\textbf{$n=200$}} & \multicolumn{3}{c}{\textbf{$n=400$}} \\
    \cmidrule(lr){3-5}\cmidrule(lr){6-8}
    & \textbf{$\alpha/\gamma$} & 0 & 5 & 10 & 0 & 5 & 10\\
    \midrule
    \multirow{3}{*}{DM}
      & 0.10  & 0.121 & 0.364 & 0.164 & 0.131 & 0.792 & 0.183\\
      & 0.05  & 0.065 & 0.205 & 0.081 & 0.071 & 0.471 & 0.092\\
      & 0.01  & 0.013 & 0.049 & 0.021 & 0.018 & 0.124 & 0.029\\
    \midrule
    \multirow{3}{*}{PJ}
      & 0.10  & 0.107 & 0.340 & 0.215 & 0.105 & 0.665 & 0.308\\
      & 0.05  & 0.048 & 0.175 & 0.122 & 0.055 & 0.392 & 0.202\\
      & 0.01  & 0.011 & 0.070 & 0.044 & 0.010 & 0.131 & 0.062\\
    \midrule
    $c=2$ & & \multicolumn{3}{c}{\textbf{$n=200$}} & \multicolumn{3}{c}{\textbf{$n=400$}} \\
    \cmidrule(lr){3-5}\cmidrule(lr){6-8}
    & \textbf{$\alpha/\gamma$} & 0 & 5 & 10 & 0 & 5 & 10\\
    \midrule
    \multirow{3}{*}{DM}
      & 0.10  & 0.118 & 0.362 & 0.149 & 0.128 & 0.791 & 0.181\\
      & 0.05  & 0.066 & 0.190 & 0.079 & 0.065 & 0.477 & 0.086\\
      & 0.01  & 0.012 & 0.049 & 0.021 & 0.018 & 0.124 & 0.023\\
    \midrule
    \multirow{3}{*}{PJ}
      & 0.10  & 0.093 & 0.278 & 0.198 & 0.102 & 0.604 & 0.291\\
      & 0.05  & 0.045 & 0.138 & 0.109 & 0.052 & 0.348 & 0.174\\
      & 0.01  & 0.008 & 0.028 & 0.027 & 0.012 & 0.105 & 0.058\\
    \bottomrule
  \end{tabular}
  \label{tab:SNsind2p2q1}
\end{table}

\begin{table}[htbp]
  \centering
  \setlength{\tabcolsep}{5pt}
  \renewcommand{\arraystretch}{1.1}
  \caption{Significance test when $\Psi(X,Z)=\sin(\gamma Z^{(1)})$, $X\sim U(0,1)$, $Z^{(1)}\sim U(0,1)$ and $Z^{(2)}\sim U(0,1)$}
  \begin{tabular}{llcccccccc}
    \toprule
    $c=0.5$ & & \multicolumn{3}{c}{\textbf{$n=200$}} & \multicolumn{3}{c}{\textbf{$n=400$}} \\
    \cmidrule(lr){3-5}\cmidrule(lr){6-8}
    & \textbf{$\alpha/\gamma$} & 0 & 5 & 10 & 0 & 5 & 10\\
    \midrule
    \multirow{3}{*}{DM}
      & 0.10  & 0.147 & 0.969 & 0.615 & 0.129 & 1.000 & 0.866\\
      & 0.05  & 0.079 & 0.920 & 0.496 & 0.064 & 1.000 & 0.778\\
      & 0.01  & 0.025 & 0.601 & 0.289 & 0.018 & 0.992 & 0.562\\
    \midrule
    \multirow{3}{*}{PJ}
      & 0.10  & 0.105 & 0.967 & 0.611 & 0.110 & 1.000 & 0.851\\
      & 0.05  & 0.059 & 0.858 & 0.485 & 0.050 & 0.999 & 0.769\\
      & 0.01  & 0.014 & 0.583 & 0.288 & 0.014 & 0.983 & 0.579\\
    \midrule
    $c=1$ & & \multicolumn{3}{c}{\textbf{$n=200$}} & \multicolumn{3}{c}{\textbf{$n=400$}} \\
    \cmidrule(lr){3-5}\cmidrule(lr){6-8}
    & \textbf{$\alpha/\gamma$} & 0 & 5 & 10 & 0 & 5 & 10\\
    \midrule
    \multirow{3}{*}{DM}
      & 0.10  & 0.150 & 0.977 & 0.641 & 0.135 & 1.000 & 0.867\\
      & 0.05  & 0.078 & 0.928 & 0.506 & 0.067 & 1.000 & 0.794\\
      & 0.01  & 0.021 & 0.611 & 0.302 & 0.013 & 0.995 & 0.575\\
    \midrule
    \multirow{3}{*}{PJ}
      & 0.10  & 0.110 & 0.973 & 0.621 & 0.099 & 1.000 & 0.859\\
      & 0.05  & 0.061 & 0.905 & 0.487 & 0.048 & 1.000 & 0.782\\
      & 0.01  & 0.015 & 0.675 & 0.287 & 0.013 & 0.997 & 0.607\\
    \midrule
    $c=2$ & & \multicolumn{3}{c}{\textbf{$n=200$}} & \multicolumn{3}{c}{\textbf{$n=400$}} \\
    \cmidrule(lr){3-5}\cmidrule(lr){6-8}
    & \textbf{$\alpha/\gamma$} & 0 & 5 & 10 & 0 & 5 & 10\\
    \midrule
    \multirow{3}{*}{DM}
      & 0.10  & 0.160 & 0.983 & 0.648 & 0.150 & 1.000 & 0.874\\
      & 0.05  & 0.095 & 0.930 & 0.516 & 0.073 & 1.000 & 0.798\\
      & 0.01  & 0.025 & 0.643 & 0.299 & 0.017 & 0.994 & 0.592\\
    \midrule
    \multirow{3}{*}{PJ}
      & 0.10  & 0.114 & 0.987 & 0.597 & 0.104 & 1.000 & 0.847\\
      & 0.05  & 0.061 & 0.968 & 0.464 & 0.051 & 1.000 & 0.765\\
      & 0.01  & 0.016 & 0.827 & 0.260 & 0.015 & 0.998 & 0.562\\
    \bottomrule
  \end{tabular}
  \label{tab:SNsind3p2q1}
\end{table}

\begin{table}[htbp]
  \centering
  \setlength{\tabcolsep}{5pt}
  \renewcommand{\arraystretch}{1.1}
  \caption{Significance test when $\Psi(X,Z)=\sin(\gamma Z)$, $X^{(1)}=Z+U$, $X^{(2)}\sim N(0,1)$,$Z\sim N(0,1)$ and $U\sim N(0,1)$}
  \begin{tabular}{llcccccccc}
    \toprule
    $c=0.5$ & & \multicolumn{3}{c}{\textbf{$n=200$}} & \multicolumn{3}{c}{\textbf{$n=400$}} \\
    \cmidrule(lr){3-5}\cmidrule(lr){6-8}
    & \textbf{$\alpha/\gamma$} & 0 & 5 & 10 & 0 & 5 & 10\\
    \midrule
    \multirow{3}{*}{DM}
      & 0.10  & 0.598 & 0.669 & 0.742 & 0.497 & 0.764 & 0.537\\
      & 0.05  & 0.425 & 0.493 & 0.528 & 0.328 & 0.597 & 0.393\\
      & 0.01  & 0.200 & 0.254 & 0.234 & 0.134 & 0.253 & 0.144\\
    \midrule
    \multirow{3}{*}{PJ}
      & 0.10  & 0.213 & 0.338 & 0.274 & 0.192 & 0.547 & 0.326\\
      & 0.05  & 0.099 & 0.204 & 0.168 & 0.110 & 0.388 & 0.195\\
      & 0.01  & 0.025 & 0.069 & 0.044 & 0.033 & 0.171 & 0.067\\
    \midrule
    $c=1$ & & \multicolumn{3}{c}{\textbf{$n=200$}} & \multicolumn{3}{c}{\textbf{$n=400$}} \\
    \cmidrule(lr){3-5}\cmidrule(lr){6-8}
    & \textbf{$\alpha/\gamma$} & 0 & 5 & 10 & 0 & 5 & 10\\
    \midrule
    \multirow{3}{*}{DM}
      & 0.10  & 0.398 & 0.567 & 0.444 & 0.293 & 0.760 & 0.399\\
      & 0.05  & 0.231 & 0.396 & 0.237 & 0.163 & 0.561 & 0.215\\
      & 0.01  & 0.074 & 0.157 & 0.090 & 0.049 & 0.192 & 0.062\\
    \midrule
    \multirow{3}{*}{PJ}
      & 0.10  & 0.153 & 0.361 & 0.225 & 0.143 & 0.577 & 0.287\\
      & 0.05  & 0.083 & 0.231 & 0.119 & 0.069 & 0.434 & 0.171\\
      & 0.01  & 0.015 & 0.072 & 0.041 & 0.014 & 0.194 & 0.054\\
    \midrule
    $c=2$ & & \multicolumn{3}{c}{\textbf{$n=200$}} & \multicolumn{3}{c}{\textbf{$n=400$}} \\
    \cmidrule(lr){3-5}\cmidrule(lr){6-8}
    & \textbf{$\alpha/\gamma$} & 0 & 5 & 10 & 0 & 5 & 10\\
    \midrule
    \multirow{3}{*}{DM}
      & 0.10  & 0.278 & 0.499 & 0.315 & 0.238 & 0.778 & 0.359\\
      & 0.05  & 0.155 & 0.336 & 0.160 & 0.125 & 0.588 & 0.191\\
      & 0.01  & 0.033 & 0.114 & 0.044 & 0.037 & 0.185 & 0.042\\
    \midrule
    \multirow{3}{*}{PJ}
      & 0.10  & 0.102 & 0.363 & 0.173 & 0.109 & 0.624 & 0.260\\
      & 0.05  & 0.048 & 0.225 & 0.086 & 0.044 & 0.460 & 0.164\\
      & 0.01  & 0.011 & 0.082 & 0.016 & 0.009 & 0.231 & 0.049\\
    \bottomrule
  \end{tabular}
  \label{tab:SNsind1p1q2}
\end{table}

\begin{table}[htbp]
  \centering
  \setlength{\tabcolsep}{5pt}
  \renewcommand{\arraystretch}{1.1}
  \caption{Significance test when $\Psi(X,Z)=\sin(\gamma Z)$, $X^{(1)}\sim N(0,1)$, $X^{(2)}\sim N(0,1)$, and $Z\sim N(0,1)$}
  \begin{tabular}{llcccccccc}
    \toprule
    $c=0.5$ & & \multicolumn{3}{c}{\textbf{$n=200$}} & \multicolumn{3}{c}{\textbf{$n=400$}} \\
    \cmidrule(lr){3-5}\cmidrule(lr){6-8}
    & \textbf{$\alpha/\gamma$} & 0 & 5 & 10 & 0 & 5 & 10\\
    \midrule
    \multirow{3}{*}{DM}
      & 0.10  & 0.618 & 0.615 & 0.644 & 0.431 & 0.657 & 0.530\\
      & 0.05  & 0.421 & 0.440 & 0.431 & 0.298 & 0.500 & 0.321\\
      & 0.01  & 0.169 & 0.190 & 0.180 & 0.115 & 0.266 & 0.132\\
    \midrule
    \multirow{3}{*}{PJ}
      & 0.10  & 0.225 & 0.332 & 0.281 & 0.213 & 0.521 & 0.355\\
      & 0.05  & 0.123 & 0.189 & 0.162 & 0.105 & 0.339 & 0.232\\
      & 0.01  & 0.032 & 0.045 & 0.047 & 0.043 & 0.118 & 0.102\\
    \midrule
    $c=1$ & & \multicolumn{3}{c}{\textbf{$n=200$}} & \multicolumn{3}{c}{\textbf{$n=400$}} \\
    \cmidrule(lr){3-5}\cmidrule(lr){6-8}
    & \textbf{$\alpha/\gamma$} & 0 & 5 & 10 & 0 & 5 & 10\\
    \midrule
    \multirow{3}{*}{DM}
      & 0.10  & 0.314 & 0.497 & 0.379 & 0.214 & 0.621 & 0.321\\
      & 0.05  & 0.183 & 0.314 & 0.204 & 0.133 & 0.440 & 0.176\\
      & 0.01  & 0.064 & 0.111 & 0.069 & 0.037 & 0.211 & 0.053\\
    \midrule
    \multirow{3}{*}{PJ}
      & 0.10  & 0.159 & 0.312 & 0.230 & 0.148 & 0.508 & 0.334\\
      & 0.05  & 0.092 & 0.161 & 0.145 & 0.071 & 0.298 & 0.225\\
      & 0.01  & 0.021 & 0.052 & 0.043 & 0.014 & 0.100 & 0.082\\
    \midrule
    $c=2$ & & \multicolumn{3}{c}{\textbf{$n=200$}} & \multicolumn{3}{c}{\textbf{$n=400$}} \\
    \cmidrule(lr){3-5}\cmidrule(lr){6-8}
    & \textbf{$\alpha/\gamma$} & 0 & 5 & 10 & 0 & 5 & 10\\
    \midrule
    \multirow{3}{*}{DM}
      & 0.10  & 0.203 & 0.421 & 0.275 & 0.167 & 0.627 & 0.264\\
      & 0.05  & 0.100 & 0.242 & 0.139 & 0.084 & 0.437 & 0.127\\
      & 0.01  & 0.027 & 0.075 & 0.031 & 0.024 & 0.176 & 0.034\\
    \midrule
    \multirow{3}{*}{PJ}
      & 0.10  & 0.104 & 0.269 & 0.198 & 0.102 & 0.534 & 0.332\\
      & 0.05  & 0.046 & 0.150 & 0.122 & 0.046 & 0.358 & 0.193\\
      & 0.01  & 0.012 & 0.044 & 0.032 & 0.010 & 0.125 & 0.069\\
    \bottomrule
  \end{tabular}
  \label{tab:SNsind2p1q2}
\end{table}

\begin{table}[htbp]
  \centering
  \setlength{\tabcolsep}{5pt}
  \renewcommand{\arraystretch}{1.1}
  \caption{Significance test when $\Psi(X,Z)=\sin(\gamma Z)$, $X^{(1)}\sim U(0,1)$, $X^{(2)}\sim U(0,1)$ and $Z\sim U(0,1)$}
  \begin{tabular}{llcccccccc}
    \toprule
    $c=0.5$ & & \multicolumn{3}{c}{\textbf{$n=200$}} & \multicolumn{3}{c}{\textbf{$n=400$}} \\
    \cmidrule(lr){3-5}\cmidrule(lr){6-8}
    & \textbf{$\alpha/\gamma$} & 0 & 5 & 10 & 0 & 5 & 10\\
    \midrule
    \multirow{3}{*}{DM}
      & 0.10  & 0.245 & 0.963 & 0.471 & 0.209 & 0.998 & 0.637\\
      & 0.05  & 0.144 & 0.888 & 0.320 & 0.111 & 0.998 & 0.510\\
      & 0.01  & 0.049 & 0.619 & 0.151 & 0.036 & 0.986 & 0.287\\
    \midrule
    \multirow{3}{*}{PJ}
      & 0.10  & 0.133 & 0.922 & 0.369 & 0.120 & 1.000 & 0.600\\
      & 0.05  & 0.069 & 0.817 & 0.260 & 0.077 & 1.000 & 0.484\\
      & 0.01  & 0.019 & 0.544 & 0.112 & 0.018 & 0.977 & 0.294\\
    \midrule
    $c=1$ & & \multicolumn{3}{c}{\textbf{$n=200$}} & \multicolumn{3}{c}{\textbf{$n=400$}} \\
    \cmidrule(lr){3-5}\cmidrule(lr){6-8}
    & \textbf{$\alpha/\gamma$} & 0 & 5 & 10 & 0 & 5 & 10\\
    \midrule
    \multirow{3}{*}{DM}
      & 0.10  & 0.201 & 0.985 & 0.501 & 0.196 & 1.000 & 0.671\\
      & 0.05  & 0.107 & 0.950 & 0.345 & 0.092 & 1.000 & 0.567\\
      & 0.01  & 0.033 & 0.692 & 0.159 & 0.022 & 0.990 & 0.330\\
    \midrule
    \multirow{3}{*}{PJ}
      & 0.10  & 0.106 & 0.968 & 0.361 & 0.115 & 1.000 & 0.633\\
      & 0.05  & 0.052 & 0.907 & 0.249 & 0.058 & 1.000 & 0.512\\
      & 0.01  & 0.012 & 0.696 & 0.110 & 0.012 & 0.994 & 0.307\\
    \midrule
    $c=2$ & & \multicolumn{3}{c}{\textbf{$n=200$}} & \multicolumn{3}{c}{\textbf{$n=400$}} \\
    \cmidrule(lr){3-5}\cmidrule(lr){6-8}
    & \textbf{$\alpha/\gamma$} & 0 & 5 & 10 & 0 & 5 & 10\\
    \midrule
    \multirow{3}{*}{DM}
      & 0.10  & 0.408 & 0.988 & 0.610 & 0.365 & 1.000 & 0.815\\
      & 0.05  & 0.230 & 0.955 & 0.463 & 0.215 & 1.000 & 0.721\\
      & 0.01  & 0.082 & 0.802 & 0.245 & 0.077 & 0.997 & 0.511\\
    \midrule
    \multirow{3}{*}{PJ}
      & 0.10  & 0.101 & 0.968 & 0.329 & 0.107 & 1.000 & 0.571\\
      & 0.05  & 0.045 & 0.905 & 0.224 & 0.052 & 1.000 & 0.459\\
      & 0.01  & 0.010 & 0.709 & 0.087 & 0.010 & 0.990 & 0.241\\
    \bottomrule
  \end{tabular}
  \label{tab:SNsind3p1q2}
\end{table}

Tables \ref{tab:SNsind1}–\ref{tab:SNsind3} report power results obtained by replacing $\Psi(x,z)=\sin(\gamma z^{(1)})$ with $\Psi(x,z)=\exp(z^{(1)})$ in the above designs to verify that the tests retain good powers under other types of alternatives and to provide further evidence on the admissibility of ICM-type tests discussed in \cite{bierens1997asymptotic}. For notational convenience, we label each design in the tables by the values of $p$ and $q$ together with the distributions of the first components, rather than repeating the full distribution in the table captions.
\begin{table}[htbp]
  \centering
  \setlength{\tabcolsep}{5pt}
  \renewcommand{\arraystretch}{1.1}
  \caption{Power of significance test when $\Psi(X,Z)=\exp(\gamma Z^{(1)})$, $X^{(1)}=Z^{(1)}+U$, $Z^{(1)}\sim N(0,1)$ and $U\sim N(0,1)$}
  \begin{tabular}{llcccccccc}
    \toprule
    $p=1,q=1$ & & \multicolumn{3}{c}{\textbf{$n=200$}} & \multicolumn{3}{c}{\textbf{$n=400$}} \\
    \cmidrule(lr){3-5}\cmidrule(lr){6-8}
    & \textbf{$\alpha/c$} & 0.5 & 1 & 2 & 0.5 & 1 & 2\\
    \midrule
    \multirow{3}{*}{DM}
      & 0.10  & 0.998 & 1.000 & 1.000 & 1.000 & 1.000 & 1.000\\
      & 0.05  & 0.994 & 0.998 & 0.999 & 1.000 & 1.000 & 1.000\\
      & 0.01  & 0.961 & 0.972 & 0.986 & 0.999 & 1.000 & 1.000\\
    \midrule
    \multirow{3}{*}{PJ}
      & 0.10  & 0.992 & 0.992 & 0.997 & 1.000 & 1.000 & 1.000\\
      & 0.05  & 0.981 & 0.981 & 0.984 & 1.000 & 1.000 & 1.000\\
      & 0.01  & 0.913 & 0.911 & 0.902 & 1.000 & 0.999 & 0.998\\
    \midrule
    $p=2,q=1$ & & \multicolumn{3}{c}{\textbf{$n=200$}} & \multicolumn{3}{c}{\textbf{$n=400$}} \\
    \cmidrule(lr){3-5}\cmidrule(lr){6-8}
    & \textbf{$\alpha/c$} & 0.5 & 1 & 2 & 0.5 & 1 & 2\\
    \midrule
    \multirow{3}{*}{DM}
      & 0.10  & 0.801 & 0.800 & 0.817 & 0.991 & 0.992 & 0.991\\
      & 0.05  & 0.607 & 0.618 & 0.632 & 0.952 & 0.964 & 0.969\\
      & 0.01  & 0.296 & 0.278 & 0.301 & 0.690 & 0.711 & 0.724\\
    \midrule
    \multirow{3}{*}{PJ}
      & 0.10  & 0.821 & 0.844 & 0.872 & 0.996 & 0.998 & 0.998\\
      & 0.05  & 0.670 & 0.725 & 0.765 & 0.984 & 0.993 & 0.993\\
      & 0.01  & 0.399 & 0.442 & 0.510 & 0.903 & 0.943 & 0.963\\
    \midrule
    $p=1,q=2$ & & \multicolumn{3}{c}{\textbf{$n=200$}} & \multicolumn{3}{c}{\textbf{$n=400$}} \\
    \cmidrule(lr){3-5}\cmidrule(lr){6-8}
    & \textbf{$\alpha/c$} & 0.5 & 1 & 2 & 0.5 & 1 & 2\\
    \midrule
    \multirow{3}{*}{DM}
      & 0.10  & 0.922 & 0.969 & 0.984 & 0.988 & 0.989 & 0.999\\
      & 0.05  & 0.806 & 0.923 & 0.971 & 0.972 & 0.984 & 0.998\\
      & 0.01  & 0.530 & 0.769 & 0.884 & 0.898 & 0.974 & 0.997\\
    \midrule
    \multirow{3}{*}{PJ}
      & 0.10  & 0.788 & 0.948 & 0.978 & 0.994 & 1.000 & 1.000\\
      & 0.05  & 0.598 & 0.867 & 0.948 & 0.980 & 1.000 & 1.000\\
      & 0.01  & 0.291 & 0.663 & 0.777 & 0.862 & 0.983 & 0.999\\
    \bottomrule
  \end{tabular}
  \label{tab:SNsind1}
\end{table}

\begin{table}[htbp]
  \centering
  \setlength{\tabcolsep}{5pt}
  \renewcommand{\arraystretch}{1.1}
  \caption{Power of significance test when $\Psi(X,Z)=\exp(\gamma Z^{(1)})$, $X^{(1)}\sim N(0,1)$ and $Z^{(1)}\sim N(0,1)$}
  \begin{tabular}{llcccccccc}
    \toprule
    $p=1,q=1$ & & \multicolumn{3}{c}{\textbf{$n=200$}} & \multicolumn{3}{c}{\textbf{$n=400$}} \\
    \cmidrule(lr){3-5}\cmidrule(lr){6-8}
    & \textbf{$\alpha/c$} & 0.5 & 1 & 2 & 0.5 & 1 & 2\\
    \midrule
    \multirow{3}{*}{DM}
      & 0.10  & 0.495 & 0.494 & 0.501 & 0.927 & 0.930 & 0.943\\
      & 0.05  & 0.253 & 0.262 & 0.272 & 0.693 & 0.718 & 0.743\\
      & 0.01  & 0.075 & 0.073 & 0.072 & 0.255 & 0.247 & 0.252\\
    \midrule
    \multirow{3}{*}{PJ}
      & 0.10  & 0.571 & 0.546 & 0.604 & 0.877 & 0.905 & 0.914\\
      & 0.05  & 0.391 & 0.391 & 0.420 & 0.744 & 0.752 & 0.795\\
      & 0.01  & 0.186 & 0.184 & 0.197 & 0.461 & 0.477 & 0.512\\
    \midrule
    $p=2,q=1$ & & \multicolumn{3}{c}{\textbf{$n=200$}} & \multicolumn{3}{c}{\textbf{$n=400$}} \\
    \cmidrule(lr){3-5}\cmidrule(lr){6-8}
    & \textbf{$\alpha/c$} & 0.5 & 1 & 2 & 0.5 & 1 & 2\\
    \midrule
    \multirow{3}{*}{DM}
      & 0.10  & 0.447 & 0.432 & 0.409 & 0.795 & 0.788 & 0.775\\
      & 0.05  & 0.299 & 0.261 & 0.221 & 0.650 & 0.636 & 0.626\\
      & 0.01  & 0.090 & 0.072 & 0.064 & 0.259 & 0.236 & 0.215\\
    \midrule
    \multirow{3}{*}{PJ}
      & 0.10  & 0.472 & 0.468 & 0.494 & 0.852 & 0.843 & 0.845\\
      & 0.05  & 0.319 & 0.324 & 0.338 & 0.685 & 0.671 & 0.681\\
      & 0.01  & 0.137 & 0.130 & 0.137 & 0.397 & 0.396 & 0.391\\
    \midrule
    $p=1,q=2$ & & \multicolumn{3}{c}{\textbf{$n=200$}} & \multicolumn{3}{c}{\textbf{$n=400$}} \\
    \cmidrule(lr){3-5}\cmidrule(lr){6-8}
    & \textbf{$\alpha/c$} & 0.5 & 1 & 2 & 0.5 & 1 & 2\\
    \midrule
    \multirow{3}{*}{DM}
      & 0.10  & 0.700 & 0.588 & 0.537 & 0.803 & 0.783 & 0.839\\
      & 0.05  & 0.525 & 0.351 & 0.288 & 0.578 & 0.563 & 0.556\\
      & 0.01  & 0.254 & 0.138 & 0.075 & 0.240 & 0.181 & 0.165\\
    \midrule
    \multirow{3}{*}{PJ}
      & 0.10  & 0.421 & 0.409 & 0.442 & 0.647 & 0.690 & 0.753\\
      & 0.05  & 0.270 & 0.274 & 0.301 & 0.464 & 0.539 & 0.576\\
      & 0.01  & 0.110 & 0.106 & 0.125 & 0.220 & 0.266 & 0.297\\
    \bottomrule
  \end{tabular}
  \label{tab:SNsind2}
\end{table}

\begin{table}[htbp]
  \centering
  \setlength{\tabcolsep}{5pt}
  \renewcommand{\arraystretch}{1.1}
  \caption{Power of significance test when $\Psi(X,Z)=\exp(\gamma Z^{(1)})$, $X^{(1)}\sim U(0,1)$ and $Z^{(1)}\sim U(0,1)$}
  \begin{tabular}{llcccccccc}
    \toprule
    $p=1,q=1$ & & \multicolumn{3}{c}{\textbf{$n=200$}} & \multicolumn{3}{c}{\textbf{$n=400$}} \\
    \cmidrule(lr){3-5}\cmidrule(lr){6-8}
    & \textbf{$\alpha/c$} & 0.5 & 1 & 2 & 0.5 & 1 & 2\\
    \midrule
    \multirow{3}{*}{DM}
      & 0.10  & 1.000 & 1.000 & 1.000 & 1.000 & 1.000 & 1.000\\
      & 0.05  & 1.000 & 1.000 & 1.000 & 1.000 & 1.000 & 1.000\\
      & 0.01  & 1.000 & 1.000 & 1.000 & 1.000 & 1.000 & 1.000\\
    \midrule
    \multirow{3}{*}{PJ}
      & 0.10  & 1.000 & 1.000 & 1.000 & 1.000 & 1.000 & 1.000\\
      & 0.05  & 1.000 & 1.000 & 1.000 & 1.000 & 1.000 & 1.000\\
      & 0.01  & 1.000 & 1.000 & 1.000 & 1.000 & 1.000 & 1.000\\
    \midrule
    $p=2,q=1$ & & \multicolumn{3}{c}{\textbf{$n=200$}} & \multicolumn{3}{c}{\textbf{$n=400$}} \\
    \cmidrule(lr){3-5}\cmidrule(lr){6-8}
    & \textbf{$\alpha/c$} & 0.5 & 1 & 2 & 0.5 & 1 & 2\\
    \midrule
    \multirow{3}{*}{DM}
      & 0.10  & 1.000 & 1.000 & 1.000 & 1.000 & 1.000 & 1.000\\
      & 0.05  & 1.000 & 1.000 & 1.000 & 1.000 & 1.000 & 1.000\\
      & 0.01  & 1.000 & 1.000 & 1.000 & 1.000 & 1.000 & 1.000\\
    \midrule
    \multirow{3}{*}{PJ}
      & 0.10  & 1.000 & 1.000 & 1.000 & 1.000 & 1.000 & 1.000\\
      & 0.05  & 1.000 & 1.000 & 1.000 & 1.000 & 1.000 & 1.000\\
      & 0.01  & 1.000 & 1.000 & 1.000 & 1.000 & 1.000 & 1.000\\
    \midrule
    $p=1,q=2$ & & \multicolumn{3}{c}{\textbf{$n=200$}} & \multicolumn{3}{c}{\textbf{$n=400$}} \\
    \cmidrule(lr){3-5}\cmidrule(lr){6-8}
    & \textbf{$\alpha/c$} & 0.5 & 1 & 2 & 0.5 & 1 & 2\\
    \midrule
    \multirow{3}{*}{DM}
      & 0.10  & 0.998 & 1.000 & 1.000 & 0.998 & 1.000 & 1.000\\
      & 0.05  & 0.998 & 1.000 & 1.000 & 0.998 & 1.000 & 1.000\\
      & 0.01  & 0.998 & 1.000 & 0.999 & 0.998 & 1.000 & 1.000\\
    \midrule
    \multirow{3}{*}{PJ}
      & 0.10  & 1.000 & 1.000 & 1.000 & 1.000 & 1.000 & 1.000\\
      & 0.05  & 1.000 & 1.000 & 1.000 & 1.000 & 1.000 & 1.000\\
      & 0.01  & 0.999 & 0.998 & 0.998 & 1.000 & 1.000 & 1.000\\
    \bottomrule
  \end{tabular}
  \label{tab:SNsind3}
\end{table}

Still, as a further complement to the results reported in the main text, tables \ref{tab:CIsind1p2q1}–\ref{tab:CIsind3} report the corresponding results when these DGPs are used in the conditional independence testing framework. The only difference relative to the previous designs is that the data generating process for $Y$ is modified so that the variable whose significance is being tested enters through the variance rather than the mean, which is
\begin{align*}
    Y=1+\left[X+\Psi(X,Z)\right]\epsilon.
\end{align*}
Since the significance tests for the conditional mean considered above can themselves be viewed as tests of conditional independence, it is not surprising that the qualitative conclusions are similar.
\begin{table}[htbp]
  \centering
  \setlength{\tabcolsep}{5pt}
  \renewcommand{\arraystretch}{1.1}
  \caption{Testing conditional independence when $\Psi(X,Z)=\sin(\gamma Z^{(1)})$, $X=Z^{(1)}+U$, $Z^{(1)}\sim N(0,1)$, $Z^{(2)}\sim N(0,1)$ and $U\sim N(0,1)$}
  \begin{tabular}{llcccccccc}
    \toprule
    $c=0.5$ & & \multicolumn{3}{c}{\textbf{$n=200$}} & \multicolumn{3}{c}{\textbf{$n=400$}} \\
    \cmidrule(lr){3-5}\cmidrule(lr){6-8}
    & \textbf{$\alpha/\gamma$} & 0 & 5 & 10 & 0 & 5 & 10\\
    \midrule
    \multirow{3}{*}{DM}
      & 0.10  & 0.166 & 0.151 & 0.138 & 0.161 & 0.161 & 0.141\\
      & 0.05  & 0.093 & 0.086 & 0.073 & 0.084 & 0.079 & 0.076\\
      & 0.01  & 0.038 & 0.027 & 0.026 & 0.026 & 0.025 & 0.023\\
    \midrule
    \multirow{3}{*}{PJ}
      & 0.10  & 0.093 & 0.120 & 0.113 & 0.099 & 0.122 & 0.120\\
      & 0.05  & 0.053 & 0.065 & 0.064 & 0.054 & 0.067 & 0.058\\
      & 0.01  & 0.015 & 0.024 & 0.020 & 0.013 & 0.016 & 0.015\\
    \midrule
    $c=1$ & & \multicolumn{3}{c}{\textbf{$n=200$}} & \multicolumn{3}{c}{\textbf{$n=400$}} \\
    \cmidrule(lr){3-5}\cmidrule(lr){6-8}
    & \textbf{$\alpha/\gamma$} & 0 & 5 & 10 & 0 & 5 & 10\\
    \midrule
    \multirow{3}{*}{DM}
      & 0.10  & 0.140 & 0.126 & 0.125 & 0.135 & 0.136 & 0.124\\
      & 0.05  & 0.094 & 0.074 & 0.066 & 0.071 & 0.066 & 0.068\\
      & 0.01  & 0.032 & 0.027 & 0.021 & 0.022 & 0.015 & 0.017\\
    \midrule
    \multirow{3}{*}{PJ}
      & 0.10  & 0.100 & 0.108 & 0.100 & 0.098 & 0.095 & 0.094\\
      & 0.05  & 0.053 & 0.060 & 0.055 & 0.049 & 0.053 & 0.044\\
      & 0.01  & 0.018 & 0.019 & 0.017 & 0.010 & 0.026 & 0.010\\
    \midrule
    $c=2$ & & \multicolumn{3}{c}{\textbf{$n=200$}} & \multicolumn{3}{c}{\textbf{$n=400$}} \\
    \cmidrule(lr){3-5}\cmidrule(lr){6-8}
    & \textbf{$\alpha/\gamma$} & 0 & 5 & 10 & 0 & 5 & 10\\
    \midrule
    \multirow{3}{*}{DM}
      & 0.10  & 0.144 & 0.126 & 0.115 & 0.135 & 0.122 & 0.123\\
      & 0.05  & 0.093 & 0.066 & 0.066 & 0.073 & 0.065 & 0.062\\
      & 0.01  & 0.030 & 0.020 & 0.015 & 0.017 & 0.012 & 0.014\\
    \midrule
    \multirow{3}{*}{PJ}
      & 0.10  & 0.098 & 0.103 & 0.102 & 0.087 & 0.100 & 0.096\\
      & 0.05  & 0.046 & 0.055 & 0.058 & 0.047 & 0.053 & 0.047\\
      & 0.01  & 0.013 & 0.014 & 0.014 & 0.013 & 0.012 & 0.010\\
    \bottomrule
  \end{tabular}
  \label{tab:CIsind1p2q1}
\end{table}

\begin{table}[htbp]
  \centering
  \setlength{\tabcolsep}{5pt}
  \renewcommand{\arraystretch}{1.1}
  \caption{Testing conditional independence when $\Psi(X,Z)=\sin(\gamma Z^{(1)})$, $X\sim N(0,1)$, $Z^{(1)}\sim N(0,1)$ and $Z^{(2)}\sim N(0,1)$}
  \begin{tabular}{llcccccccc}
    \toprule
    $c=0.5$ & & \multicolumn{3}{c}{\textbf{$n=200$}} & \multicolumn{3}{c}{\textbf{$n=400$}} \\
    \cmidrule(lr){3-5}\cmidrule(lr){6-8}
    & \textbf{$\alpha/\gamma$} & 0 & 5 & 10 & 0 & 5 & 10\\
    \midrule
    \multirow{3}{*}{DM}
      & 0.10  & 0.185 & 0.183 & 0.164 & 0.167 & 0.139 & 0.140\\
      & 0.05  & 0.101 & 0.119 & 0.088 & 0.093 & 0.074 & 0.080\\
      & 0.01  & 0.036 & 0.038 & 0.023 & 0.026 & 0.019 & 0.019\\
    \midrule
    \multirow{3}{*}{PJ}
      & 0.10  & 0.092 & 0.137 & 0.127 & 0.114 & 0.119 & 0.106\\
      & 0.05  & 0.051 & 0.078 & 0.061 & 0.056 & 0.060 & 0.056\\
      & 0.01  & 0.017 & 0.031 & 0.016 & 0.019 & 0.014 & 0.010\\
    \midrule
    $c=1$ & & \multicolumn{3}{c}{\textbf{$n=200$}} & \multicolumn{3}{c}{\textbf{$n=400$}} \\
    \cmidrule(lr){3-5}\cmidrule(lr){6-8}
    & \textbf{$\alpha/\gamma$} & 0 & 5 & 10 & 0 & 5 & 10\\
    \midrule
    \multirow{3}{*}{DM}
      & 0.10  & 0.156 & 0.152 & 0.126 & 0.144 & 0.111 & 0.122\\
      & 0.05  & 0.089 & 0.088 & 0.070 & 0.078 & 0.060 & 0.061\\
      & 0.01  & 0.029 & 0.032 & 0.016 & 0.024 & 0.018 & 0.017\\
    \midrule
    \multirow{3}{*}{PJ}
      & 0.10  & 0.096 & 0.122 & 0.102 & 0.100 & 0.106 & 0.097\\
      & 0.05  & 0.046 & 0.072 & 0.055 & 0.049 & 0.060 & 0.041\\
      & 0.01  & 0.015 & 0.026 & 0.016 & 0.012 & 0.016 & 0.013\\
    \midrule
    $c=2$ & & \multicolumn{3}{c}{\textbf{$n=200$}} & \multicolumn{3}{c}{\textbf{$n=400$}} \\
    \cmidrule(lr){3-5}\cmidrule(lr){6-8}
    & \textbf{$\alpha/\gamma$} & 0 & 5 & 10 & 0 & 5 & 10\\
    \midrule
    \multirow{3}{*}{DM}
      & 0.10  & 0.155 & 0.136 & 0.116 & 0.141 & 0.104 & 0.105\\
      & 0.05  & 0.088 & 0.076 & 0.060 & 0.079 & 0.053 & 0.056\\
      & 0.01  & 0.029 & 0.025 & 0.016 & 0.024 & 0.015 & 0.012\\
    \midrule
    \multirow{3}{*}{PJ}
      & 0.10  & 0.100 & 0.108 & 0.088 & 0.101 & 0.108 & 0.085\\
      & 0.05  & 0.061 & 0.065 & 0.039 & 0.052 & 0.053 & 0.042\\
      & 0.01  & 0.014 & 0.020 & 0.011 & 0.012 & 0.016 & 0.012\\
    \bottomrule
  \end{tabular}
  \label{tab:CIsind2p2q1}
\end{table}

\begin{table}[htbp]
  \centering
  \setlength{\tabcolsep}{5pt}
  \renewcommand{\arraystretch}{1.1}
  \caption{Testing conditional independence when $\Psi(X,Z)=\sin(\gamma Z^{(1)})$, $X\sim U(0,1)$, $Z^{(1)}\sim U(0,1)$ and $Z^{(2)}\sim U(0,1)$}
  \begin{tabular}{llcccccccc}
    \toprule
    $c=0.5$ & & \multicolumn{3}{c}{\textbf{$n=200$}} & \multicolumn{3}{c}{\textbf{$n=400$}} \\
    \cmidrule(lr){3-5}\cmidrule(lr){6-8}
    & \textbf{$\alpha/\gamma$} & 0 & 5 & 10 & 0 & 5 & 10\\
    \midrule
    \multirow{3}{*}{DM}
      & 0.10  & 0.169 & 0.926 & 0.157 & 0.152 & 1.000 & 0.239\\
      & 0.05  & 0.095 & 0.835 & 0.081 & 0.082 & 0.998 & 0.136\\
      & 0.01  & 0.027 & 0.557 & 0.023 & 0.019 & 0.958 & 0.036\\
    \midrule
    \multirow{3}{*}{PJ}
      & 0.10  & 0.107 & 0.924 & 0.133 & 0.100 & 1.000 & 0.196\\
      & 0.05  & 0.055 & 0.798 & 0.067 & 0.054 & 0.995 & 0.115\\
      & 0.01  & 0.017 & 0.409 & 0.023 & 0.009 & 0.930 & 0.039\\
    \midrule
    $c=1$ & & \multicolumn{3}{c}{\textbf{$n=200$}} & \multicolumn{3}{c}{\textbf{$n=400$}} \\
    \cmidrule(lr){3-5}\cmidrule(lr){6-8}
    & \textbf{$\alpha/\gamma$} & 0 & 5 & 10 & 0 & 5 & 10\\
    \midrule
    \multirow{3}{*}{DM}
      & 0.10  & 0.157 & 0.930 & 0.146 & 0.154 & 1.000 & 0.228\\
      & 0.05  & 0.094 & 0.836 & 0.081 & 0.084 & 0.999 & 0.134\\
      & 0.01  & 0.027 & 0.554 & 0.019 & 0.018 & 0.962 & 0.035\\
    \midrule
    \multirow{3}{*}{PJ}
      & 0.10  & 0.093 & 0.918 & 0.138 & 0.096 & 1.000 & 0.216\\
      & 0.05  & 0.050 & 0.785 & 0.076 & 0.047 & 0.997 & 0.131\\
      & 0.01  & 0.017 & 0.424 & 0.023 & 0.013 & 0.941 & 0.042\\
    \midrule
    $c=2$ & & \multicolumn{3}{c}{\textbf{$n=200$}} & \multicolumn{3}{c}{\textbf{$n=400$}} \\
    \cmidrule(lr){3-5}\cmidrule(lr){6-8}
    & \textbf{$\alpha/\gamma$} & 0 & 5 & 10 & 0 & 5 & 10\\
    \midrule
    \multirow{3}{*}{DM}
      & 0.10  & 0.206 & 0.924 & 0.159 & 0.208 & 1.000 & 0.233\\
      & 0.05  & 0.115 & 0.834 & 0.076 & 0.114 & 0.997 & 0.134\\
      & 0.01  & 0.036 & 0.547 & 0.021 & 0.024 & 0.962 & 0.039\\
    \midrule
    \multirow{3}{*}{PJ}
      & 0.10  & 0.103 & 0.904 & 0.163 & 0.100 & 1.000 & 0.283\\
      & 0.05  & 0.056 & 0.772 & 0.090 & 0.052 & 0.996 & 0.175\\
      & 0.01  & 0.018 & 0.381 & 0.029 & 0.011 & 0.917 & 0.049\\
    \bottomrule
  \end{tabular}
  \label{tab:CIsind3p2q1}
\end{table}

\begin{table}[htbp]
  \centering
  \setlength{\tabcolsep}{5pt}
  \renewcommand{\arraystretch}{1.1}
  \caption{Testing conditional independence when $\Psi(X,Z)=\sin(\gamma Z)$, $X^{(1)}=Z+U$, $X^{(2)}\sim N(0,1)$,$Z\sim N(0,1)$ and $U\sim N(0,1)$}
  \begin{tabular}{llcccccccc}
    \toprule
    $c=0.5$ & & \multicolumn{3}{c}{\textbf{$n=200$}} & \multicolumn{3}{c}{\textbf{$n=400$}} \\
    \cmidrule(lr){3-5}\cmidrule(lr){6-8}
    & \textbf{$\alpha/\gamma$} & 0 & 5 & 10 & 0 & 5 & 10\\
    \midrule
    \multirow{3}{*}{DM}
      & 0.10  & 0.905 & 0.508 & 0.504 & 0.831 & 0.464 & 0.426\\
      & 0.05  & 0.781 & 0.335 & 0.345 & 0.667 & 0.304 & 0.274\\
      & 0.01  & 0.476 & 0.127 & 0.117 & 0.332 & 0.122 & 0.124\\
    \midrule
    \multirow{3}{*}{PJ}
      & 0.10  & 0.181 & 0.270 & 0.278 & 0.183 & 0.278 & 0.267\\
      & 0.05  & 0.098 & 0.126 & 0.130 & 0.105 & 0.155 & 0.142\\
      & 0.01  & 0.017 & 0.036 & 0.032 & 0.031 & 0.052 & 0.044\\
    \midrule
    $c=1$ & & \multicolumn{3}{c}{\textbf{$n=200$}} & \multicolumn{3}{c}{\textbf{$n=400$}} \\
    \cmidrule(lr){3-5}\cmidrule(lr){6-8}
    & \textbf{$\alpha/\gamma$} & 0 & 5 & 10 & 0 & 5 & 10\\
    \midrule
    \multirow{3}{*}{DM}
      & 0.10  & 0.712 & 0.319 & 0.274 & 0.555 & 0.236 & 0.231\\
      & 0.05  & 0.547 & 0.185 & 0.159 & 0.351 & 0.139 & 0.139\\
      & 0.01  & 0.237 & 0.049 & 0.051 & 0.140 & 0.051 & 0.041\\
    \midrule
    \multirow{3}{*}{PJ}
      & 0.10  & 0.147 & 0.199 & 0.172 & 0.149 & 0.169 & 0.159\\
      & 0.05  & 0.077 & 0.101 & 0.099 & 0.076 & 0.100 & 0.083\\
      & 0.01  & 0.018 & 0.022 & 0.024 & 0.019 & 0.035 & 0.027\\
    \midrule
    $c=2$ & & \multicolumn{3}{c}{\textbf{$n=200$}} & \multicolumn{3}{c}{\textbf{$n=400$}} \\
    \cmidrule(lr){3-5}\cmidrule(lr){6-8}
    & \textbf{$\alpha/\gamma$} & 0 & 5 & 10 & 0 & 5 & 10\\
    \midrule
    \multirow{3}{*}{DM}
      & 0.10  & 0.626 & 0.217 & 0.222 & 0.474 & 0.178 & 0.178\\
      & 0.05  & 0.422 & 0.133 & 0.114 & 0.271 & 0.107 & 0.098\\
      & 0.01  & 0.182 & 0.047 & 0.026 & 0.100 & 0.038 & 0.023\\
    \midrule
    \multirow{3}{*}{PJ}
      & 0.10  & 0.097 & 0.130 & 0.103 & 0.099 & 0.109 & 0.107\\
      & 0.05  & 0.051 & 0.068 & 0.059 & 0.053 & 0.061 & 0.054\\
      & 0.01  & 0.012 & 0.020 & 0.015 & 0.013 & 0.021 & 0.015\\
    \bottomrule
  \end{tabular}
  \label{tab:CIsind1p1q2}
\end{table}

\begin{table}[htbp]
  \centering
  \setlength{\tabcolsep}{5pt}
  \renewcommand{\arraystretch}{1.1}
  \caption{Testing conditional independence when $\Psi(X,Z)=\sin(\gamma Z)$, $X^{(1)}\sim N(0,1)$, $X^{(2)}\sim N(0,1)$, and $Z\sim N(0,1)$}
  \begin{tabular}{llcccccccc}
    \toprule
    $c=0.5$ & & \multicolumn{3}{c}{\textbf{$n=200$}} & \multicolumn{3}{c}{\textbf{$n=400$}} \\
    \cmidrule(lr){3-5}\cmidrule(lr){6-8}
    & \textbf{$\alpha/\gamma$} & 0 & 5 & 10 & 0 & 5 & 10\\
    \midrule
    \multirow{3}{*}{DM}
      & 0.10  & 0.854 & 0.513 & 0.527 & 0.755 & 0.441 & 0.455\\
      & 0.05  & 0.706 & 0.353 & 0.343 & 0.565 & 0.292 & 0.317\\
      & 0.01  & 0.370 & 0.145 & 0.141 & 0.251 & 0.124 & 0.137\\
    \midrule
    \multirow{3}{*}{PJ}
      & 0.10  & 0.181 & 0.292 & 0.321 & 0.219 & 0.307 & 0.332\\
      & 0.05  & 0.095 & 0.156 & 0.169 & 0.116 & 0.180 & 0.172\\
      & 0.01  & 0.027 & 0.039 & 0.039 & 0.040 & 0.057 & 0.056\\
    \midrule
    $c=1$ & & \multicolumn{3}{c}{\textbf{$n=200$}} & \multicolumn{3}{c}{\textbf{$n=400$}} \\
    \cmidrule(lr){3-5}\cmidrule(lr){6-8}
    & \textbf{$\alpha/\gamma$} & 0 & 5 & 10 & 0 & 5 & 10\\
    \midrule
    \multirow{3}{*}{DM}
      & 0.10  & 0.604 & 0.283 & 0.280 & 0.431 & 0.244 & 0.253\\
      & 0.05  & 0.421 & 0.179 & 0.170 & 0.265 & 0.141 & 0.154\\
      & 0.01  & 0.168 & 0.054 & 0.070 & 0.083 & 0.050 & 0.055\\
    \midrule
    \multirow{3}{*}{PJ}
      & 0.10  & 0.156 & 0.166 & 0.214 & 0.166 & 0.199 & 0.195\\
      & 0.05  & 0.087 & 0.096 & 0.103 & 0.085 & 0.108 & 0.107\\
      & 0.01  & 0.029 & 0.025 & 0.030 & 0.023 & 0.036 & 0.032\\
    \midrule
    $c=2$ & & \multicolumn{3}{c}{\textbf{$n=200$}} & \multicolumn{3}{c}{\textbf{$n=400$}} \\
    \cmidrule(lr){3-5}\cmidrule(lr){6-8}
    & \textbf{$\alpha/\gamma$} & 0 & 5 & 10 & 0 & 5 & 10\\
    \midrule
    \multirow{3}{*}{DM}
      & 0.10  & 0.436 & 0.208 & 0.192 & 0.310 & 0.190 & 0.177\\
      & 0.05  & 0.279 & 0.117 & 0.112 & 0.163 & 0.104 & 0.104\\
      & 0.01  & 0.111 & 0.032 & 0.035 & 0.048 & 0.034 & 0.027\\
    \midrule
    \multirow{3}{*}{PJ}
      & 0.10  & 0.122 & 0.131 & 0.122 & 0.094 & 0.130 & 0.116\\
      & 0.05  & 0.062 & 0.070 & 0.059 & 0.054 & 0.070 & 0.069\\
      & 0.01  & 0.015 & 0.017 & 0.019 & 0.009 & 0.016 & 0.019\\
    \bottomrule
  \end{tabular}
  \label{tab:CIsind2p1q2}
\end{table}

\begin{table}[htbp]
  \centering
  \setlength{\tabcolsep}{5pt}
  \renewcommand{\arraystretch}{1.1}
  \caption{Testing conditional independence when $\Psi(X,Z)=\sin(\gamma Z)$, $X^{(1)}\sim U(0,1)$, $X^{(2)}\sim U(0,1)$ and $Z\sim U(0,1)$}
  \begin{tabular}{llcccccccc}
    \toprule
    $c=0.5$ & & \multicolumn{3}{c}{\textbf{$n=200$}} & \multicolumn{3}{c}{\textbf{$n=400$}} \\
    \cmidrule(lr){3-5}\cmidrule(lr){6-8}
    & \textbf{$\alpha/\gamma$} & 0 & 5 & 10 & 0 & 5 & 10\\
    \midrule
    \multirow{3}{*}{DM}
      & 0.10  & 0.579 & 0.972 & 0.470 & 0.416 & 0.998 & 0.738\\
      & 0.05  & 0.359 & 0.938 & 0.318 & 0.231 & 0.998 & 0.557\\
      & 0.01  & 0.130 & 0.762 & 0.125 & 0.078 & 0.991 & 0.249\\
    \midrule
    \multirow{3}{*}{PJ}
      & 0.10  & 0.126 & 0.961 & 0.326 & 0.127 & 0.999 & 0.611\\
      & 0.05  & 0.063 & 0.918 & 0.193 & 0.073 & 0.998 & 0.446\\
      & 0.01  & 0.018 & 0.706 & 0.069 & 0.018 & 0.990 & 0.194\\
    \midrule
    $c=1$ & & \multicolumn{3}{c}{\textbf{$n=200$}} & \multicolumn{3}{c}{\textbf{$n=400$}} \\
    \cmidrule(lr){3-5}\cmidrule(lr){6-8}
    & \textbf{$\alpha/\gamma$} & 0 & 5 & 10 & 0 & 5 & 10\\
    \midrule
    \multirow{3}{*}{DM}
      & 0.10  & 0.467 & 0.985 & 0.428 & 0.353 & 1.000 & 0.744\\
      & 0.05  & 0.267 & 0.953 & 0.264 & 0.198 & 1.000 & 0.574\\
      & 0.01  & 0.079 & 0.809 & 0.110 & 0.057 & 0.998 & 0.234\\
    \midrule
    \multirow{3}{*}{PJ}
      & 0.10  & 0.103 & 0.985 & 0.341 & 0.104 & 1.000 & 0.654\\
      & 0.05  & 0.051 & 0.957 & 0.213 & 0.062 & 1.000 & 0.494\\
      & 0.01  & 0.011 & 0.855 & 0.083 & 0.018 & 0.998 & 0.241\\
    \midrule
    $c=2$ & & \multicolumn{3}{c}{\textbf{$n=200$}} & \multicolumn{3}{c}{\textbf{$n=400$}} \\
    \cmidrule(lr){3-5}\cmidrule(lr){6-8}
    & \textbf{$\alpha/\gamma$} & 0 & 5 & 10 & 0 & 5 & 10\\
    \midrule
    \multirow{3}{*}{DM}
      & 0.10  & 0.444 & 0.986 & 0.639 & 0.590 & 1.000 & 0.764\\
      & 0.05  & 0.291 & 0.958 & 0.419 & 0.362 & 1.000 & 0.586\\
      & 0.01  & 0.118 & 0.818 & 0.148 & 0.112 & 0.997 & 0.289\\
    \midrule
    \multirow{3}{*}{PJ}
      & 0.10  & 0.102 & 0.982 & 0.377 & 0.104 & 1.000 & 0.768\\
      & 0.05  & 0.046 & 0.956 & 0.222 & 0.051 & 1.000 & 0.555\\
      & 0.01  & 0.010 & 0.851 & 0.087 & 0.017 & 0.999 & 0.280\\
    \bottomrule
  \end{tabular}
  \label{tab:CIsind3p1q2}
\end{table}

\begin{table}[htbp]
  \centering
  \setlength{\tabcolsep}{5pt}
  \renewcommand{\arraystretch}{1.1}
  \caption{Power of testing conditional independence when $\Psi(X,Z)=\exp(\gamma Z^{(1)})$, $X^{(1)}=Z^{(1)}+U$, $Z^{(1)}\sim N(0,1)$ and $U\sim N(0,1)$}
  \begin{tabular}{llcccccccc}
    \toprule
    $p=1,q=1$ & & \multicolumn{3}{c}{\textbf{$n=200$}} & \multicolumn{3}{c}{\textbf{$n=400$}} \\
    \cmidrule(lr){3-5}\cmidrule(lr){6-8}
    & \textbf{$\alpha/c$} & 0.5 & 1 & 2 & 0.5 & 1 & 2\\
    \midrule
    \multirow{3}{*}{DM}
      & 0.10  & 0.922 & 0.918 & 0.922 & 1.000 & 1.000 & 1.000\\
      & 0.05  & 0.779 & 0.768 & 0.771 & 0.995 & 0.995 & 0.997\\
      & 0.01  & 0.442 & 0.436 & 0.429 & 0.904 & 0.918 & 0.922\\
    \midrule
    \multirow{3}{*}{PJ}
      & 0.10  & 0.897 & 0.826 & 0.855 & 0.999 & 1.000 & 1.000\\
      & 0.05  & 0.649 & 0.633 & 0.711 & 0.992 & 0.989 & 0.997\\
      & 0.01  & 0.293 & 0.337 & 0.410 & 0.791 & 0.820 & 0.880\\
    \midrule
    $p=2,q=1$ & & \multicolumn{3}{c}{\textbf{$n=200$}} & \multicolumn{3}{c}{\textbf{$n=400$}} \\
    \cmidrule(lr){3-5}\cmidrule(lr){6-8}
    & \textbf{$\alpha/c$} & 0.5 & 1 & 2 & 0.5 & 1 & 2\\
    \midrule
    \multirow{3}{*}{DM}
      & 0.10  & 0.282 & 0.245 & 0.237 & 0.474 & 0.464 & 0.468\\
      & 0.05  & 0.179 & 0.143 & 0.137 & 0.326 & 0.323 & 0.312\\
      & 0.01  & 0.049 & 0.048 & 0.046 & 0.104 & 0.091 & 0.086\\
    \midrule
    \multirow{3}{*}{PJ}
      & 0.10  & 0.237 & 0.203 & 0.191 & 0.439 & 0.315 & 0.360\\
      & 0.05  & 0.128 & 0.099 & 0.091 & 0.270 & 0.151 & 0.180\\
      & 0.01  & 0.036 & 0.028 & 0.024 & 0.102 & 0.039 & 0.034\\
    \midrule
    $p=1,q=2$ & & \multicolumn{3}{c}{\textbf{$n=200$}} & \multicolumn{3}{c}{\textbf{$n=400$}} \\
    \cmidrule(lr){3-5}\cmidrule(lr){6-8}
    & \textbf{$\alpha/c$} & 0.5 & 1 & 2 & 0.5 & 1 & 2\\
    \midrule
    \multirow{3}{*}{DM}
      & 0.10  & 0.594 & 0.405 & 0.389 & 0.627 & 0.512 & 0.478\\
      & 0.05  & 0.417 & 0.258 & 0.232 & 0.442 & 0.339 & 0.314\\
      & 0.01  & 0.172 & 0.095 & 0.079 & 0.215 & 0.135 & 0.117\\
    \midrule
    \multirow{3}{*}{PJ}
      & 0.10  & 0.422 & 0.365 & 0.309 & 0.561 & 0.558 & 0.543\\
      & 0.05  & 0.216 & 0.223 & 0.156 & 0.381 & 0.352 & 0.311\\
      & 0.01  & 0.054 & 0.063 & 0.049 & 0.137 & 0.127 & 0.087\\
    \bottomrule
  \end{tabular}
  \label{tab:CIsind1}
\end{table}

\begin{table}[htbp]
  \centering
  \setlength{\tabcolsep}{5pt}
  \renewcommand{\arraystretch}{1.1}
  \caption{Power of testing conditional independence when $\Psi(X,Z)=\exp(\gamma Z^{(1)})$, $X^{(1)}\sim N(0,1)$ and $Z^{(1)}\sim N(0,1)$}
  \begin{tabular}{llcccccccc}
    \toprule
    $p=1,q=1$ & & \multicolumn{3}{c}{\textbf{$n=200$}} & \multicolumn{3}{c}{\textbf{$n=400$}} \\
    \cmidrule(lr){3-5}\cmidrule(lr){6-8}
    & \textbf{$\alpha/c$} & 0.5 & 1 & 2 & 0.5 & 1 & 2\\
    \midrule
    \multirow{3}{*}{DM}
      & 0.10  & 0.869 & 0.870 & 0.880 & 0.999 & 1.000 & 1.000\\
      & 0.05  & 0.699 & 0.678 & 0.699 & 0.989 & 0.991 & 0.991\\
      & 0.01  & 0.338 & 0.301 & 0.307 & 0.860 & 0.858 & 0.860\\
    \midrule
    \multirow{3}{*}{PJ}
      & 0.10  & 0.851 & 0.838 & 0.920 & 1.000 & 1.000 & 1.000\\
      & 0.05  & 0.581 & 0.662 & 0.810 & 0.994 & 0.987 & 0.998\\
      & 0.01  & 0.293 & 0.361 & 0.502 & 0.752 & 0.863 & 0.937\\
    \midrule
    $p=2,q=1$ & & \multicolumn{3}{c}{\textbf{$n=200$}} & \multicolumn{3}{c}{\textbf{$n=400$}} \\
    \cmidrule(lr){3-5}\cmidrule(lr){6-8}
    & \textbf{$\alpha/c$} & 0.5 & 1 & 2 & 0.5 & 1 & 2\\
    \midrule
    \multirow{3}{*}{DM}
      & 0.10  & 0.238 & 0.201 & 0.195 & 0.337 & 0.334 & 0.326\\
      & 0.05  & 0.130 & 0.112 & 0.099 & 0.217 & 0.202 & 0.200\\
      & 0.01  & 0.038 & 0.027 & 0.025 & 0.080 & 0.064 & 0.062\\
    \midrule
    \multirow{3}{*}{PJ}
      & 0.10  & 0.201 & 0.167 & 0.166 & 0.305 & 0.226 & 0.298\\
      & 0.05  & 0.121 & 0.096 & 0.094 & 0.176 & 0.116 & 0.158\\
      & 0.01  & 0.040 & 0.031 & 0.024 & 0.054 & 0.034 & 0.030\\
    \midrule
    $p=1,q=2$ & & \multicolumn{3}{c}{\textbf{$n=200$}} & \multicolumn{3}{c}{\textbf{$n=400$}} \\
    \cmidrule(lr){3-5}\cmidrule(lr){6-8}
    & \textbf{$\alpha/c$} & 0.5 & 1 & 2 & 0.5 & 1 & 2\\
    \midrule
    \multirow{3}{*}{DM}
      & 0.10  & 0.542 & 0.351 & 0.271 & 0.529 & 0.372 & 0.289\\
      & 0.05  & 0.363 & 0.223 & 0.140 & 0.367 & 0.231 & 0.166\\
      & 0.01  & 0.160 & 0.079 & 0.050 & 0.164 & 0.078 & 0.051\\
    \midrule
    \multirow{3}{*}{PJ}
      & 0.10  & 0.360 & 0.253 & 0.203 & 0.398 & 0.294 & 0.287\\
      & 0.05  & 0.196 & 0.152 & 0.107 & 0.232 & 0.154 & 0.151\\
      & 0.01  & 0.051 & 0.044 & 0.031 & 0.085 & 0.050 & 0.042\\
    \bottomrule
  \end{tabular}
  \label{tab:CIsind2}
\end{table}

\begin{table}[htbp]
  \centering
  \setlength{\tabcolsep}{5pt}
  \renewcommand{\arraystretch}{1.1}
  \caption{Power of testing conditional independence when $\Psi(X,Z)=\exp(\gamma Z^{(1)})$, $X^{(1)}\sim U(0,1)$ and $Z^{(1)}\sim U(0,1)$}
  \begin{tabular}{llcccccccc}
    \toprule
    $p=1,q=1$ & & \multicolumn{3}{c}{\textbf{$n=200$}} & \multicolumn{3}{c}{\textbf{$n=400$}} \\
    \cmidrule(lr){3-5}\cmidrule(lr){6-8}
    & \textbf{$\alpha/c$} & 0.5 & 1 & 2 & 0.5 & 1 & 2\\
    \midrule
    \multirow{3}{*}{DM}
      & 0.10  & 1.000 & 1.000 & 1.000 & 1.000 & 1.000 & 1.000\\
      & 0.05  & 1.000 & 1.000 & 1.000 & 1.000 & 1.000 & 1.000\\
      & 0.01  & 1.000 & 1.000 & 1.000 & 1.000 & 1.000 & 1.000\\
    \midrule
    \multirow{3}{*}{PJ}
      & 0.10  & 1.000 & 1.000 & 1.000 & 1.000 & 1.000 & 1.000\\
      & 0.05  & 1.000 & 1.000 & 1.000 & 1.000 & 1.000 & 1.000\\
      & 0.01  & 1.000 & 1.000 & 1.000 & 1.000 & 1.000 & 1.000\\
    \midrule
    $p=2,q=1$ & & \multicolumn{3}{c}{\textbf{$n=200$}} & \multicolumn{3}{c}{\textbf{$n=400$}} \\
    \cmidrule(lr){3-5}\cmidrule(lr){6-8}
    & \textbf{$\alpha/c$} & 0.5 & 1 & 2 & 0.5 & 1 & 2\\
    \midrule
    \multirow{3}{*}{DM}
      & 0.10  & 0.924 & 0.930 & 0.917 & 1.000 & 1.000 & 1.000\\
      & 0.05  & 0.817 & 0.820 & 0.819 & 0.997 & 0.998 & 0.998\\
      & 0.01  & 0.505 & 0.519 & 0.501 & 0.948 & 0.942 & 0.947\\
    \midrule
    \multirow{3}{*}{PJ}
      & 0.10  & 0.990 & 0.990 & 0.990 & 1.000 & 1.000 & 1.000\\
      & 0.05  & 0.951 & 0.956 & 0.968 & 1.000 & 1.000 & 1.000\\
      & 0.01  & 0.686 & 0.712 & 0.759 & 0.997 & 0.999 & 1.000\\
    \midrule
    $p=1,q=2$ & & \multicolumn{3}{c}{\textbf{$n=200$}} & \multicolumn{3}{c}{\textbf{$n=400$}} \\
    \cmidrule(lr){3-5}\cmidrule(lr){6-8}
    & \textbf{$\alpha/c$} & 0.5 & 1 & 2 & 0.5 & 1 & 2\\
    \midrule
    \multirow{3}{*}{DM}
      & 0.10  & 0.401 & 0.362 & 0.420 & 0.620 & 0.626 & 0.657\\
      & 0.05  & 0.239 & 0.228 & 0.269 & 0.450 & 0.447 & 0.469\\
      & 0.01  & 0.091 & 0.072 & 0.088 & 0.221 & 0.198 & 0.208\\
    \midrule
    \multirow{3}{*}{PJ}
      & 0.10  & 0.388 & 0.357 & 0.320 & 0.703 & 0.711 & 0.652\\
      & 0.05  & 0.230 & 0.247 & 0.183 & 0.544 & 0.569 & 0.468\\
      & 0.01  & 0.067 & 0.073 & 0.058 & 0.265 & 0.267 & 0.203\\
    \bottomrule
  \end{tabular}
  \label{tab:CIsind3}
\end{table}

%\section{Definitions In This Article}\label{sec.AppendixB}
\section{Some definitions}\label{sec.AppendixB}
    In this section, we first introduce the necessary notations to simplify the presentation of the proofs and recall several important notations defined earlier in the paper. We adopt the notation used in \cite{delgado2001significance} and denote each sample as $\chi_i = (Y_i,X_i,Z_i)$. For a three-variable asymmetric kernel function $\varphi_{wb}(\chi_i,\chi_j,\chi_k)$, indexed by $w$ and depending on the bandwidth parameter $b$, we define its symmetric counterpart 
    \begin{align*}
        \Phi_{wb}(\chi_i,\chi_j,\chi_k)=\frac{1}{6}\sum_{\substack{l_1,l_2,l_3\in\{i,j,k\}\\ l_1\neq l_2 \neq l_3}}\varphi_{wb}(\chi_{l_1},\chi_{l_2},\chi_{l_3}).
    \end{align*}
    We then employ the $U$-process and its Hoeffding projections constructed from the symmetric kernel in \cite{delgado2001significance} to define the counterpart from the asymmetric kernel,
    \begin{align*}
        U_n^{(m)}(\varphi_{wb}):=U_n^{(m)}(\Phi_{wb}),\quad U_n^{(m)}(\pi_k\varphi_{wb}):=U_n^{(m)}(\pi_k\Phi_{wb}).
    \end{align*}
    More specifically, we define the conditional expectation of the function $\phi$ given the $i$-th sample as $\mathbb{E}_i(\phi):=\mathbb{E}(\phi\vert \chi_i)$. Consequently, the first-order component of the Hoeffding projection $U_n^{(1)}(\pi_1\varphi_{wb})$ is equivalent to
    \begin{align*}
        \frac{1}{n}\sum_{l=1}^n\sum_{l\in\{i,j,k\}}\mathbb{E}_l\left[\varphi_{wb}(\chi_i,\chi_j,\chi_k)\right]-3\mathbb{E}\left[\varphi_{wb}(\chi_i,\chi_j,\chi_k)\right].
    \end{align*}
    Subsequently, for the $q$-dimensional variable $X$ and the $p$-dimensional variable $Z$, we employ product kernel functions $K(\cdot)$ on $\mathbb{R}^q$ with bandwidth $a$ and  $L(\cdot)$ on $\mathbb{R}^p$ with bandwidth $b$, respectively, which are defined as follows,
    \begin{align*}
        K_{ij} = K_a\left(\frac{X_i-X_j}{a}\right), \quad L_{ij} = L_b\left(\frac{Z_i-Z_j}{b}\right).
    \end{align*}
    Note that for a $q$-dimensional vector $u$ and a $p$-dimensional vector $v$, the kernel functions $K(\cdot)$ and $L(\cdot)$ are constructed using the same univariate kernel function $k(\cdot)$. Specifically,
    \begin{align*}
        K(u) = \prod_{i=1}^qk(u^{(i)}), \quad L(v) = \prod_{j=1}^pk(v^{(j)}),
    \end{align*}
    where $u^{(i)}$ denotes the $i$-th component of $u$. Consequently, the kernel density estimator defined above can be expressed as
    \begin{align*} 
        \hat{f}_X(X_i) = \frac{1}{(n-1)} \sum_{\substack{j=1 \\ j \neq i}}^n K_{ij},\quad \hat{f}_W(W_i) = \frac{1}{(n-1)} \sum_{\substack{j=1 \\ j \neq i}}^n K_{ij}L_{ij}.
    \end{align*}
    In addition, we define several projections of the kernel function given a sample $\chi_i$ or a covariate $X_i$, which are frequently used in the subsequent analysis,
    \begin{align*}
        &d_X(X_i) = \mathbb{E}\left[K_{ij}\vert X_i\right], \quad d_W(W_i) = \mathbb{E}\left[K_{ij}L_{ij}\vert W_i\right], \quad d_W^{(2)}(W_i) = \mathbb{E}\left[K^2_{ij}L_{ij}\vert W_i\right].
    \end{align*}
    Subsequently, we simplify the notations for several terms that are repeatedly used. Specifically, we use $1_x(X_i) = 1(X_i\leq x)$ to denote the indicator function, $p_w(X_i) = f_X(X_i)1_x(X_i)\Delta^{-1}(X_i)G(z;X_i)$ as a shorthand for the main structure of the projection constructed in this paper, and define the commonly used components in the statistics as
    \begin{align*}
        &\psi^{(1)}_{wb}(\chi_i,\chi_j) = \left(Y_i-m(X_i)\right)K^2_{ij}L_{ij}1_x(X_i)\Delta^{-1}(X_i)G(z;X_i),\\
        &\varphi^{(1)}_{wb}(\chi_i,\chi_j,\chi_k) = \left(Y_i-m(X_i)\right)K_{ij}K_{ik}L_{ik}1_x(X_i)\Delta^{-1}(X_i)G(z;X_i),\\
        &\psi_{wb}^{(2)}(\chi_i,\chi_j) = \left(Y_j-m(X_j)\right)K^2_{ij}L_{ij}1_x(X_i)\Delta^{-1}(X_i)G(z;X_i),\\
        &\varphi_{wb}^{(2)}(\chi_i,\chi_j,\chi_k) = \left(Y_j-m(X_j)\right)K_{ij}K_{ik}L_{ik}1_x(X_i)\Delta^{-1}(X_i)G(z;X_i),\\
        &\psi_{wb}^{(3)}(\chi_i,\chi_j) = \left(m(X_i)-m(X_j)\right)K^2_{ij}L_{ij}1_x(X_i)\Delta^{-1}(X_i)G(z;X_i),\\
        &\varphi_{wb}^{(3)}(\chi_i,\chi_j,\chi_k) = \left(m(X_i)-m(X_j)\right)K_{ij}K_{ik}L_{ik}1_x(X_i)\Delta^{-1}(X_i)G(z;X_i).
    \end{align*}
    We decompose the statistic into the same form as in \cite{delgado2001significance}, along with the components introduced by the projection. We then define its population counterpart, obtained by removing the estimation effect, as well as the bootstrap versions of all decomposed terms, as follows
    \begin{align*}
         &\hat T_n(w) = \frac{1}{n(n-1)} \sum_{i=1}^n\sum_{\substack{j=1 \\ j \neq i}}^n (Y_i - Y_j)K_{ij}1_w(W_i),\\
         &\hat T_n^\ast(w) = \frac{1}{n(n-1)} \sum_{i=1}^n\sum_{\substack{j=1 \\ j \neq i}}^n V_i(Y_i - Y_j)K_{ij}1_w(W_i),\\
         &\hat{S}_n(w) = \frac{1}{n(n-1)^2} \sum_{i=1}^n\sum_{\substack{j=1 \\ j \neq i}}^n\sum_{\substack{k=1 \\ k \neq i}}^n (Y_i - Y_j)K_{ij}K_{ik}L_{ik}1_x(X_i)\hat{\Delta}_n^{-1}(X_i)\hat{G}_n(z;X_i),\\
        &\hat{S}^\ast_n(w) = \frac{1}{n(n-1)^2} \sum_{i=1}^n\sum_{\substack{j=1 \\ j \neq i}}^n\sum_{\substack{k=1 \\ k \neq i}}^n V_i (Y_i - Y_j)K_{ij}K_{ik}L_{ik}1_x(X_i)\hat{\Delta}_n^{-1}(X_i)\hat{G}_n(z;X_i),\\
        &S_n(w) = \frac{1}{n(n-1)^2} \sum_{i=1}^n\sum_{\substack{j=1 \\ j \neq i}}^n\sum_{\substack{k=1 \\ k \neq i}}^n (Y_i - Y_j)K_{ij}K_{ik}L_{ik}1_x(X_i)\Delta^{-1}(X_i)G(z;X_i),\\
        &S_n^\ast(w) = \frac{1}{n(n-1)^2} \sum_{i=1}^n\sum_{\substack{j=1 \\ j \neq i}}^n\sum_{\substack{k=1 \\ k \neq i}}^n V_i(Y_i - Y_j)K_{ij}K_{ik}L_{ik}1_x(X_i)\Delta^{-1}(X_i)G(z;X_i),
    \end{align*}
    and their asymptotically equivalent counterparts
    \begin{align*}
        &T_n^{(1)}(w) = \frac{1}{n}\sum_{i=1}^n\left(Y_i-m(X_i)\right)f_X(X_i)1_w(W_i),\\
        &T_n^{(2)}(w) = \frac{1}{n}\sum_{i=1}^n\left(Y_i-m(X_i)\right)1_x(X_i)G(z;X_i),\\
        &S_n^{(1)}(w) = \frac{1}{n}\sum_{i=1}^n\left(Y_i-m(X_i)\right)f_X(X_i)1_x(X_i)f_W(W_i)\Delta^{-1}(X_i)G(z;X_i),\\
        &{T_n^{(1)}}^\ast(w) = \frac{1}{n}\sum_{i=1}^nV_i\left(Y_i-m(X_i)\right)f_X(X_i)1_w(W_i),\\
        &{S_n^{(1)}}^\ast(w) = \frac{1}{n}\sum_{i=1}^nV_i\left(Y_i-m(X_i)\right)f_X(X_i)1_x(X_i)f_W(W_i)\Delta^{-1}(X_i)G(z;X_i).
    \end{align*}

\section{Proofs of theorems}\label{sec.AppendixC}
\begin{proof}[Proof of Theorem \ref{thm.null}]
    We start by decomposing 
    \begin{align}\label{thmproof.null decomp}
        \hat R_n(w) = \hat T_n(w) - \hat{S}_n(w),
    \end{align}
    where we recall the definitions of $\hat T_n(w)$ and $\hat{S}_n(w)$ in Section \ref{sec.AppendixB}. Noting that, under the null hypothesis, the combined results of Lemmas \ref{lemma.Shat-S} and \ref{lemma.Sn} imply 
    \begin{align}\label{thmproof.null Sn}
        \sup_{w}\left\vert \hat{S}_n(w) - S_n^{(1)}(w) + T_n^{(2)}(w)\right\vert = o_p\left(n^{-1/2}\right).
    \end{align}
    Moreover, based on Propositions $2$ and $3$ in \cite{delgado2001significance}, we can also obtain the result concerning $\hat T_n(w)$ under the null hypothesis:
    \begin{align}\label{thmproof.null Tn}
        \sup_{w}\left\vert \hat T_n(w) - T_n^{(1)}(w) + T_n^{(2)}(w)\right\vert = o_p\left(n^{-1/2}\right).
    \end{align}
    Combining \eqref{thmproof.null Sn} and \eqref{thmproof.null Tn} leads to the result stated in equation \eqref{thm.eq null}. As stated in \cite{delgado2001significance}, $T_n^{(1)}(w)$ is VC-type with common square $\mathbb{P}$-integrable envelope proportional to $\vert Y-m(X)\vert$, whose integrability is ensured under Assumption \ref{ass.sample}. Given the restriction on $\Delta^{-1}(x)f^2_X(x)$ imposed by Assumption \ref{ass.sample} and the property of $G(z;x)$ derived in Lemma \ref{lemma.G}, it follows that $S_n^{(1)}(w)$ is a centered empirical process indexed by a VC-type class of functions with common square P-integrable envelope $\vert (Y-m(X))f_X^2(X)\Delta^{-1}(X)\vert$. Since the sum of two VC-type classes of functions is also VC-type with envelope the sum of the envelopes, the class $\{\gamma^{(1)}_w-\gamma^{(2)}_w:w\in\mathbb{R}^{p+q}\}$ is a VC-type class with square $\mathbb{P}$-integrable envelope, where $\gamma^{(1)}_w$ and $\gamma^{(2)}_w$ is defined in Section \ref{sec.AppendixB}. Therefore, the aforementioned function class is $\mathbb{P}$-Donsker, and the conclusion of Theorem \ref{thm.null} holds.
\end{proof}

\begin{proof}[Proof of Theorem \ref{thm.local}]
    Noting that decomposition \eqref{thmproof.null decomp} does not depend on the null hypothesis, and that expressions \eqref{thmproof.null Sn} and \eqref{thmproof.null Tn} hold under the local alternatives for the same reasons as under the null, we can similarly first establish 
    \begin{align}\label{thm.eq local}
        \sup_{w}\left\vert \sqrt{n}\hat R_n(w)-\frac{1}{\sqrt{n}}\sum_{i=1}^n\epsilon_if_X(X_i)1_x(X_i)\mathcal{P}1_z(Z_i)\right\vert=o_p\left(1\right).
    \end{align}
    Based on the expression of the local alternatives defined in \eqref{hyp.local}, and noting that the result established in \cite{delgado2001significance} under the local alternatives can also be applied, the result in Theorem \ref{thm.local} follows.
\end{proof}

\begin{proof}[Proof of Theorem \ref{thm.alt}]
    The proof of Theorem \ref{thm.alt} follows the same line of reasoning as that of Theorem \ref{thm.null}. The distinction is that, under the alternative hypothesis, only convergence in probability is required. Consequently, when Lemmas \ref{lemma.Shat-S} and \ref{lemma.Sn} are invoked, results \eqref{lemma.Shat-S eq2} and \eqref{lemma.Sn alt} are employed instead of the previous conclusions regarding the convergence rate.
\end{proof}

\begin{proof}[Proof of Theorem \ref{thm.boot}]
    The proof of the bootstrap version likewise begins with the decomposition of
    \begin{align}\label{thmproof.boot decomp}
        \hat R_n^\ast(w) = \hat T_n^\ast(w) - \hat{S}_n^\ast(w).
    \end{align}
    Subsequently, regardless of the hypothesis considered, Lemmas \ref{lemma.Shat-S boot} and \ref{lemma.Sn boot}, along with Lemma $2$ from \cite{delgado2001significance}, imply the conclusion for the bootstrap counterpart, namely,
    \begin{align}\label{thmproof.boot res}
        \sup_{w}\left\vert \hat R^\ast_n(w)-\frac{1}{n}\sum_{i=1}^nV_i\left(Y_i-m(X_i)\right)f_X(X_i)1_x(X_i)\mathcal{P}1_z(Z_i)\right\vert=o_p\left(n^{-1/2}\right).
    \end{align}
    Under the null and local alternative hypotheses, the argument presented in the proof of Theorem \ref{thm.null} can be established in an entirely analogous manner to establish
    \begin{align}\label{thmproof.boot null}
        \sup_w\left\vert \frac{1}{n}\sum_{i=1}^nV_i\left(Y_i-m(X_i)\right)f_X(X_i)1_x(X_i)\mathcal{P}1_z(Z_i) - R_{\infty}(w)\right\vert  = o_p\left(n^{-1/2}\right),
    \end{align}
    where $R_\infty(\cdot)$ denotes the same limiting process as that discussed in Theorems \ref{thm.null} and \ref{thm.local}. Under the alternative hypothesis, as indicated by the previously mentioned lemmas, the form of equation \eqref{thmproof.boot res} remains the same; however, the distribution of the underlying data differs, and hence the limiting process of the empirical process in \eqref{thmproof.boot res} will change accordingly,
    \begin{align}\label{thmproof.boot alt}
        \sup_w\left\vert \frac{1}{n}\sum_{i=1}^nV_i\left(Y_i-m(X_i)\right)f_X(X_i)1_x(X_i)\mathcal{P}1_z(Z_i) - R_{\infty}^1(w)\right\vert  = o_p\left(n^{-1/2}\right).
    \end{align}
    Combining equations \eqref{thmproof.boot decomp}--\eqref{thmproof.boot alt}, all the results required for Theorem \ref{thm.boot} hold.
\end{proof}

\section{Lemmas and proofs}\label{sec.AppendixD}
\begin{lemma}\label{lemma.Shat-S}
    Suppose that Assumptions \ref{ass.sample}--\ref{ass.bandwidth} hold, whether under the null hypothesis or the sequence of local alternatives,
    \begin{align}\label{lemma.Shat-S eq1}
        \sup_w\left\vert \hat{S}_n(w)-S_n(w)\right\vert = o_p\left(n^{-1/2}\right).
    \end{align}
    Under the alternative hypothesis,
    \begin{align}\label{lemma.Shat-S eq2}
        \sup_w\left\vert \hat{S}_n(w)-S_n(w)\right\vert = o_p\left(1\right).
    \end{align}
\end{lemma}

\begin{proof}[Proof of Lemma \ref{lemma.Shat-S}]
    We start by decomposing
    \begin{align}\label{lemmaproof.Shat-S decomp}
        \hat{S}_n(w)-S_n(w) = S_{n1}(w)+S_{n2}(w)+S_{n3}(w),
    \end{align}
    where
    \begin{align*}
        &S_{n1}(w) = \frac{1}{n}\sum_{i=1}^n\left(Y_i-\hat{m}(X_i)\right)\hat{f}_X(X_i)\hat{f}_W(W_i)1_x(X_i)\left[\hat{\Delta}_n^{-1}(X_i)-\Delta^{-1}(X_i)\right]G(z;X_i),\\
        &S_{n2}(w) = \frac{1}{n}\sum_{i=1}^n\left(Y_i-\hat{m}(X_i)\right)\hat{f}_X(X_i)\hat{f}_W(W_i)1_x(X_i)\Delta^{-1}(X_i)\left[\hat{G}_n(z;X_i)-G(z;X_i)\right],\\
        &S_{n3}(w) = \frac{1}{n}\sum_{i=1}^n\left(Y_i-\hat{m}(X_i)\right)\hat{f}_X(X_i)\hat{f}_W(W_i)1_x(X_i)\left[\hat{\Delta}_n^{-1}(X_i)-\Delta^{-1}(X_i)\right]\left[\hat{G}_n(z;X_i)-G(z;X_i)\right].
    \end{align*}
    By further decomposing
    \begin{align*}
        \hat{\Delta}_n^{-1}(X_i)-\Delta^{-1}(X_i) = -\frac{\hat{\Delta}_n(X_i)-\Delta(X_i)}{\Delta^2(X_i)}+\frac{\left(\hat{\Delta}_n(X_i)-\Delta(X_i)\right)^2}{\Delta^2(X_i)\hat{\Delta}_n(X_i)},
    \end{align*}
    $S_{n1}(w)$ can be rewritten as
    \begin{align}\label{lemmaproof.Shat-S Sn1 decomp}
        S_{n1}(w) = -S_{n11}(w)+S_{n12}(w),
    \end{align}
    where
    \begin{align*}
        &S_{n11}(w) = \frac{1}{n}\sum_{i=1}^n\left(Y_i-\hat{m}(X_i)\right)\hat{f}_X(X_i)\hat{f}_W(W_i)1_x(X_i)\left[\hat{\Delta}_n(X_i)-\Delta(X_i)\right]\frac{G(z;X_i)}{\Delta^2(X_i)}\\
        &S_{n12}(w) = \frac{1}{n}\sum_{i=1}^n\left(Y_i-\hat{m}(X_i)\right)\hat{f}_X(X_i)\hat{f}_W(W_i)1_x(X_i)\left[\hat{\Delta}_n(X_i)-\Delta(X_i)\right]^2\frac{G(z;X_i)}{\Delta^2(X_i)\hat{\Delta}_n(X_i)}.
    \end{align*}
    We first discuss the cases under the null and local alternative hypotheses. For the first term, expanding $\hat{m}(\cdot)$, $\hat{f}_X(\cdot)$, $\hat{f}_W(\cdot)$, and $\hat{\Delta}_n(\cdot)$ allows us to rewrite $S_{n11}(\cdot)$ as a $U$-process, from which
    \begin{align}\label{lemmaproof.Shat-S Sn11}
        \sup_{w}\left\vert S_{n11}(w)\right\vert=o_p\left(n^{-1/2}\right)
    \end{align}
    follows by arguments analogous to those in Lemma \ref{lemma.Sn}. With $Y_i-\hat{m}(X_i)$ rewritten as $[Y_i-m(X_i)]-[\hat{m}(X_i)-m(X_i)]$, we can still, by arguments analogous to those in the proof of Lemma \ref{lemma.Sn}, derive
    \begin{align*}
        \mathbb{E}\sup_w\left\vert \left(Y_i-\hat{m}(X_i)\right)\hat{f}_X(X_i)\hat{f}_W(W_i)1_x(X_i)\frac{f_X^4(X_i)G(z;X_i)}{\Delta^2(X_i)\hat{\Delta}_n(X_i)}\right\vert<\infty.
    \end{align*}
    Lemma \ref{lemma.G} further implies that $\sup_x\vert\hat{\Delta}_n(x)-\Delta(x)\vert^2=o_p(n^{-1/2})$, which completes the proof of 
    \begin{align}\label{lemmaproof.Shat-S Sn12}
        \sup_{w}\left\vert S_{n12}(w)\right\vert=o_p\left(n^{-1/2}\right).
    \end{align}
    Combining equations \eqref{lemmaproof.Shat-S Sn1 decomp}, \eqref{lemmaproof.Shat-S Sn11} and \eqref{lemmaproof.Shat-S Sn12} gives $\sup_{w}\vert S_{n1}(w)\vert=o_p(n^{-1/2})$. Under the alternative hypothesis, the corresponding result is $\sup_{w}\vert S_{n1}(w)\vert=o_p(1)$. For $S_{n2}(w)$ and $S_{n3}(w)$, the negligibility can be established in a similar manner by expanding and rewriting them as $U$-processes. For brevity, we omit the detailed proofs. In short, for $S_{n2}(w)$, $\hat{G}_n(\cdot;X_i)$ can be expanded according to its expression and the result in Lemma \ref{lemma.G} can be applied. For $S_{n3}(w)$, both $\hat{G}_n(\cdot;X_i)$ and $\hat{\Delta}_n(\cdot)$ can be expanded and the results in Lemmas \ref{lemma.delta} and \ref{lemma.G} can be invoked. The negligibility of $S_{n1}(w)$, $S_{n2}(w)$ and $S_{n3}(w)$ together leads to the validity of equation \eqref{lemma.Shat-S eq1} under the null and local alternative hypotheses. Similar arguments can be employed to verify equation \eqref{lemma.Shat-S eq2} under the alternative hypothesis.
\end{proof}

\begin{lemma}\label{lemma.Shat-S boot}
    Suppose that Assumptions \ref{ass.sample}--\ref{ass.bandwidth} hold, 
    \begin{align}\label{lemma.Shat-S boot eq1}
        \sup_w\left\vert \hat{S}^\ast_n(w)-S^\ast_n(w)\right\vert = o_p\left(n^{-1/2}\right)
    \end{align}
    regardless of the underlying hypothesis. 
\end{lemma}

\begin{proof}[Proof of Lemma \ref{lemma.Shat-S boot}]
    Following arguments analogous to those in the proof of Lemma \ref{lemma.Shat-S} above, we can decompose the bootstrap version as
    \begin{align*}
        \hat{S}_n^\ast(w)-S_n^\ast(w) = S_{n1}^\ast(w)+S_{n2}^\ast(w)+S_{n3}^\ast(w),
    \end{align*}
    where $S_{nl}^\ast(w)$ represents the bootstrap version of $S_{nl}(w)$ for $l=1,2,3$. The proof is entirely analogous to that in Lemma \ref{lemma.Shat-S}, except that Lemma \ref{lemma.Sn boot} replaces Lemma \ref{lemma.Sn} in the reasoning. The details are omitted for brevity.
\end{proof}

\begin{lemma}\label{lemma.Sn}
    Suppose that Assumptions \ref{ass.sample}--\ref{ass.bandwidth} hold, whether under the null hypothesis or the sequence of local alternatives,
    \begin{align}\label{lemma.Sn null}
        \sup_{w}\left\vert S_n(w) - S_n^{(1)}(w) + T_n^{(2)}(w)\right\vert = o_p\left(n^{-1/2}\right).
    \end{align}
    Under the alternative hypothesis,
    \begin{align}\label{lemma.Sn alt}
        \sup_{w}\left\vert S_n(w) - S_n^{(1)}(w) + T_n^{(2)}(w)\right\vert = o_p\left(1\right).
    \end{align}
\end{lemma}

\begin{proof}[Proof of Lemma \ref{lemma.Sn}]
    We start by decomposing 
    \begin{align}\label{lemmaproof.Sn decomp}
        S_n(w) = S_{n1}(w)-S_{n2}(w),
    \end{align}
    where
    \begin{align*}
        &S_{n1}(w) = \frac{1}{n}\sum_{i=1}^n\left(Y_i-m(X_i)\right)\hat{f}_X(X_i)\hat{f}_W(W_i)1_x(X_i)\Delta^{-1}(X_i)G(z;X_i),\\
        &S_{n2}(w) = \frac{1}{n}\sum_{i=1}^n\left(\hat{m}(X_i)-m(X_i)\right)\hat{f}_X(X_i)\hat{f}_W(W_i)1_x(X_i)\Delta^{-1}(X_i)G(z;X_i).
    \end{align*}
    Since the first term $S_{n1}(w)$ in \eqref{lemmaproof.Sn decomp} can be represented as 
    \begin{align*}
        \frac{1}{n-1}U_n^{(2)}(\psi^{(1)}_{wb})+\frac{n-2}{n-1}U_n^{(3)}(\varphi^{(1)}_{wb}),
    \end{align*}
    where we recall
    \begin{align*}
        &\psi^{(1)}_{wb}(\chi_i,\chi_j) = \left(Y_i-m(X_i)\right)K^2_{ij}L_{ij}1_x(X_i)\Delta^{-1}(X_i)G(z;X_i),\\
        &\varphi^{(1)}_{wb}(\chi_i,\chi_j,\chi_k) = \left(Y_i-m(X_i)\right)K_{ij}K_{ik}L_{ik}1_x(X_i)\Delta^{-1}(X_i)G(z;X_i).
    \end{align*}
    By Lemma \ref{lemma.Unphi1} and \ref{lemma.Unpsi1}, it follows that 
    \begin{align}\label{lemmaproof.Sn first null}
        \sup_{w}\left\vert S_{n1}(w)-S_n^{(1)}(w)\right\vert = o_p\left(n^{-1/2}\right)
    \end{align}
    holds under the null and local alternative hypotheses, and
    \begin{align}\label{lemmaproof.Sn first alt}
        \sup_{w}\left\vert S_{n1}(w)-S_n^{(1)}(w)\right\vert = o_p\left(1\right)
    \end{align}
    holds under the alternative hypothesis. 
    
    For the second term,
    \begin{align}\label{lemmaproof.Sn second decomp}
        S_{n2}(w) = S_{n21}(w)-S_{n22}(w),
    \end{align}
    where
    \begin{align*}
        &S_{n21}(w) = \frac{1}{n(n-1)}\sum_{i=1}^n\sum_{\substack{j=1 \\ j \neq i}}^n\left(Y_j-m(X_j)\right)K_{ij}\hat{f}_W(W_i)1_x(X_i)\Delta^{-1}(X_i)G(z;X_i),\\
        &S_{n22}(w) = \frac{1}{n(n-1)}\sum_{i=1}^n\sum_{\substack{j=1 \\ j \neq i}}^n\left(m(X_i)-m(X_j)\right)K_{ij}\hat{f}_W(W_i)1_x(X_i)\Delta^{-1}(X_i)G(z;X_i).
    \end{align*}
    Similarly, the terms $S_{n21}(w)$ and $S_{n22}(w)$ in \eqref{lemmaproof.Sn second decomp} can be rewritten as
    \begin{align*}
        &S_{n21}(w) = \frac{1}{n-1}U_{n}^{(2)}(\psi_{wb}^{(2)})+\frac{n-2}{n-1}U_{n}^{(3)}(\varphi_{wb}^{(2)})
    \end{align*}
    and
    \begin{align*}
        &S_{n22}(w) = \frac{1}{n-1}U_{n}^{(2)}(\psi_{wb}^{(3)})+\frac{n-2}{n-1}U_{n}^{(3)}(\varphi_{wb}^{(3)}),
    \end{align*}
    respectively, where we still recall
    \begin{align*}
        &\psi_{wb}^{(2)}(\chi_i,\chi_j) = \left(Y_j-m(X_j)\right)K^2_{ij}L_{ij}1_x(X_i)\Delta^{-1}(X_i)G(z;X_i),\\
        &\varphi_{wb}^{(2)}(\chi_i,\chi_j,\chi_k) = \left(Y_j-m(X_j)\right)K_{ij}K_{ik}L_{ik}1_x(X_i)\Delta^{-1}(X_i)G(z;X_i),\\
        &\psi_{wb}^{(3)}(\chi_i,\chi_j) = \left(m(X_i)-m(X_j)\right)K^2_{ij}L_{ij}1_x(X_i)\Delta^{-1}(X_i)G(z;X_i),\\
        &\varphi_{wb}^{(3)}(\chi_i,\chi_j,\chi_k) = \left(m(X_i)-m(X_j)\right)K_{ij}K_{ik}L_{ik}1_x(X_i)\Delta^{-1}(X_i)G(z;X_i).
    \end{align*}
    The desired results
    \begin{align}\label{lemmaproof.Sn second and third}
        \sup_w\left\vert S_{n21}(w)-T_n^{(2)}(w)\right\vert = o_p\left(n^{-1/2}\right)\quad \text{and}\quad \sup_w\left\vert S_{n22}(w)\right\vert = o_p\left(n^{-1/2}\right)
    \end{align}
    follow by combining the conclusions of Lemmas \ref{lemma.Unphi2}, \ref{lemma.Unpsi2}, \ref{lemma.Unphi3}, and \ref{lemma.Unpsi3}. In conclusion, the combination of \eqref{lemmaproof.Sn first null}, \eqref{lemmaproof.Sn second decomp} and \eqref{lemmaproof.Sn second and third} establishes \eqref{lemma.Sn null} under the null and local alternatives, whereas the combination of \eqref{lemmaproof.Sn first alt}, \eqref{lemmaproof.Sn second decomp} and \eqref{lemmaproof.Sn second and third} establishes \eqref{lemma.Sn alt} under the alternative hypothesis. The proof is thus completed.
\end{proof}

\begin{lemma}\label{lemma.Sn boot}
    Suppose that Assumptions \ref{ass.sample}--\ref{ass.bandwidth} hold, 
    \begin{align}\label{lemma.Sn boot eq1}
        \sup_{w}\left\vert S^\ast_n(w) - S_n^{(1)\ast}(w)\right\vert = o_p\left(n^{-1/2}\right)
    \end{align}
    regardless of the underlying hypothesis. 
\end{lemma}

\begin{proof}[Proof of Lemma \ref{lemma.Sn boot}]
    The proof for the bootstrap version follows exactly the same line of reasoning as that of Lemma \ref{lemma.Sn}. We first decompose $S_n^\ast(w)$ into two components, $S_{n1}^\ast(w)$ and $S_{n2}^\ast(w)$, where $S_{n1}^\ast(w)$ and $S_{n2}^\ast(w)$ denote the bootstrap counterparts of $S_{n1}(w)$ and $S_{n2}(w)$ mentioned in the proof of Lemma \ref{lemma.Sn}, respectively. For $S_{n1}^\ast(w)$, we further separate the terms according to whether $j=k$, yielding a main term and a bias term. These two parts can then be handled by applying Lemmas \ref{lemma.Unphi1 boot} and \ref{lemma.Unpsi1}, respectively, which together imply
    \begin{align}\label{lemmaproof.Sn boot sn1}
        \sup_w\left\vert S_{n1}^\ast(w)-{S_n^{(1)}}^\ast(w)\right\vert = o_p\left(n^{-1/2}\right).
    \end{align}
    
    For $S_{n2}^\ast(w)$, we proceed as in Lemma \ref{lemma.Sn} by decomposing it into $S_{n21}^\ast(w)$ and $S_{n22}^\ast(w)$, which represent the bootstrap counterparts of $S_{n21}(w)$ and $S_{n22}(w)$, respectively. Noting
    \begin{align*}
        S_{n21}^\ast(w) = \frac{1}{n(n-1)}\sum_{i=1}^n\sum_{\substack{j=1 \\ j \neq i}}^nV_i\left(Y_j-m(X_j)\right)K_{ij}\hat{f}_W(W_i)1_x(X_i)\Delta^{-1}(X_i)G(z;X_i)
    \end{align*}
    and using the facts that the multiplier has zero mean and is independent of all sample observations, we can invoke the argument in the proof of Lemmas \ref{lemma.Unphi2} and \ref{lemma.Unpsi2}. Combining this reasoning with the result therein regarding $\mathbb{E}_i(\epsilon_jK_{ij})=0$, we obtain
    \begin{align}\label{lemmaproof.Sn boot Sn21}
        \sup_w\left\vert S_{n21}^\ast(w)\right\vert = o_p\left(n^{-1/2}\right).
    \end{align}
    As in the previous arguments, for
    \begin{align*}
        S_{n22}^\ast(w) = \frac{1}{n(n-1)}\sum_{i=1}^n\sum_{\substack{j=1 \\ j \neq i}}^nV_i\left(m(X_i)-m(X_j)\right)K_{ij}\hat{f}_W(W_i)1_x(X_i)\Delta^{-1}(X_i)G(z;X_i),
    \end{align*}
    we only need to analyze the projection conditional on $\chi_i$ because of the properties of the multipliers. Furthermore, the unit-variance property of $V_i$ ensures that all second-moment properties established in Lemmas \ref{lemma.Unphi3} and \ref{lemma.Unpsi3} hold, whereas its zero-mean property guarantees the negligibility of the bias term. Therefore,
    \begin{align}\label{lemmaproof.Sn boot Sn22}
        \sup_w\left\vert S_{n22}^\ast(w)\right\vert = o_p\left(n^{-1/2}\right).
    \end{align}
    Thus, the desired conclusion \eqref{lemma.Sn boot eq1} is established by combining \eqref{lemmaproof.Sn boot sn1}, \eqref{lemmaproof.Sn boot Sn21} and \eqref{lemmaproof.Sn boot Sn22}.
\end{proof}

\begin{lemma}\label{lemma.Unphi1}
    Suppose that Assumptions \ref{ass.sample}--\ref{ass.bandwidth} hold, whether under the null hypothesis or the local alternatives
    \begin{align}\label{lemma.Unphi1 null}
        \sup_{w}\left\vert U_n^{(3)}(\varphi^{(1)}_{wb})-S_n^{(1)}(w)\right\vert = o_p\left(n^{-1/2}\right).
    \end{align}
    Under the alternative hypothesis, 
    \begin{align}\label{lemma.Unphi1 alt}
        \sup_{w}\left\vert U_n^{(3)}(\varphi^{(1)}_{wb})-S_n^{(1)}(w)\right\vert = o_p\left(1\right).
    \end{align}
\end{lemma}

\begin{proof}[Proof of Lemma \ref{lemma.Unphi1}]
    We start by analyzing the main term $U_n^{(3)}(\varphi^{(1)}_{wb})$, which plays a leading role in the asymptotic expansion. Note that 
    \begin{align*}
        \mathbb{E}\left[\sup_{w}\left\vert n^{1/2}U_n^{(3)}(\pi_3\varphi^{(1)}_{wb})\right\vert^2\right]\leq \frac{1}{n^2b^{4q+2p}}O\left(\mathbb{E}\left({\varphi^{(1)}}\right)^2\right) = O\left(\frac{1}{n^2b^{2q+p}}\right)
    \end{align*}
    as the reasoning in the proofs of \cite{delgado2001significance} proposition $1$ and Lemma $2$ suggest, $U_n^{(3)}(\varphi^{(1)}_{wb})$ can be expressed as a third-order $U$-process and the H\'{a}jek projection can then be invoked by applying the symmetrization technique based on permutation invariance. Combining this with the facts 
    \begin{align*}
        \mathbb{E}\left[\sup_{w}\left\vert n^{1/2}U_n^{(2)}(\pi_2\varphi^{(1)}_{wb})\right\vert^2\right] \leq O\left(\frac{1}{nb^{2q+p}}\right) = o\left(1\right)
    \end{align*}
    by the same reasoning as $U_n^{(3)}(\varphi^{(1)}_{wb})$ together with implied by Assumption \ref{ass.bandwidth}, we obtain $n^{1/2}U_n^{(3)}(\varphi^{(1)}_{wb}) = 3n^{1/2}U_n^{(1)}(\pi_1\varphi^{(1)}_{wb})+\mathbb{E}[U_n^{(3)}(\varphi^{(1)}_{wb})]+o_p(1)$ holds uniformly over $w$, where 
    \begin{align}\label{lemmaproof.Sn Hajek 3-order}
        &3U_n^{(1)}(\pi_1\varphi^{(1)}_{wb}) = U_{ni}^{(1)}(\varphi^{(1)}_{wb})+U_{nj}^{(1)}(\varphi^{(1)}_{wb})+ U_{nk}^{(1)}(\varphi^{(1)}_{wb})-3\mathbb{E}\left[U_{n}^{(1)}(\varphi^{(1)}_{wb})\right].
    \end{align}
    For the first term in equation \eqref{lemmaproof.Sn Hajek 3-order}, the independence across samples implies that
    \begin{align*}
        U_{ni}^{(1)}(\varphi^{(1)}_{wb}) = \frac{1}{n}\sum_{i=1}^n\left(Y_i-m(X_i)\right)d_X(X_i)d_W(W_i)1_x(X_i)\Delta^{-1}(X_i)G(z;X_i).
    \end{align*}
    Applying H\"older’s inequality, taking into account Assumption \ref{ass.sample} concerning $\epsilon$ and $\Delta(X)$ together with Lemma \ref{lemma.condition dist},
    \begin{align}\label{lemmaproof.unidecomp term3 variance}
        &\mathbb{E}\left\vert \left(Y_i-m(X_i)\right)f_X^2(X_i)\left[f_{Z\vert X}(W_i)-\frac{d_W(W_i)}{f_X(X_i)}\right]\Delta^{-1}(X_i)\right\vert^2\notag\\
        \leq&\left(\mathbb{E}\left\vert \epsilon\right\vert^{2+\delta_1+\delta_2}\right)^{\frac{2}{2+\delta_1+\delta_2}} \left[\mathbb{E}\left\vert f_{Z\vert X}(W_i)-\frac{d_W(W_i)}{f_X(X_i)}\right\vert^{2+\frac{4+2\delta_2}{\delta_1}}\right]^{\frac{\delta_1}{2+\delta_1+\delta_2}}\left[\mathbb{E}\left\vert f_X^2(X)\Delta^{-1}(X)\right\vert^{2+\frac{4+2\delta_1}{\delta_2}}\right]^{\frac{\delta_2}{2+\delta_1+\delta_2}}\notag\\
        =&o\left(1\right).
    \end{align}
    Moreover, since we can perform the decomposition
    \begin{align}\label{lemmaproof.unidecomp}
        &U_{ni}^{(1)}(\varphi^{(1)}_{wb}) - \frac{1}{n}\sum_{i=1}^n\left(Y_i-m(X_i)\right)f_X(X_i)f_{W}(W_i)1_x(X_i)\Delta^{-1}(X_i)G(z;X_i)\notag\\
        =&\frac{1}{n}\sum_{i=1}^n\left(Y_i-m(X_i)\right)\left[d_X(X_i)-f_X(X_i)\right]\left[f_{Z\vert X}(W_i)-\frac{d_W(W_i)}{f_X(X_i)}\right]f_X(X_i)1_x(X_i)\Delta^{-1}(X_i)G(z;X_i)\notag\\
        &+\frac{1}{n}\sum_{i=1}^n\left(Y_i-m(X_i)\right)\left[d_X(X_i)-f_X(X_i)\right]f_{Z\vert X}(W_i)f_X(X_i)1_x(X_i)\Delta^{-1}(X_i)G(z;X_i)\notag\\
        &+\frac{1}{n}\sum_{i=1}^n\left(Y_i-m(X_i)\right)f^2_X(X_i)\left[f_{Z\vert X}(W_i)-\frac{d_W(W_i)}{f_X(X_i)}\right]1_x(X_i)\Delta^{-1}(X_i)G(z;X_i),
    \end{align}
    it follows that the third empirical process in the expansion can be viewed as being indexed by a VC-type class of functions, with VC characteristics independent of its envelope
    \begin{align*}
        \left\vert \left(Y_i-m(X_i)\right)f^2_X(X_i)\left[f_{Z\vert X}(W_i)-\frac{d_W(W_i)}{f_X(X_i)}\right]\Delta^{-1}(X_i)\right\vert.
    \end{align*}
    Thus, the second-moment bound of the desired empirical process in the expansion is of order $o(1)$ with the application of  Proposition $4$ in \cite{delgado2001significance}. Similar arguments also show that the variances of the envelope functions for other terms in the decomposition are negligible. For the bias term of the third empirical process in the expansion, under the null hypothesis, since we have $\mathbb{E}(Y\vert W)=\mathbb{E}(Y\vert X)$, we obtain
    \begin{align*}
        \mathbb{E}\left[\frac{1}{n}\sum_{i=1}^n\left(Y_i-m(X_i)\right)f^2_X(X_i)\left[f_{Z\vert X}(W_i)-\frac{d_W(W_i)}{f_X(X_i)}\right]1_x(X_i)\Delta^{-1}(X_i)G(z;X_i)\right] = 0.
    \end{align*}
    Under the alternative hypothesis, it suffices to show that
    \begin{align}\label{lemmaproof.unidecomp term3 alt}
        \sup_{w}\left\vert\mathbb{E}\left\{\left(Y_i-m(X_i)\right)f^2_X(X_i)\left[f_{Z\vert X}(W_i)-\frac{d_W(W_i)}{f_X(X_i)}\right]1_x(X_i)\Delta^{-1}(X_i)G(z;X_i)\right\}\right\vert = o\left(1\right).
    \end{align}
    Since the left-hand side of \eqref{lemmaproof.unidecomp term3 alt} is bounded by
    \begin{align}\label{lemmaproof.unidecomp term3 bias}
        &\mathbb{E}\left\vert\left(\mathbb{E}(Y_i\vert W_i)-m(X_i)\right)\left[f_{Z\vert X}(W_i)-\frac{d_W(W_i)}{f_X(X_i)}\right]f_X^2(X_i)\Delta^{-1}(X_i)\right\vert=o\left(1\right),
    \end{align}
    which holds by the same reasoning as the proof of \eqref{lemmaproof.unidecomp term3 variance}, \eqref{lemmaproof.unidecomp term3 alt} follows. Combining the result on the variance with the uniform law of large numbers established in \cite{newey1991uniform}, we obtain
    \begin{align*}
        \sup_{w}\left\vert \frac{1}{n}\sum_{i=1}^n\left(Y_i-m(X_i)\right)f^2_X(X_i)\left[f_{Z\vert X}(W_i)-\frac{d_W(W_i)}{f_X(X_i)}\right]1_x(X_i)\Delta^{-1}(X_i)G(z;X_i)\right\vert=o_p\left(1\right).
    \end{align*}
    Under the local alternative hypothesis specified in \eqref{hyp.local}, replacing $\mathbb{E}(Y\vert W)-m(X)$ in equation \eqref{lemmaproof.unidecomp term3 bias} with $\Psi(W)$
    yields
    \begin{align*}
        \sup_{w}\left\vert \frac{1}{n}\sum_{i=1}^n\Psi(W_i)f^2_X(X_i)\left[f_{Z\vert X}(W_i)-\frac{d_W(W_i)}{f_X(X_i)}\right]1_x(X_i)\Delta^{-1}(X_i)G(z;X_i)\right\vert=o_p\left(1\right)
    \end{align*}
    by the same argument, and hence
    \begin{align*}
        \sup_{w}\left\vert \frac{1}{n}\sum_{i=1}^n\left(Y_i-m(X_i)\right)f^2_X(X_i)\left[f_{Z\vert X}(W_i)-\frac{d_W(W_i)}{f_X(X_i)}\right]1_x(X_i)\Delta^{-1}(X_i)G(z;X_i)\right\vert=o_p\left(n^{-1/2}\right)
    \end{align*}
    follows. The remaining terms in equation \eqref{lemmaproof.unidecomp} can be treated analogously. Therefore, we obtain
    \begin{align}\label{lemmaproof.Unphi1 uni}
        \sup_{w}\left\vert U_{ni}^{(1)}(\varphi^{(1)}_{wb}) - \frac{1}{n}\sum_{i=1}^n\left(Y_i-m(X_i)\right)f_{W}(W_i)f_X(X_i)1_x(X_i)\Delta^{-1}(X_i)G(z;X_i)\right\vert=o_p\left(n^{-1/2}\right)
    \end{align}
    whether under the null or local alternatives invoking Proposition $5$ in \cite{delgado2001significance}. Under the alternative hypothesis, a corresponding convergence-in-probability result can be obtained.
    
    Subsequently, by conditioning on $\chi_i$,
    \begin{align}\label{lemmaproof.unjdecomp}
        U_{nj}^{(1)}(\varphi^{(1)}_{wb}) = &\frac{1}{n}\sum_{j=1}^{n}\mathbb{E}_j\left\{\left(Y_i-m(X_i)\right)K_{ij}d_W(W_i)1_x(X_i)\Delta^{-1}(X_i)G(z;X_i)\right\}\notag\\
        =&\frac{1}{n}\sum_{j=1}^{n}\mathbb{E}_j\left\{\left(Y_i-m(X_i)\right)K_{ij}\left[\frac{d_W(W_i)}{f_X(X_i)}-f_{Z\vert X}(W_i)\right]f_X(X_i)1_x(X_i)\Delta^{-1}(X_i)G(z;X_i)\right\}\notag\\
        &+\frac{1}{n}\sum_{j=1}^{n}\mathbb{E}_j\left[\left(Y_i-m(X_i)\right)K_{ij}f_{Z\vert X}(W_i)f_X(X_i)1_x(X_i)\Delta^{-1}(X_i)G(z;X_i)\right].
    \end{align}
    We first consider the first term of the decomposition. Under the null hypothesis,
    \begin{align*}
        &\mathbb{E}_j\left\{\left(Y_i-m(X_i)\right)K_{ij}\left[\frac{d_W(W_i)}{f_X(X_i)}-f_{Z\vert X}(W_i)\right]f_X(X_i)1_x(X_i)\Delta^{-1}(X_i)G(z;X_i)\right\}\\
        =&\mathbb{E}_j\left\{\left[\mathbb{E}(Y_i\vert W_i)-m(X_i)\right]K_{ij}\left[\frac{d_W(W_i)}{f_X(X_i)}-f_{Z\vert X}(W_i)\right]f_X(X_i)1_x(X_i)\Delta^{-1}(X_i)G(z;X_i)\right\}\\
        =&0.
    \end{align*}
    Under the local alternative hypothesis, we denote 
    \begin{align*}
        &g_\Psi(X_i) = \mathbb{E}\left\{\Psi(X_i)\left[\frac{d_W(W_i)}{f_X(X_i)}-f_{Z\vert X}(W_i)\right]\vert X_i\right\},
    \end{align*}
    from which it follows that
    \begin{align*}
        &\mathbb{E}_j\left\{\left(Y_i-m(X_i)\right)K_{ij}\left[\frac{d_W(W_i)}{f_X(X_i)}-f_{Z\vert X}(W_i)\right]f_X(X_i)1_x(X_i)\Delta^{-1}(X_i)G(z;X_i)\right\}\\
        =&\frac{1}{\sqrt{n}}\mathbb{E}_j\left\{\Psi(W_i)K_{ij}\left[\frac{d_W(W_i)}{f_X(X_i)}-f_{Z\vert X}(W_i)\right]f_X(X_i)1_x(X_i)\Delta^{-1}(X_i)G(z;X_i)\right\}\\
        =&\frac{1}{\sqrt{n}}\mathbb{E}_j\left[g_\Psi(X_i)K_{ij}p_w(X_i)\right].
    \end{align*}
    Consequently, our task reduces to proving that
    \begin{align}\label{lemmaproof.unjdecomp term1 var}
        \sup_{w}\left\vert\frac{1}{n}\sum_{j=1}^n\mathbb{E}_j\left[g_\Psi(X_i)K_{ij}p_w(X_i)\right]\right\vert = o_p\left(1\right).
    \end{align}
    The proof for the bias term, obtained by taking expectations of the entire expression, is identical to that for the first term in \eqref{lemmaproof.unidecomp}. Hence, it suffices to show that
    \begin{align*}
        &\mathbb{E}\sup_{w}\left\vert g_\Psi(X_i)K_{ij}p_w(X_i)\right\vert\\
        \leq&\int \left\vert g_\Psi(x)\right\vert\left\vert \Delta^{-1}(x)f_X^2(x)\right\vert\mathbb{E}\left\vert K_{a}\left(\frac{x-X_j}{a}\right)\right\vert f_X(x)\,dx\\
        =&O\left(\mathbb{E}\left\vert g_\Psi(X)\Delta^{-1}(X)f_X^2(X)\right\vert\right)=o\left(1\right),
    \end{align*}
    where the result follows from Lemmas \ref{lemma.condition dist}, \ref{lemma.G} and \ref{lemma.delta}, together with Lemma $2$ in \cite{delgado2001significance}, and by applying H\"older’s inequality. Under the alternative hypothesis, the same reasoning as for equation \eqref{lemmaproof.unjdecomp term1 var} applies to show that
    \begin{align*}
        \sup_{w}\left\vert \frac{1}{n}\sum_{j=1}^{n}\mathbb{E}_j\left\{\left(Y_i-m(X_i)\right)K_{ij}\left[\frac{d_W(W_i)}{f_X(X_i)}-f_{Z\vert X}(W_i)\right]f_X(X_i)1_x(X_i)\Delta^{-1}(X_i)G(z;X_i)\right\}\right\vert=o_p\left(1\right).
    \end{align*}
    For the second term in the decomposition \eqref{lemmaproof.unjdecomp}, it can be shown that it equals zero under the null hypothesis, and is equivalent to
    \begin{align}\label{lemmaproof.unjdecomp second}
        &\frac{1}{n}\sum_{j=1}^n\mathbb{E}_j\left\{\left[\mathbb{E}(Y_i\vert W_i)-m(X_i)\right]f_{Z\vert X}(W_i)K_{ij}f_X(X_i)1_x(X_i)\Delta^{-1}(X_i)G(z;X_i)\right\}\notag\\
        =&\frac{1}{n}\sum_{j=1}^n\mathbb{E}_j\left\{\left[\mathbb{E}(Y_i\vert W_i)-m(X_i)\right]f_{Z\vert X}(W_i)\left[K_{ij}-f_X(X_i)\right]f_X(X_i)1_x(X_i)\Delta^{-1}(X_i)G(z;X_i)\right\}\notag\\
        &+\frac{1}{n}\sum_{j=1}^n\mathbb{E}_j\left\{\left[\mathbb{E}(Y_i\vert W_i)-m(X_i)\right]f_{W}(W_i)f_X(X_i)1_x(X_i)\Delta^{-1}(X_i)G(z;X_i)\right\}
    \end{align}
    under the local and the global alternative hypotheses. Redefining
    \begin{align*}
        m_\Psi(X_i) = \mathbb{E}\left[\Psi(W_i)f_{Z\vert X}(W_i)\vert X_i\right]
    \end{align*}
    while keeping $p_w(X_i)$ as previously defined, under the local alternative hypothesis, we can rewrite the first term in the expansion \eqref{lemmaproof.unjdecomp second} as
    \begin{align}\label{lemmaproof.unjdecomp second first term}
        n^{-3/2}\sum_{j=1}^n \int m_\Psi(x)p_w(x)\left[K_a\left(\frac{x-X_j}{a}\right)-f_X(x)\right]f_X(x)\,dx.
    \end{align}
    By an argument analogous to that for equation \eqref{lemmaproof.unjdecomp term1 var}, we can derive
    \begin{align*}
        \mathbb{E}\int \left\vert m_\Psi(x)\Delta^{-1}(x)f^2_X(x)\left[K_a\left(\frac{x-X_j}{a}\right)-f_X(x)\right]\right\vert f_X(x)\,dx<\infty.
    \end{align*}
    The Law of Large Numbers in \cite{newey1991uniform} implies that equation \eqref{lemmaproof.unjdecomp second first term} converges in probability to
    \begin{align*}
        n^{-1/2}\int m_\Psi(x)p_w(x)\mathbb{E}\left[K_a\left(\frac{x-X_j}{a}\right)-f_X(x)\right]f_X(x)\,dx.
    \end{align*}
    Using Lemmas \ref{lemma.G} and \ref{lemma.delta} together with Lemma $4$ in \cite{delgado2001significance}, we obtain
    \begin{align*}
        n^{-1/2}\int m_\Psi(x)p_w(x)\mathbb{E}\left[K_a\left(\frac{x-X_j}{a}\right)-f_X(x)\right]f_X(x)\,dx = o\left(n^{-1/2}\right),
    \end{align*}
    and thus the result
    \begin{align*}
        \sup_{w}\left\vert U_{nj}^{(1)}(\varphi^{(1)}_{wb}) - \frac{1}{n}\sum_{j=1}^n\mathbb{E}_j\left[\left(Y_i-m(X_i)\right)f_{W}(W_i)f_X(X_i)1_x(X_i)\Delta^{-1}(X_i)G(z;X_i)\right]\right\vert=o_p\left(n^{-1/2}\right)
    \end{align*}
    is proved and subsequently,
    \begin{align}\label{lemmaproof.Sn unj}
        \sup_{w}\left\vert U_{nj}^{(1)}(\varphi^{(1)}_{wb}) - \mathbb{E}\left[U_{n}^{(3)}(\varphi^{(1)}_{wb})\right]\right\vert=o_p\left(n^{-1/2}\right)
    \end{align}
    follows by combining the conclusion concerning the bias term, whether under the null or local alternatives.

    For $U_{nk}^{(1)}(\varphi^{(1)}_{wb})$, Similar to the proof of $U_{nj}^{(1)}(\varphi^{(1)}_{wb})$, by conditioning on $\chi_i$, we obtain
    \begin{align}\label{lemmaproof.unkdecomp}
        U_{nk}^{(1)}(\varphi^{(1)}_{wb}) = &\frac{1}{n}\sum_{k=1}^{n}\mathbb{E}_k\left\{\left(Y_i-m(X_i)\right)d_X(X_i)K_{ik}L_{ik}1_x(X_i)\Delta^{-1}(X_i)G(z;X_i)\right\}\notag\\
        =&\frac{1}{n}\sum_{k=1}^{n}\mathbb{E}_k\left\{\left(Y_i-m(X_i)\right)\left[d_X(X_i)-f_{X}(X_i)\right]K_{ik}L_{ik}1_x(X_i)\Delta^{-1}(X_i)G(z;X_i)\right\}\notag\\
        &+\frac{1}{n}\sum_{k=1}^{n}\mathbb{E}_k\left[\left(Y_i-m(X_i)\right)K_{ik}L_{ik}f_{X}(X_i)1_x(X_i)\Delta^{-1}(X_i)G(z;X_i)\right].
    \end{align}
    By similar arguments, conditioning on $W_i$ under the null hypothesis implies that $U_{nk}^{(1)}(\varphi^{(1)}_{wb})=0$, and under the local alternatives it is equivalent to 
    \begin{align}\label{lemmaproof.unkdecomp eqv}
        &n^{-3/2}\sum_{k=1}^{n}\mathbb{E}_k\left\{\Psi(W_i)\left[d_X(X_i)-f_{X}(X_i)\right]K_{ik}L_{ik}1_x(X_i)\Delta^{-1}(X_i)G(z;X_i)\right\}\notag\\
        &+n^{-3/2}\sum_{k=1}^{n}\mathbb{E}_k\left[\Psi(W_i)K_{ik}L_{ik}f_{X}(X_i)1_x(X_i)\Delta^{-1}(X_i)G(z;X_i)\right].
    \end{align}
    The first term in the expansion \eqref{lemmaproof.unkdecomp eqv} equals to
    \begin{align*}
        n^{-3/2}\sum_{k=1}^{n}\mathbb{E}_k\left\{\mathbb{E}_k\left[\Psi(W_i)L_{ik}\vert X_i\right]\left[\frac{d_X(X_i)}{f_{X}(X_i)}-1\right]K_{ik}p_w(X_i)\right\}
    \end{align*}
    and the bias term can be proved to be degenerate by arguments entirely analogous to those used for equation \eqref{lemmaproof.unjdecomp term1 var}. Hence, by arguments similar to those used for equation \eqref{lemmaproof.unjdecomp term1 var} once again, it suffices to show that
    \begin{align*}
        \int\left\vert \frac{d_X(x)}{f_{X}(x)}-1\right\vert\left\vert \Delta^{-1}(x)f_X^2(x)\right\vert\mathbb{E}\left\vert \mathbb{E}\left[\Psi(x,Z_i)L_b\left(\frac{Z_i-Z_k}{b}\right)\vert Z_k\right]K_{a}\left(\frac{x-X_k}{a}\right)\right\vert f_X(x)\,dx
    \end{align*}
    is negligible, that is, of order $o(1)$. Using Lemma \ref{lemma.bound double kernel}, the above expression reduces to showing the negligibility of
    \begin{align*}
        \int\left\vert \frac{d_X(x)}{f_{X}(x)}-1\right\vert\left\vert \Delta^{-1}(x)f_X^2(x)\right\vert f_X(x)\,dx
    \end{align*}
    which follows from the negligibility of $\mathbb{E}[d_X(X)/f_X(X)-1]$, as implied by arguments analogous to those in Lemma \ref{lemma.condition dist}, and from the boundedness of the integral of $\Delta^{-1}(x)f_X^2(x)$ ensured by Assumption \ref{ass.sample}. For the second term in decomposition \eqref{lemmaproof.unkdecomp eqv}, 
    \begin{align*}
        \mathbb{E}\left\vert \sup_{w}\mathbb{E}_k\left\{\Psi(W_i)\left[K_{ik}L_{ik}-f_W(W_i)\right]p_w(X_i)\right\}\right\vert = o_p\left(1\right)
    \end{align*}
    can be established by arguments similar to those used in the proof of equation \eqref{lemmaproof.unjdecomp term1 var}, which consequently leads to the fact that
    \begin{align}\label{lemmaproof.Sn unk}
        \sup_{w}\left\vert U_{nk}^{(1)}(\varphi^{(1)}_{wb}) - \mathbb{E}\left[U_{n}^{(3)}(\varphi^{(1)}_{wb})\right]\right\vert=o_p\left(n^{-1/2}\right)
    \end{align}
    holds by combining the conclusion concerning the bias term, whether under the null or local alternatives.
    
    Subsequently, by combining equations \eqref{lemmaproof.Sn Hajek 3-order}, \eqref{lemmaproof.Unphi1 uni}, \eqref{lemmaproof.Sn unj} and \eqref{lemmaproof.Sn unk}, \eqref{lemma.Unphi1 null} follows. The same reasoning as in the proof under the local alternatives shows that \eqref{lemma.Unphi1 alt} also holds.
\end{proof}

\begin{lemma}\label{lemma.Unphi1 boot}
    Suppose that Assumptions \ref{ass.sample}--\ref{ass.bandwidth} hold,
    \begin{align}\label{lemma.Unphi1 boot null}
        \sup_{w}\left\vert {U_n^{(3)}}^\ast(\varphi^{(1)}_{wb})-{S_n^{(1)}}^\ast(w)\right\vert = o_p\left(n^{-1/2}\right).
    \end{align}
    regardless of the underlying hypothesis.
\end{lemma}

\begin{proof}[Proof of Lemma \ref{lemma.Unphi1 boot}]
    Owing to the zero-mean and unit-variance properties of the multiplier, together with its independence from the sample $\chi_i$ for any $i$, reasoning similar to that employed in the proof of \eqref{lemmaproof.Sn Hajek 3-order} establishes
    \begin{align}\label{lemmaproof.Sn Hajek 3-order boot}
        &\sup_{w}\left\vert {U_n^{(3)}}^\ast(\varphi^{(1)}_{wb}) - {U_{ni}^\ast}^{(1)}(\varphi^{(1)}_{wb})\right\vert =o_p\left(n^{-1/2}\right),
    \end{align}
    where ${U_{ni}^\ast}^{(1)}(\varphi^{(1)}_{wb})$ is the bootstrap version of $U_{ni}^{(1)}(\varphi^{(1)}_{wb})$ mentioned in proof of Lemma \ref{lemma.Unphi1}. Still, the distributional properties of the multiplier and its independence from the sample allow the reasoning for $U_{ni}^{(1)}(\varphi^{(1)}_{wb})$ to be extended to ${U_{ni}^\ast}^{(1)}(\varphi^{(1)}_{wb})$, from which
    \begin{align}\label{lemmaproof.Sn boot uni boot}
        \sup_{w}\left\vert {U_{ni}^\ast}^{(1)}(\varphi^{(1)}_{wb}) - {S_n^{(1)}}^\ast(w)\right\vert=o_p\left(n^{-1/2}\right)
    \end{align}
    follows under the null or local alternatives. Consequently, equations \eqref{lemmaproof.Sn Hajek 3-order boot} and \eqref{lemmaproof.Sn boot uni boot} together imply \eqref{lemma.Unphi1 boot null} under the null or local alternative hypotheses. Moreover, the zero-mean property of the multiplier implies that the bootstrap version of the empirical processes in equation \eqref{lemmaproof.unidecomp} is also centered under the alternative hypothesis. Together with the square integrability of their envelope functions established in the proof of Lemma \eqref{lemma.Unphi1}, equation \eqref{lemma.Unphi1 boot null} holds under the alternative hypothesis as well.
\end{proof}

\begin{lemma}\label{lemma.Unpsi1}
    Suppose that Assumptions \ref{ass.sample}--\ref{ass.bandwidth} hold,
    \begin{align}\label{lemma.Unpsi1 null}
        \sup_{w}\left\vert n^{-1/2}U_n^{(2)}(\psi^{(1)}_{wb})\right\vert = o_p\left(1\right).
    \end{align}
    For the bootstrap version, 
    \begin{align}\label{lemma.Unpsi1 boot}
        \sup_{w}\left\vert n^{-1/2}{U_n^{(2)}}^\ast(\psi^{(1)}_{wb})\right\vert = o_p\left(1\right).
    \end{align}
\end{lemma}

\begin{proof}[Proof of Lemma \ref{lemma.Unpsi1}]
    Following an argument analogous to that used in establishing equation \eqref{lemmaproof.Sn Hajek 3-order}, and based on
    \begin{align*}
        \mathbb{E}\left[\sup_{w}\left\vert n^{1/2}U_n^{(2)}(\pi_2\psi^{(1)}_{wb})\right\vert^2\right]\leq \frac{1}{n^2b^{4q+2p}}O\left(\mathbb{E}\left(\psi^{(1)}\right)^2\right) = O\left(\frac{1}{n^2b^{2q+p}}\right),
    \end{align*}
    which can be proved in a similar manner, we obtain the decomposition
    \begin{align}\label{lemmaproof.Unpsi1 decomp}
        \sup_w\left\vert U_n^{(2)}(\psi^{(1)}_{wb}) - U_{ni}^{(1)}(\psi^{(1)}_{wb})-U_{nj}^{(1)}(\psi^{(1)}_{wb})+\mathbb{E}\left[U_n^{(2)}(\psi^{(1)}_{wb})\right]\right\vert= o_p\left(n^{-1/2}\right).
    \end{align}
    For the first term in decomposition, following the same reasoning as in the proof of \eqref{lemmaproof.Unphi1 uni} and applying Lemma \ref{lemma.condition dist}, we obtain that
    \begin{align}\label{lemmaproof.Unpsi1 pi eq1}
        \sup_{w}\left\vert U_{ni}^{(1)}(\psi^{(1)}_{wb}) - S_n^{(1)}(w)\right\vert=o_p\left(1\right)
    \end{align}
    holds under any hypothesis. According to Lemma \ref{lemma.G}, $S_n^{(1)}(w)$ can be viewed as an empirical process indexed by a VC-type class of functions whose VC characteristics do not depend on $\epsilon_if^2_X(X_i)f_W(W_i)\Delta^{-1}(X_i)$. The envelope of this class corresponds to the absolute value of the expression, and Assumption \ref{ass.sample} ensures that its second moment is finite. Hence, when \eqref{lemmaproof.Unpsi1 pi eq1} is taken into consideration,
    \begin{align}\label{lemmaproof.Unpsi1 pi}
        \sup_{w}\left\vert n^{-1/2}U_{ni}^{(1)}(\psi^{(1)}_{wb})\right\vert=o_p\left(1\right).
    \end{align}
    The proof for the second term is entirely analogous to that of \eqref{lemmaproof.Sn unj}. After obtaining 
    \begin{align*}
        \sup_{w}\left\vert U_{nj}^{(1)}(\psi^{(1)}_{wb}) - \mathbb{E}\left[S_n^{(1)}(w)\right]\right\vert=o_p\left(1\right),
    \end{align*}
    the boundedness of $\mathbb{E}[S_n^{(1)}(w)]$ immediately implies
    \begin{align}\label{lemmaproof.Unpsi1 pj}
        \sup_{w}\left\vert n^{-1/2}U_{nj}^{(1)}(\psi^{(1)}_{wb})\right\vert=o_p\left(1\right).
    \end{align}
    By combining \eqref{lemmaproof.Unpsi1 decomp}, \eqref{lemmaproof.Unpsi1 pi} and \eqref{lemmaproof.Unpsi1 pj}, \eqref{lemma.Unpsi1 null} is established.

    For the bootstrap version of $U_n^{(2)}(\psi^{(1)}_{wb})$, the zero-mean, unit-variance property of the multiplier and its independence from any sample $\chi_i$ allow us to obtain the decomposition
    \begin{align*}
        \sup_w\left\vert {U_n^{(2)}}^\ast(\psi^{(1)}_{wb}) - {U_{ni}^{(1)}}^\ast(\psi^{(1)}_{wb})\right\vert= o_p\left(n^{-1/2}\right).
    \end{align*}
    By analogous arguments to those used in the proof of Lemma \ref{lemma.Unphi1 boot}, we can then establish
    \begin{align*}
        \sup_{w}\left\vert n^{-1/2}{U_{ni}^{(1)}}^\ast(\psi^{(1)}_{wb})\right\vert=o_p\left(1\right),
    \end{align*}
    thereby completing the proof of \eqref{lemma.Unpsi1 boot}.
\end{proof}

\begin{lemma}\label{lemma.Unphi2}
    Suppose that Assumptions \ref{ass.sample}--\ref{ass.bandwidth} hold, 
    \begin{align}\label{lemma.Unphi2 null}
        \sup_{w}\left\vert U_n^{(3)}(\varphi^{(2)}_{wb})-T_n^{(2)}(w)\right\vert = o_p\left(n^{-1/2}\right)
    \end{align}
    regardless of the underlying hypothesis.
\end{lemma}

\begin{proof}[Proof of Lemma \ref{lemma.Unphi2}]
    Owing to arguments similar to those used for \eqref{lemmaproof.Sn Hajek 3-order}, we first show that
    \begin{align*}
        \mathbb{E}\left[\sup_{w}\left\vert n^{1/2}U_n^{(3)}(\pi_3\varphi^{(2)}_{wb})\right\vert^2\right]=o(1) \quad \text{and} \quad \mathbb{E}\left[\sup_{w}\left\vert n^{1/2}U_n^{(2)}(\pi_2\varphi^{(2)}_{wb})\right\vert^2\right]=o(1),
    \end{align*}
    which in turn implies
    \begin{align}\label{lemmaproof.Unphi2 decomp}
        \sup_{w}\left\vert U_n^{(3)}(\varphi^{(2)}_{wb})-U_{ni}^{(1)}(\varphi^{(2)}_{wb})-U_{nj}^{(1)}(\varphi^{(2)}_{wb})-U_{nk}^{(1)}(\varphi^{(2)}_{wb})+2\mathbb{E}\left[U_n^{(3)}(\varphi^{(2)}_{wb})\right]\right\vert = o_p\left(n^{-1/2}\right).
    \end{align}
    For the first term $U_{ni}^{(1)}(\varphi^{(2)}_{wb})$, the independence of the sample observations allows $U_{ni}^{(1)}(\varphi^{(2)}_{wb})$ to be rewritten as
    \begin{align}\label{lemmaproof.Unphi2 pi}
        U_{ni}^{(1)}(\varphi^{(2)}_{wb}) = \frac{1}{n}\sum_{i=1}^n\mathbb{E}_i\left(\epsilon_jK_{ij}\right)d_W(W_i)1_x(X_i)\Delta^{-1}(X_i)G(z;X_i)=0,
    \end{align}
    where the last equality follows from the fact that $\mathbb{E}_i(\epsilon_jK_{ij})=0$, which holds by definition of $\epsilon_j$ since conditioning on $X_j$ yields $\mathbb{E}_i(\epsilon_jK_{ij})=\mathbb{E}_i[\mathbb{E}_{}(\epsilon_j\vert X_j)K_{ij}]=0$. A similar line of reasoning can be applied to the third term $U_{nk}^{(1)}(\varphi^{(2)}_{wb})$ in decomposition \eqref{lemmaproof.Unphi2 decomp}. By conditioning on $\chi_i$, 
    \begin{align}\label{lemmaproof.Unphi2 pk}
        U_{nk}^{(1)}(\varphi^{(2)}_{wb}) = \frac{1}{n}\sum_{k=1}^n\mathbb{E}_k\left[\mathbb{E}_{i}\left(\epsilon_jK_{ij}\right)K_{ik}L_{ik}1_x(X_i)\Delta^{-1}(X_i)G(z;X_i)\right] = 0.
    \end{align}
    From equations \eqref{lemmaproof.Unphi2 decomp}, \eqref{lemmaproof.Unphi2 pi} and \eqref{lemmaproof.Unphi2 pk}, it follows that
    \begin{align}\label{lemmaproof.Unphi2 decomp eqv}
        \sup_{w}\left\vert U_n^{(3)}(\varphi^{(2)}_{wb})-U_{nj}^{(1)}(\varphi^{(2)}_{wb})\right\vert = o_p\left(n^{-1/2}\right)
    \end{align}
    regardless of the underlying hypothesis. Given the independence of the sample observations, we further obtain
    \begin{align*}
        U_{nj}^{(1)}(\varphi^{(2)}_{wb}) = \frac{1}{n}\sum_{j=1}^n\epsilon_j\mathbb{E}_j\left[K_{ij}K_{ik}L_{ik}1_x(X_i)\Delta^{-1}(X_i)G(z;X_i)\right].
    \end{align*}
    By noting that
    \begin{align*}
        &\mathbb{E}_j\left[K_{ij}f_W(W_i)1_x(X_i)\Delta^{-1}(X_i)G(z;X_i)\right] = \mathbb{E}_j\left\{K_{ij}\mathbb{E}\left[f_W(W_i)\vert X_i\right]1_x(X_i)\Delta^{-1}(X_i)G(z;X_i)\right\}\\
        =&\mathbb{E}_j\left[\frac{K_{ij}}{f_X(X_i)}\Delta(X_i)1_x(X_i)\Delta^{-1}(X_i)G(z;X_i)\right] = \mathbb{E}_j\left[\frac{K_{ij}}{f_X(X_i)}1_x(X_i)G(z;X_i)\right],
    \end{align*}
    we obtain
    \begin{align}\label{lemmaproof.Unphi2 pj}
        U_{nj}^{(1)}(\varphi^{(2)}_{wb})-T_n^{(2)}(w)=&\frac{1}{n}\sum_{j=1}^n\epsilon_j\mathbb{E}_j\left\{K_{ij}\left[\mathbb{E}_i\left(K_{ik}L_{ik}\right)-f_W(W_i)\right]1_x(X_i)\Delta^{-1}(X_i)G(z;X_i)\right\}\notag\\
        &+\frac{1}{n}\sum_{j=1}^n\epsilon_j\left\{\mathbb{E}_j\left[\frac{K_{ij}}{f_X(X_i)}1_x(X_i)G(z;X_i)\right]-1_x(X_j)G(z;X_j)\right\}.
    \end{align}
    By denoting $g_f(X_i) = \mathbb{E}[K_{ik}L_{ik}/f_X(X_i)-f_{Z\vert X}(W_i)\vert X_i]$ and recalling the notation $p_w(X_i)$ in the proof of Lemma \ref{lemma.Unphi1}, the first term in the expansion of equation \eqref{lemmaproof.Unphi2 pj} can be equivalently rewritten as
    \begin{align*}
        \frac{1}{n}\sum_{j=1}^n\epsilon_j\mathbb{E}_j\left[K_{ij}g_f(X_i)p_w(X_i)\right].
    \end{align*}
    Owing to
    \begin{align*}
        &\mathbb{E}_j\left[K_{ij}g_f(X_i)p_w(X_i)\right]=\int K_a\left(\frac{u-X_j}{a}\right)g_f(u)f_X(u)1_x(u)\left[\Delta^{-1}(u)f^2_X(u)\right]\frac{G(z;u)}{f_X(u)}\,du
    \end{align*}
    and the Lipschitz continuity of $f_X(\cdot)$, $g_f(\cdot)$, $\Delta^{-1}(u)f_X^2(u)$ and  $G(z;u)/f(u)$ (as established in Assumption \ref{ass.structual}, Lemmas \ref{lemma.condition dist}, \ref{lemma.delta} and \ref{lemma.G}, respectively),
    \begin{align*}
        \mathbb{E}\left\vert \mathbb{E}_j\left[K_{ij}g_f(X_i)p_w(X_i)\right]-g_f(X_j)f_X(X_j)p_w(X_i)\right\vert^{2+\frac{4}{\delta_1+\delta_2}}=o\left(1\right)
    \end{align*}
    follows by the same line of reasoning as Proposition $3$ in \cite{delgado2001significance}. To analyze the second moment of the first term in decomposition \eqref{lemmaproof.Unphi2 pj}, we apply H\"older’s inequality, from which
    \begin{align*}
        &\mathbb{E}\left\vert \epsilon_j\mathbb{E}_j\left[K_{ij}g_f(X_i)p_w(X_i)\right]-\epsilon_jg_f(X_j)p_w(X_j)\right\vert^2\\
        \leq&\left(\mathbb{E}\left\vert \epsilon\right\vert^{2+\delta_1+\delta_2}\right)^{\frac{2}{2+\delta_1+\delta_2}}\left(\mathbb{E}\left\vert \mathbb{E}_j\left[K_{ij}g_f(X_i)p_w(X_i)\right]-g_f(X_j)f_X(X_j)p_w(X_i)\right\vert^{2+\frac{4}{\delta_1+\delta_2}}\right)^{\frac{\delta_1+\delta_2}{2+\delta_1+\delta_2}}.
    \end{align*}
    Meanwhile, by Lemma \ref{lemma.condition dist} we obtain $\mathbb{E}\vert \epsilon g_f(X)p_w(X)\vert^2=o(1)$, which establishes the negligibility of the second moment. Combining this with the earlier result for the bias term leads to
    \begin{align}\label{lemmaproof.Unphi2 pj first}
        \sup_w\left\vert\frac{1}{n}\sum_{j=1}^n\epsilon_j\mathbb{E}_j\left\{K_{ij}\left[\mathbb{E}_i\left(K_{ik}L_{ik}\right)-f_W(W_i)\right]1_x(X_i)\Delta^{-1}(X_i)G(z;X_i)\right\}\right\vert=o_p\left(n^{-1/2}\right).
    \end{align}
    Following the same reasoning, we first obtain
    \begin{align*}
        \mathbb{E}\left\vert \mathbb{E}_j\left[\frac{K_{ij}}{f_X(X_i)}1_x(X_i)G(z;X_i)\right]-1_x(X_j)G(z;X_j)\right\vert^{2+\frac{4}{\delta_1+\delta_2}}=o\left(1\right),
    \end{align*}
    and 
    \begin{align}\label{lemmaproof.Unphi2 pj second}
        \sup_w\left\vert\frac{1}{n}\sum_{j=1}^n\epsilon_j\left\{\mathbb{E}_j\left[\frac{K_{ij}}{f_X(X_i)}1_x(X_i)G(z;X_i)\right]-1_x(X_j)G(z;X_j)\right\}\right\vert=o_p\left(n^{-1/2}\right)
    \end{align}
    can be established by H\"older’s inequality. By combining \eqref{lemmaproof.Unphi2 pj}, \eqref{lemmaproof.Unphi2 pj first} and \eqref{lemmaproof.Unphi2 pj second}, we obtain
    \begin{align*}
        \sup_w\left\vert \epsilon_j\mathbb{E}_j\left[K_{ij}K_{ik}L_{ik}\frac{p_w(X_i)}{f_X(X_i)}-1_x(X_j)G(z;X_j)\right]\right\vert^2 = o_p\left(n^{-1/2}\right),
    \end{align*}
    which in turn implies
    \begin{align*}
        \sup_w\left\vert U_{nj}^{(1)}(\varphi^{(2)}_{wb})-\mathbb{E}\left[U_{nj}^{(1)}(\varphi^{(2)}_{wb})\right]-\left\{T_n^{(2)}(w)-\mathbb{E}\left[T_n^{(2)}(w)\right]\right\}\right\vert = o_p\left(n^{-1/2}\right).
    \end{align*}
    Noting that
    \begin{align*}
        \mathbb{E}\left[U_{nj}^{(1)}(\varphi^{(2)}_{wb})\right] =  \mathbb{E}\left[\mathbb{E}_i\left(\epsilon_jK_{ij}\right)\mathbb{E}_i\left(K_{ik}L_{ik}\right)1_x(X_i)\Delta^{-1}(X_i)G(z;X_i)\right]=0
    \end{align*}
    and 
    \begin{align*}
        \mathbb{E}\left[T_n^{(2)}(w)\right] = \mathbb{E}\left[\mathbb{E}\left(\epsilon_i\vert X_i\right)1_x(X_i)G(z;X_i)\right] = 0,
    \end{align*}
    \eqref{lemma.Unphi2 null} follows.
\end{proof}

\begin{lemma}\label{lemma.Unpsi2}
    Suppose that Assumptions \ref{ass.sample}--\ref{ass.bandwidth} hold,
    \begin{align}\label{lemma.Unpsi2 null}
        \sup_{w}\left\vert n^{-1/2}U_n^{(2)}(\psi^{(2)}_{wb})\right\vert = o_p\left(1\right)
    \end{align}
    regardless of the underlying hypothesis.
\end{lemma}

\begin{proof}[Proof of Lemma \ref{lemma.Unpsi2}]
    Just as the proof of Lemma \ref{lemma.Unpsi1} follows that of Lemma \ref{lemma.Unphi1}, the proof of this lemma can be carried out analogously to Lemma \ref{lemma.Unphi2}, and the details are therefore omitted.
\end{proof}

\begin{lemma}\label{lemma.Unphi3}
    Suppose that Assumptions \ref{ass.sample}--\ref{ass.bandwidth} hold, 
    \begin{align}\label{lemma.Unphi3 null}
        \sup_{w}\left\vert U_n^{(3)}(\varphi^{(3)}_{wb})\right\vert = o_p\left(n^{-1/2}\right)
    \end{align}
    regardless of the underlying hypothesis.
\end{lemma}

\begin{proof}[Proof of Lemma \ref{lemma.Unphi3}]
    Using the Lipschitz continuity of $m(\cdot)$ and the moment conditions stated in Assumption \ref{ass.structual}, and by applying arguments similar to those used in the proof of equation \eqref{lemmaproof.Sn Hajek 3-order}, we can first establish
    \begin{align*}
        \mathbb{E}\left[\sup_{w}\left\vert n^{1/2}U_n^{(3)}(\pi_3\varphi^{(3)}_{wb})\right\vert^2\right]=o(1) \quad \text{and} \quad \mathbb{E}\left[\sup_{w}\left\vert n^{1/2}U_n^{(2)}(\pi_2\varphi^{(3)}_{wb})\right\vert^2\right]=o(1),
    \end{align*}
    and subsequently obtain
    \begin{align}\label{lemmaproof.Unphi3 decomp}
        \sup_{w}\left\vert U_n^{(3)}(\varphi^{(3)}_{wb})-U_{ni}^{(1)}(\varphi^{(3)}_{wb})-U_{nj}^{(1)}(\varphi^{(3)}_{wb})-U_{nk}^{(1)}(\varphi^{(3)}_{wb})+2\mathbb{E}\left[U_n^{(3)}(\varphi^{(3)}_{wb})\right]\right\vert = o_p\left(n^{-1/2}\right).
    \end{align}
    Regarding the first term $U_{ni}^{(1)}(\varphi^{(3)}_{wb})$, the independence across observations ensures that it can be equivalently expressed as
    \begin{align*}
        U_{ni}^{(1)}(\varphi^{(3)}_{wb}) = \frac{1}{n}\sum_{i=1}^{n}\mathbb{E}_i\left[\left(m(X_i)-m(X_j)\right)K_{ij}\right]d_W(W_i)1_x(X_i)\Delta^{-1}(X_i)G(z;X_i).
    \end{align*}
    By invoking Lemma $5$ from \cite{delgado2001significance} and employing arguments parallel to those used in establishing \eqref{lemmaproof.Unphi1 uni}, the negligibility of the bias term and the second-order moment can be verified with Assumptions \ref{ass.sample}--\ref{ass.bandwidth}, which leads to
    \begin{align}\label{lemmaproof.Unphi3 pi}
        \sup_w\left\vert U_{ni}^{(1)}(\varphi^{(3)}_{wb})\right\vert = o_p\left(n^{-1/2}\right).
    \end{align}
    For the second term $U_{nj}^{(1)}(\varphi^{(3)}_{wb})$ in decomposition \eqref{lemmaproof.Unphi3 decomp}, by conditioning on $\chi_i$ and using the independence across samples once again, we obtain
    \begin{align*}
        U_{nj}^{(1)}(\varphi^{(3)}_{wb}) = \frac{1}{n}\sum_{j=1}^{n}\mathbb{E}_j\left[\left(m(X_i)-m(X_j)\right)K_{ij}d_W(W_i)1_x(X_i)\Delta^{-1}(X_i)G(z;X_i)\right].
    \end{align*}
    Following the definition of $g_f(\cdot)$ introduced in Lemma \ref{lemma.Unphi2} and recalling that of $p_w(\cdot)$, the expression can be further expressed as
    \begin{align*}
        U_{nj}^{(1)}(\varphi^{(3)}_{wb}) = &\frac{1}{n}\sum_{j=1}^{n}\mathbb{E}_j\left[\left(m(X_i)-m(X_j)\right)K_{ij}g_f(X_i)p_w(X_i)\right]\\
        &+\frac{1}{n}\sum_{j=1}^{n}\mathbb{E}_j\left[\left(m(X_i)-m(X_j)\right)K_{ij}\frac{\Delta(X_i)}{f^2(X_i)}p_w(X_i)\right].
    \end{align*}
    Noting the result on $g_f(\cdot)$ established in Lemma \ref{lemma.Unphi2} and the restriction on $f_{Z\vert X}(\cdot)$ imposed by Assumption \ref{ass.structual}, we can similarly show
    \begin{align}\label{lemmaproof.Unphi3 pj}
        \sup_w\left\vert U_{nj}^{(1)}(\varphi^{(3)}_{wb})\right\vert = o_p\left(n^{-1/2}\right).
    \end{align}
    The third term $U_{nk}^{(1)}(\varphi^{(3)}_{wb})$ can likewise be reformulated as
    \begin{align*}
        &\frac{1}{n}\sum_{k=1}^{n}\mathbb{E}_k\left\{\mathbb{E}_i\left[\left(m(X_i)-m(X_j)\right)K_{ij}\right]\frac{K_{ik}}{f_X(X_i)}\mathbb{E}_k\left(L_{ik}\vert X_i\right)p_w(X_)\right\}\\
        =&\frac{1}{n}\sum_{k=1}^{n}\mathbb{E}_k\left\{\mathbb{E}_i\left[\left(m(X_i)-m(X_j)\right)K_{ij}\right]\frac{K_{ik}}{f_X(X_i)}\left[\mathbb{E}_k\left(L_{ik}\vert X_i\right)-f_{Z\vert X}(Z_k\vert X_i)\right]p_w(X_i)\right\}\\
        &+\frac{1}{n}\sum_{k=1}^{n}\mathbb{E}_k\left\{\mathbb{E}_i\left[\left(m(X_i)-m(X_j)\right)K_{ij}\right]\frac{K_{ik}}{f_X(X_i)}f_{Z\vert X}(Z_k\vert X_i)p_w(X_i)\right\},
    \end{align*}
    where the second moment of the first term can be verified to be negligible by arguments analogous to those employed in the proof of \eqref{lemmaproof.unkdecomp eqv}. This result follows from Lemma $5$ in \cite{delgado2001significance} together with the moment condition with regard as $\mathbb{E}_k(L_{ik}\vert X_i)-f_{Z\vert X}(Z_k\vert X_i)$, which can be derived using the conclusion obtained under the Lipschitz continuity of $f_{Z\vert X}(\cdot)$ in Lemma \ref{lemma.condition dist}. The degeneracy of the mean follows by an argument similar to that used in establishing the result for $U_{ni}^{(1)}(\varphi^{(3)}_{wb})$. Furthermore, the second term can be addressed analogously to the first, leading to
    \begin{align}\label{lemmaproof.Unphi3 pk}
        \sup_w\left\vert U_{nk}^{(1)}(\varphi^{(3)}_{wb})\right\vert = o_p\left(n^{-1/2}\right).
    \end{align}
    Combining \eqref{lemmaproof.Unphi3 decomp}, \eqref{lemmaproof.Unphi3 pi}, \eqref{lemmaproof.Unphi3 pj} and \eqref{lemmaproof.Unphi3 pk}, the proof of \eqref{lemma.Unphi3 null} completes.
\end{proof}

\begin{lemma}\label{lemma.Unpsi3}
    Suppose that Assumptions \ref{ass.sample}--\ref{ass.bandwidth} hold,
    \begin{align}\label{lemma.Unpsi3 null}
        \sup_{w}\left\vert n^{-1/2}U_n^{(2)}(\psi^{(3)}_{wb})\right\vert = o_p\left(1\right)
    \end{align}
    regardless of the underlying hypothesis.
\end{lemma}

\begin{proof}[Proof of Lemma \ref{lemma.Unpsi3}]
    The argument presented in the proof of Lemma \ref{lemma.Unpsi2} can be applied here in an analogous manner.
\end{proof}

\begin{lemma}\label{lemma.G}
    The class of functions $\{G(z;X)/f_X(X):z\in\mathbb{R}^p\}$ indexed by $z$ is VC-type with common square $\mathbb{P}$-integrable envelope and VC-characteristics independent of $X$.
\end{lemma}

\begin{proof}[Proof of Lemma \ref{lemma.G}]
    Consider the measurable space $(\mathbb{R}^q\times\mathbb{R}^p,\mathcal{A}\times\mathcal{B})$ on which $(X,Z)$ is defined. By the existence of a regular conditional distribution $P_{Z\vert X=x}$, we can write $G(z;x)/f_X(x) = F_{Z\vert X=x}(z)$, where, for each fixed $x$ the function $z\mapsto F_{Z \vert X = x}(z)$ is a right-continuous, non-decreasing in every coordinate. The class $\mathcal{M}=\{G(z;X)/f_X(X):z\in\mathbb{R}^p\}$ thus consists of measurable functions indexed by the $p$-dimensional parameter $z$. 
    
    To establish the VC-subgraph property, consider the subgraph class associated with $\mathcal{M}$, 
    \begin{align*}
        \text{Subgraph}(\mathcal{M}) = \left\{(x,t)\in\mathbb{R}^q\times\mathbb{R}:t<F_{Z \vert X = x}(z),z\in\mathbb{R}^p\right\}.
    \end{align*}
    Fix any finite collection of points $(x_i,t_i)_{i=1}^n\subset\mathbb{R}^q\times\mathbb{R}$. For each $z\in\mathbb{R}^p$, define the labeling vector $l(z)=(l_1(z),\cdots,l_n(z))$, where $l_i(z)=1(t_i<F_{Z \vert X = x_i}(z))$. Since $z\mapsto F_{Z \vert X = x}(z)$ is non-decreasing in every coordinate of $z$ for each $i$, the set $\{z\in\mathbb{R}^p:t_i<F_{Z \vert X = x_i}(z)\}$ is an upper orthant of the form $(c_i,+\infty)^p$, where $c_i=(c_{i1},\cdots, c_{ip})$ denotes the coordinate-wise threshold vector defined by $c_i=\inf\{z\in\mathbb{R}^p:t_i<F_{Z \vert X = x_i}(z)\}$ interpreted in the sense of coordinatewise order. Consequently, the labeling pattern $l(z)$ over the index $\{1,\cdots,n\}$ is determined by whether each point $c_i$ lies below or above the current $z$ in the coordinatewise partial order. Hence, the label configurations induced by $\mathcal{M}$ are identical to those induced by the collection of all upper orthants $\{(c,+\infty)^p:c\in\mathbb{R}^p\}$, which is well known to have VC dimension $p$ (see, e.g., \cite{van1996weak}, Example $2.6.1$). Indeed, one can realize all $2^p$ labelings on $p$ coordinate points by suitable orthants, but no more. Therefore, the subgraph class Subgraph($\mathcal{M}$) has VC index equal to $p$, and $\mathcal{M}$ is a VC-subgraph class with index $p$. Because the envelope of $\mathcal{M}$ is bounded by one, it follows from the general entropy bound for VC-subgraph classes (\cite{van1996weak}, Theorem $2.6.7$) that there exist constants $A,v>0$ depending only on $p$ such that, for every probability measure $\mathbb{Q}$ and all $0<\epsilon\leq 1$, 
    \begin{align*}
        N(\epsilon\Vert F\Vert_{Q,2},\mathcal{M},L_2(\mathbb{Q}))\leq (A/\epsilon)^v,
    \end{align*}
    where $F\equiv1$ is the envelope function. Hence, $\mathcal{M}$ is a VC-type class whose VC-characteristics are independent of $X$.
\end{proof}

\begin{lemma}\label{lemma.delta}
    Suppose that Assumption \ref{ass.sample}--\ref{ass.bandwidth} holds,
    \begin{align*}
        \sup_{x}\left\vert \frac{1}{f_X^2(x)}\left[\hat{\Delta}_n(x)-\Delta(x)\right]\right\vert = O_p\left(\left(\frac{\log n}{na^qb^p}\right)^{\frac{1}{2}}+a^{\lambda_1}+b^{\lambda_2}\right).
    \end{align*}
\end{lemma}

\begin{proof}[Proof of Lemma \ref{lemma.delta}]
    The result concerning the uniform convergence of kernel density estimators was established in \cite{gine2002rates} and further extended in \cite{jiang2017uniform} to the case with heterogeneous bandwidths, and, combined with the conclusion on the bias term in Lemma \ref{lemma.condition dist}, which together provide the basis for establishing
    \begin{align*}
        \sup_{w}\left\vert\frac{\hat{f}_W(x,z)}{f_X(x)}-f_{Z\vert X}(x,z)\right\vert = O_p\left(\left(\frac{\log n}{na^qb^p}\right)^{\frac{1}{2}}+a^{\lambda_1}+b^{\lambda_2}\right)
    \end{align*}
    when $\int K^2(x)L^2(z)\,dx\,dz<\infty$ and $\sup_w\vert f_{Z\vert X}(x,z)/f_X(x)\vert<\infty$ are imposed in Assumption \ref{ass.structual}.
    Thus, the upper bound of $(\hat{\Delta}_n(x)-\Delta(x))/f^2_X(x)$ can be controlled by the following expression
    \begin{align*}
        &\sup_{x}\left\vert \frac{1}{f_X^2(x)}\int\left(\hat{f}^2_W(x,z)-f^2_W(x,z)\right)\,dz\right\vert\notag\\
        \leq& \sup_{x}\int\left\vert\frac{\hat{f}_W(x,z)}{f_X(x)}-f_{Z\vert X}(x,z)\right\vert^2\,dz+2\sup_{x}\int\left\vert\frac{\hat{f}_W(x,z)}{f_X(x)}-f_{Z\vert X}(x,z)\right\vert\,dz\notag\\
        =& O_p\left(\left(\frac{\log n}{na^qb^p}\right)^{\frac{1}{2}}+a^{\lambda_1}+b^{\lambda_2}\right).
    \end{align*}
    The above result implies that the integrated functional inherits the same uniform convergence rate as the pointwise kernel density estimator. Specifically, the integration over $z$ does not affect the stochastic order, because the kernel function satisfies the integrability condition in Assumption \ref{ass.structual}, which ensures that
    \begin{align*}
        \int\left\vert\frac{\hat{f}_W(x,z)}{f_X(x)}-f_{Z\vert X}(x,z)\right\vert\,dz
    \end{align*}
    is finite and uniformly bounded in $x$. To see this, note that both $\hat{f}_W(\cdot)$ and $f_W(\cdot)$ are uniformly bounded under the standard kernel regularity conditions, and the kernel’s compact support (or exponential decay in the unbounded case) guarantees that the integral of the absolute deviation over $z$ is dominated by a constant multiple of its supremum norm. Consequently, the desired conclusion follows.
\end{proof}

\begin{lemma}\label{lemma.condition dist}
    Suppose that Assumption \ref{ass.structual} holds, for $\alpha>0$,
    \begin{align*}
        \mathbb{E}\left\vert \frac{d_W(W_i)}{f_X(X_i)}-f_{Z\vert X}(W_i)\right\vert^\alpha = O\left(\max\left\{a^{\alpha\lambda_1},b^{\alpha\lambda_2}\right\}\right).
    \end{align*}
    and 
    \begin{align*}
        \mathbb{E}\left\vert \frac{d^{(2)}_W(W_i)}{f^2_X(X_i)}-f_{Z\vert X}(W_i)\right\vert^\alpha = O\left(\max\left\{a^{\alpha\lambda_1},b^{\alpha\lambda_2}\right\}\right).
    \end{align*}
\end{lemma}

\begin{proof}[Proof of Lemma \ref{lemma.condition dist}]
    Noting that we can decompose
    \begin{align*}
        &\frac{1}{f_X(x)}\mathbb{E}\left[K_a\left(\frac{x-X}{a}\right)L_b\left(\frac{z-Z}{b}\right)\right]-f_{Z\vert X}(x,z)\\
        =&\frac{1}{f_X(x)}\mathbb{E}\left\{K_a\left(\frac{x-X}{a}\right)\left\{\mathbb{E}\left[L_b\left(\frac{z-Z}{b}\right)\vert X\right]-f_{Z\vert X}(X,z)\right\}\right\}\\
        &+\frac{1}{f_X(x)}\mathbb{E}\left[K_a\left(\frac{x-X}{a}\right)f_{Z\vert X}(X,z)\right]-f_{Z\vert X}(x,z),
    \end{align*}
    it follows that, by applying Minkowski’s inequality, the desired result can be obtained by proving the desired
    \begin{align}\label{lemmaproof.condition dist first term}
        \mathbb{E}\left\vert \frac{1}{f_X(X_1)}\mathbb{E}\left\{K_a\left(\frac{X_1-X_2}{a}\right)\left\{\mathbb{E}\left[L_b\left(\frac{z-Z}{b}\right)\vert X_2\right]-f_{Z\vert X}(X_2,z)\right\}\vert X_1\right\}\right\vert^\alpha
    \end{align}
    and 
    \begin{align}\label{lemmaproof.condition dist second term}
        \mathbb{E}\left\vert \frac{1}{f_X(X_1)}\mathbb{E}\left[K_a\left(\frac{X_1-X_2}{a}\right)f_{Z\vert X}(X_2,z)\vert X_1\right]-f_{Z\vert X}(X_1,z)\right\vert^\alpha
    \end{align}
    are of the required order.
    Following an argument analogous to Lemma $5$ in \cite{delgado2001significance} and invoking the Lipschitz continuity of $f_{Z\vert X}(x,z)$ with respect to $x$, which is imposed in Assumption \ref{ass.structual}, we first obtain the convergence of the second term \eqref{lemmaproof.condition dist second term}. Next, it remains to establish the Lipschitz continuity of
    \begin{align*}
        g_L(x,z) = \mathbb{E}\left[L_b\left(\frac{z-Z}{b}\right)\vert X_2=x\right] = \int L(u)f_{Z\vert X}(x,z-bu)\,du
    \end{align*}
    with respect to $x$. Specifically, once this property is obtained, we can apply all the arguments used in the proof of \eqref{lemmaproof.condition dist second term} to the function $g_L(X_2,z)$. Together with the Lipschitz continuity of $f_{Z\vert X}(x,z)$ and the use of Lemma $5$ of \cite{delgado2001significance}, the desired result in \eqref{lemmaproof.condition dist first term} holds. From the Lipschitz continuity of $f_{Z\vert X}(\cdot)$ and the properties of kernel function imposed in Assumption \ref{ass.structual}, we have
    \begin{align*}
        \left\vert f_{Z\vert X}(x,z_1-bu)-f_{Z\vert X}(x,z_2-bu)-Q(z_1,z_2)\right\vert\leq d_{f_{Z\vert X}}(z_1-bu)\left\vert z_1-z_2\right\vert^{\lambda_2}.
    \end{align*}
    Hence,
    \begin{align*}
        \left\vert g_L(x,z_1-bu)-g_L(x,z_2-bu)-Q(z_1,z_2)\right\vert\leq \left\vert z_1-z_2\right\vert^{\lambda_2}\int d_{f_{Z\vert X}}(z_1-bu)\,du,
    \end{align*}
    and the integrability of $d_{f_{Z\vert X}}(\cdot)$ implies the desired Lipschitz continuity of $g_L(\cdot)$. An analogous reasoning can be applied to prove the conclusion concerning $d^{(2)}_W(\cdot)$.
\end{proof}

\begin{lemma}\label{lemma.bound double kernel}
    Suppose that Assumption \ref{ass.structual} holds, 
    \begin{align*}
        \sup_w\mathbb{E}\left\vert \mathbb{E}\left[\Psi(x,Z_i)L_b\left(\frac{z-Z_i}{b}\right)\right]K_{a}\left(\frac{x-X_k}{a}\right)\right\vert<\infty
    \end{align*}
\end{lemma}

\begin{proof}[Proof of Lemma \ref{lemma.bound double kernel}]
    Since the original expression is equivalent to proving
    \begin{align*}
        \sup_z\left\vert \mathbb{E}\left[\Psi(x,Z_i)L_b\left(\frac{z-Z_i}{b}\right)\right]\right\vert\sup_x\mathbb{E}\left\vert K_{a}\left(\frac{x-X_k}{a}\right)\right\vert<\infty,
    \end{align*}
    which follows by applying Lemmas 2 and 3 of \cite{delgado2001significance} to the terms concerning $K(\cdot)$ and $L(\cdot)$, respectively, provided that $\Psi(\cdot)$ satisfies the integrability condition.
\end{proof}

\end{document}